\documentclass[a4paper,onecolumn,unpublished, accepted=2025-12-22]{quantumarticle}
\pdfoutput=1

\usepackage[UKenglish]{babel}
\usepackage[T1]{fontenc}
\usepackage[utf8]{inputenc}
\usepackage{hyperref}   
\usepackage[doipre={DOI:}]{uri}   

\usepackage{csquotes}

\usepackage{etoolbox}
\apptocmd{\sloppy}{\hbadness 10000\relax}{}{} 

\usepackage{mathtools}
\usepackage{physics}
\usepackage{amsfonts}
\usepackage{amsmath}
\usepackage{amssymb}
\usepackage{mathrsfs}   
\usepackage{bigints}    
\usepackage{braket}
\usepackage{stmaryrd}   
\SetSymbolFont{stmry}{bold}{U}{stmry}{m}{n}   
\usepackage{extarrows} 

\usepackage{enumerate} 

\usepackage{amsthm}
\usepackage{thmtools}
\usepackage{thm-restate}

\usepackage{placeins}

\usepackage{bm}   

\usepackage{graphicx}
\usepackage{float}
\graphicspath{figures}
\usepackage{subcaption}
\usepackage{tikzfig}
\usepackage{pgfplots} 
\pgfplotsset{compat=newest, compat/show suggested version=false} 
\usepackage{tikz-cd}
\usetikzlibrary{decorations.markings}
\pgfdeclaresnake{square wave qubit}{initial}
{
  \state{initial}[width=0.5pt]
  {
    \pgfpathlineto{\pgfpoint{0pt}{0.15pt}}
    \pgfpathlineto{\pgfpoint{0.25pt}{0.15pt}}
    \pgfpathlineto{\pgfpoint{0.25pt}{-0.15pt}}
    \pgfpathlineto{\pgfpoint{0.5pt}{-0.15pt}}
    \pgfpathlineto{\pgfpoint{0.5pt}{0pt}}
  }
  \state{final}
  {
    \pgfpathlineto{\pgfpoint{\pgfsnakeremainingdistance}{0pt}}
  }
}

\pgfdeclaresnake{square wave dual rail}{initial}
{
  \state{initial}[width=1pt]
  {
    \pgfpathlineto{\pgfpoint{0pt}{0.1pt}}
    \pgfpathlineto{\pgfpoint{0.5pt}{0.1pt}}
    \pgfpathlineto{\pgfpoint{0.5pt}{-0.1pt}}
    \pgfpathlineto{\pgfpoint{1pt}{-0.1pt}}
    \pgfpathlineto{\pgfpoint{1pt}{0pt}}
  }
  \state{final}
  {
    \pgfpathlineto{\pgfpoint{\pgfsnakeremainingdistance}{0pt}}
  }
}

\tikzstyle{hbox}=[draw=black, shape=rectangle, fill=yellow, minimum size=.55em, inner sep=0.15em, scale=0.85, font={\scriptsize}]
\tikzstyle{box}=[draw=black, shape=rectangle, fill=white, minimum size=1em, inner sep=0.2em, scale=0.85, font={\scriptsize}]
\tikzstyle{gn}=[draw=black, shape=circle, fill={zx_green}, draw=black, inner sep=0.7mm, minimum width=0pt, minimum height=0pt, tikzit fill={rgb,255: red,181; green,215; blue,181}]
\tikzstyle{rn}=[gn, fill={zx_red}, draw=black, tikzit fill={rgb,255: red,215; green,96; blue,96}]
\tikzstyle{wn}=[gn, fill=white, draw=black]
\tikzstyle{gn_phase}=[shape=rectangle, fill={zx_green}, draw=black, minimum size=1.2em, rounded corners=0.48em, inner sep=0.2em, outer sep=-0.2em, scale=0.8, font={\footnotesize}, tikzit shape=circle, tikzit fill={rgb,255: red,181; green,215; blue,181}]
\tikzstyle{rn_phase}=[{gn_phase}, fill={zx_red}, draw=black, tikzit fill={rgb,255: red,215; green,96; blue,96}]
\tikzstyle{phase}=[draw=black, shape=rectangle, fill=white, minimum size=.95em, inner sep=0.1em, scale=0.85, font={\scriptsize}]
\tikzstyle{phaseCircle}=[{gn_phase}, draw=black, fill=white]
\tikzstyle{black node}=[draw=black, shape=circle, scale=0.3, fill=black, font={\footnotesize}]
\tikzstyle{rtriang}=[shape=isosceles triangle, fill={gray!50}, draw=black, isosceles triangle stretches=true, inner sep=0.8pt, minimum width=.5mm, minimum height=.5mm]
\tikzstyle{ltriang}=[rtriang, shape=isosceles triangle, fill={gray!50}, draw=black, shape border rotate=180]
\tikzstyle{beamSplitter}=[draw=black, shape=rectangle, fill=white, minimum width=3mm, minimum height=.2mm, inner sep=.3mm]
\tikzstyle{sLabel}=[font={\scriptsize}, auto]
\tikzstyle{midArrow}=[postaction=decorate, decoration={markings, mark=at position 0.5 with {\arrow{stealth[length=6.4pt, sep=-2pt]}}}]

\tikzstyle{N}=[-, line width=0.7pt]
\tikzstyle{NStream}=[-, N, midArrow]
\tikzstyle{LOWire}=[-, very thin, tikzit draw={rgb,255: red,134; green,42; blue,43}]
\tikzstyle{NLO}=[-, line width=.6pt, tikzit draw={rgb,255: red,117; green,44; blue,83}]
\tikzstyle{LOStream}=[-, NLO, midArrow, tikzit draw={rgb,255: red,139; green,43; blue,126}]
\tikzstyle{DRWire}=[-, LOWire, double, double distance=.5pt, tikzit draw={rgb,255: red,24; green,33; blue,119}]
\tikzstyle{DRThick}=[-, DRWire, NLO, double distance=.5pt, tikzit draw={rgb,255: red,60; green,108; blue,166}]
\tikzstyle{DRStream}=[-, DRThick, midArrow, tikzit draw={rgb,255: red,0; green,175; blue,175}]
\tikzstyle{dashedE}=[-, dashed]
\tikzstyle{hadamard edge}=[-, dashed, draw=blue]
\tikzstyle{braceedge}=[-, decorate, decoration={brace, amplitude=2mm, raise=-1mm}]
\tikzstyle{arrow}=[<-]
\tikzstyle{MixedDR}=[-, tikzit draw={rgb,255: red,255; green,128; blue,0}, decorate, decoration=square wave dual rail, very thin, draw={rgb,255: red,255; green,128; blue,0}]
\tikzstyle{MixedQubit}=[-, tikzit draw={rgb,255: red,160; green,100; blue,219}, decorate, decoration=square wave qubit, very thin, draw={rgb,255: red,160; green,100; blue,219}]
\tikzstyle{ClWire}=[-, very thin, densely dashed, tikzit draw={rgb,255: red,255; green,191; blue,191}]
\tikzstyle{filllayer}=[-, fill={rgb,255: red,255; green,235; blue,0}, opacity=1]
\tikzstyle{filllayerwhite}=[-, filllayer, opacity=0.3, fill=white]
\tikzstyle{FLX}=[-, filllayer, fill={rgb,255: red,255; green,143; blue,143}, opacity=0.5]
\tikzstyle{FLY}=[-, filllayer, fill={rgb,255: red,255; green,235; blue,133}, opacity=0.5]
\tikzstyle{FLZ}=[-, filllayer, fill={rgb,255: red,181; green,210; blue,255}, opacity=0.5]
\tikzstyle{FLNE}=[-, pattern={Lines[angle=45,distance=10pt,yshift=2pt,xshift=-.5pt]}, pattern color=gray]
\tikzstyle{FLNW}=[-, pattern={Lines[angle=-45,distance=10pt,yshift=5pt,xshift=1pt]}, pattern color=gray]
\tikzstyle{FLG}=[-, filllayer, fill={rgb,255: red,212; green,212; blue,212}]
\tikzstyle{FLP}=[-, filllayer, fill={rgb,255: red,255; green,243; blue,192}]

\definecolor{zx_grey}{RGB}{211,211,211}
\definecolor{zx_red}{RGB}{232,165,165}
\definecolor{zx_green}{RGB}{216,248,216}

\usepackage{xtab}
\usepackage{xspace}
\usepackage{tabu}
\usepackage{multicol}

\usepackage[capitalise, noabbrev]{cleveref}

\newcommand{\N}{\mathbb{N}}

\ifdef{\C}
  {\renewcommand{\C}{\mathbb{C}}}
  {\newcommand{\C}{\mathbb{C}}}
\newcommand{\minu}{\texttt{-}}
\newcommand{\plus}{\texttt{+}}

\newcommand{\interp}[1]{\left\llbracket#1\right\rrbracket}

\newcommand{\bR}{\begin{color}{red}}
\newcommand{\bB}{\begin{color}{blue}}
\newcommand{\bM}{\begin{color}{magenta}}
\newcommand{\bC}{\begin{color}{cyan}}
\newcommand{\bW}{\begin{color}{white}}
\newcommand{\bBl}{\begin{color}{black}}
\newcommand{\bG}{\begin{color}{green}}
\newcommand{\bY}{\begin{color}{yellow}}
\newcommand{\e}{\end{color}}

\usepackage{todonotes}
\setuptodonotes{tickmarkheight=6pt}

\newcommand{\tikzrefsize}[1]{\tiny{#1}}
\newcommand{\lemref}[1]{\tikzrefsize{\textsc{(Lem.~\ref{#1})}}}
\newcommand{\equref}[1]{\tikzrefsize{\textsc{(Eq.~\ref{#1})}}}

\newtheorem{theorem}{Theorem}[section]
\newtheorem{proposition}[theorem]{Proposition}

\newtheorem{lemma}[theorem]{Lemma}
\newtheorem{corollary}[theorem]{Corollary}
\theoremstyle{definition}
\newtheorem{definition}[theorem]{Definition}
\newtheorem{example}[theorem]{Example}
\theoremstyle{remark}
\newtheorem{remark}[theorem]{Remark}

\newcommand{\sub}{\subseteq}
\renewcommand{\tt}[1]{\mathtt{#1}}
\renewcommand{\bf}[1]{\ensuremath{\mathbf{#1}}}

\renewcommand{\cal}[1]{\mathcal{#1}}

\renewcommand{\phi}{\varphi}

\newcommand{\Odd}{\texttt{Odd}}
\newcommand{\comp}[1]{\overline{#1}}

\newcommand{\cvar}[1]{\underline{#1}}
\newcommand{\ccvar}[1]{#1}

\newcommand{\XY}{\normalfont XY\xspace}
\newcommand{\XZ}{\normalfont XZ\xspace}
\newcommand{\YZ}{\normalfont YZ\xspace}
\newcommand{\X}{\normalfont X\xspace}

\newcommand{\Z}{\normalfont Z\xspace}
\newcommand{\XYm}{\ensuremath\normalfont\textrm{XY}\xspace}
\newcommand{\XZm}{\normalfont\normalfont\textrm{XZ}\xspace}
\newcommand{\YZm}{\normalfont\normalfont\textrm{YZ}\xspace}
\newcommand{\Xm}{\ensuremath\normalfont\textrm{X}\xspace}
\newcommand{\Ym}{\normalfont\normalfont\textrm{Y}\xspace}
\newcommand{\Zm}{\normalfont\normalfont\textrm{Z}\xspace}

\newcommand{\Fusion}{\tiny{(\textsc{Spider})}}
\newcommand{\Bigebra}{\tiny{(\textsc{Bigebra})}}
\newcommand{\ZElim}{\tiny{(\textsc{Z-Elim})}}
\newcommand{\XElim}{\tiny{(\textsc{X-Elim})}}

\newcommand{\Colour}{\tiny{(\textsc{Color})}}
\newcommand{\Copy}{\tiny{(\textsc{Copy})}}
\newcommand{\One}{\tiny{(\textsc{One})}}
\newcommand{\Euler}{\tiny{(\textsc{Euler})}}
\newcommand{\PiCommute}{\tiny{($\pi$)}}

\graphicspath{figures}
\allowdisplaybreaks
\usepackage[numbers,sort&compress]{natbib}
\usepackage{soul}
\usetikzlibrary{patterns.meta}

\begin{document}
\title{A dataflow programming framework for linear optical distributed quantum computing}

\author[1]{Giovanni de Felice}
\thanks{\href{mailto:giodefelice@protonmail.com}{giodefelice@protonmail.com}}
\author[1,2]{Boldizsár Poór}
\author[3]{Cole Comfort}
\thanks{this work has been partially funded by the European Union through the MSCA SE project QCOMICAL and within the framework of ``Plan France
2030'', under the research projects EPIQ ANR-22-PETQ-0007 and HQI-R\&D ANR-22-PNCQ-0002.}
\author[4]{Lia Yeh}
\thanks{this work was done while at Quantinuum and University of Oxford.}
\author[1]{Mateusz Kupper}
\author[2]{William Cashman}
\author[5, 6]{Bob Coecke}
\affil[1]{Quantinuum, 17 Beaumont Street, Oxford, OX1 2NA, United Kingdom}
\affil[2]{University of Oxford, Department of Computer Science, Oxford, OX1 3QD, United Kingdom}
\affil[3]{Universit\'e Paris-Saclay, CNRS, ENS Paris-Saclay, Inria, Laboratoire M\'ethodes Formelles, 91190, Gif-sur-Yvette, France}
\affil[4]{University of Cambridge, Department of Computer Science and Technology, Cambridge, CB3 0FD, United Kingdom}
\affil[5]{Wolfson college, Linton Road, Oxford, OX2 6UD, United Kingdom}
\affil[6]{Perimeter Institute for Theoretical Physics, Waterloo, ON N2L 2Y5, Canada}

\maketitle

\begin{abstract}
    Photonic systems offer a promising platform for interconnecting quantum processors and enabling scalable, networked architectures. 
Designing and verifying such architectures requires a unified formalism that integrates linear algebraic reasoning with probabilistic and control-flow structures. 
In this work, we introduce a graphical framework for distributed quantum computing that brings together linear optics, the ZX-calculus, and dataflow programming. 
Our language supports the formal analysis and optimization of distributed protocols involving both qubits and photonic modes, 
with explicit interfaces for classical control and feedforward, all expressed within a synchronous dataflow model with discrete-time dynamics. 
Within this setting, we classify entangling photonic fusion measurements, show how their induced Pauli errors can be corrected via a novel flow structure for fusion networks, 
and establish correctness proofs for new repeat-until-success protocols enabling arbitrary fusions.
Layer by layer, we construct qubit architectures incorporating practical optical components such as beam splitters, switches, and photon sources, 
with graphical proofs that they are deterministic and support universal quantum computation.
Together, these results establish a foundation for verifiable compilation and automated optimization in networked quantum computing.

\end{abstract}
\tableofcontents

\section{Introduction}\label{sec:introduction}

Photonic systems are a leading platform for interconnecting modular quantum hardware into distributed networks. 
Recent proposals for scalable quantum computing exploit interfaces between matter qubits and discrete-variable photonic modes, 
including distributed hybrid architectures based on ions~\cite{jiang_distributed_2007, monroe_large_2014}, 
neutral atoms~\cite{covey_quantum_2023,sunami_scalable_2025}, or spins~\cite{degliniastySpinOpticalQuantumComputing2024, coganDeterministicGenerationIndistinguishable2023}, 
or drive the entire computation by measurements on photonic qubits, as in fusion-based quantum computing~\cite{bartolucciFusionbasedQuantumComputation2023}. 
These models point toward the possibility of quantum networks capable of executing large-scale distributed quantum algorithms.

Current quantum programming languages lack many necessary capabilities in order to natively capture non-qubit systems such as photons. 
The majority are designed solely for circuit-based models, where computation derives from a fixed set of unitary gates. 
Distributed quantum computation, however, demands a new paradigm that hybridizes circuit-based and 
measurement-based approaches while accommodating heterogeneous quantum systems. 
Software for such networks must therefore represent diverse physical systems and their interactions while supporting the computational abstractions needed for verification, optimization, and compilation.
Meeting these requirements calls for a formalism that links physical notions of scalability with computer science concepts such as semantics, abstraction, control flow, concurrency, and rewriting.

The ZX calculus~\cite{coeckeInteractingQuantumObservables2008} has already proved effective across
circuit-based~\cite{coeckePicturingQuantumProcesses2017, kissingerReducingNumberNonClifford2020},
measurement-based~\cite{duncanRewritingMeasurementBasedQuantum2010, backensThereBackAgain2021, mcelvanneyFlowpreservingZXcalculusRewrite2023},
and fault-tolerant~\cite{beaudrapZXCalculusLanguage2020, kissingerPhasefreeZXDiagrams2022, huangGraphicalCSSCode2023, townsend-teagueFloquetifyingColourCode2023,rodatzFloquetifyingStabiliser2024,rodatzFaultTolerance2025} 
models of quantum computation.
In the context of photonic quantum computing, the ZX calculus is starting to be applied in a top-down direction, e.g. to design novel error correcting codes~\cite{bombinUnifyingFlavorsFault2023, pankovichFlexibleEntangledState2023}.
However, the bottom-up direction, from the hardware described in the language of linear optics~\cite{aaronsonComputationalComplexityLinear2011} to the logic of ZX calculi, has not been established:
current graphical languages for optics~\cite{clementLOvCalculusGraphicalLanguage2022,defeliceQuantumLinearOptics2023,defeliceLightMatterInteractionZXW2023,heurtelCompleteGraphicalLanguage2024}
are restricted to \emph{passive} optical setups and lack formal notions of measurement and feedforward which are crucial to performing entangling gates in photonic architectures.

A second critical gap is the absence of formalisms for the \emph{dynamic} nature of optical protocols. 
Current approaches rely on static circuit diagrams, which obscure a central advantage of photonics: the ability to exploit both spatial and temporal degrees of freedom~\cite{motesScalableBoson2014, madsenQuantumComputational2022}.
Optical circuits are inherently time-dependent, acting on streams of photons with components such as routers, delays, switches, and time-controlled emitters.
A suitable formalism should therefore have the expressivity to specify what each component does at each time step.
Extending recent graphical formalisations of discrete-time dynamics~\cite{caretteGraphicalLanguageDelayed2021, dilavoreMonoidalStreamsDataflow2022} 
to real-world hardware components would enable recursive reasoning and rewriting across a wide variety of experimental setups.

In this work, we develop the mathematical foundations to enable:
(i) verification and comparison of proposals for networked quantum computing,
(ii) translation of techniques from circuit-based models into the photonic, measurement-driven setting, and
(iii) compilation of distributed algorithms across heterogeneous hardware platforms.

A number of compilation frameworks and algorithms have recently been developed for photonic one-way computation~\cite{zilkCompilerUniversalPhotonic2022,zhangOneQCompilationFramework2023,leeGraphtheoreticalOptimizationFusionbased2023} 
and for distributed architectures more broadly~\cite{ferrari_compiler_2021,andres-martinez_distributing_2024}. 
These approaches are typically algorithmic or architecture-specific, focusing on resource optimization and circuit translation. 
Our approach is complementary: by grounding compilation in formal languages and rewriting, 
we provide a unifying framework in which such strategies can be expressed, analyzed, and verified. 
In doing so, we connect hardware-level models such as linear optics with high-level programming concepts such as control flow and recursion, 
enabling systematic reasoning and the development of new design principles for networked quantum architectures.



\subsection{Graphical framework}

Our main contribution is a graphical language $\mathtt{Stream}(\mathtt{Channel}(\mathbf{ZXLO}))$ 
for reasoning about quantum networks with discrete-time dynamics.
Its construction proceeds in three stages:
\begin{enumerate}
    \item \textbf{Base layer.} The language $\mathbf{ZXLO}$ (introduced in \cref{sec:base}) combines the ZX calculus with linear optics, incorporating a node for dual-rail encoding.
    \item \textbf{Channel layer.} The $\mathtt{Channel}$ construction (\cref{sec:channel}) generalizes classical–quantum maps~\cite{coeckePicturingQuantumProcesses2017} to discrete-variable linear optics. 
    It introduces a syntax based on Kraus maps, formalizing previously ad hoc variables in diagrams~\cite{duncanGraphicalApproachMeasurementbased2013, backensThereBackAgain2021} and extending them with classical annotations to model feedforward and coarse-graining of measurement outcomes.
    \item \textbf{Stream layer.} The $\mathtt{Stream}$ construction (\cref{sec:stream}) builds on intensional monoidal streams~\cite{dilavoreMonoidalStreamsDataflow2022} and the notion of observational equivalence from~\cite{caretteGraphicalLanguageDelayed2021}, 
    enriching our syntax with time-indexed diagrams and feedback loops. It further extends previous work by introducing novel rewriting rules and a recursive $\mathtt{unroll}$ operation, which enables inductive reasoning about protocol executions.
\end{enumerate}

This language yields formal definitions of determinism and universality for hybrid qubit-photonic channels,
and supports the representation of practical optical modules such as routers, delays, and emitters. 
Its structure reflects three core features of networked architectures: 
(i) a qubit-photon interface that allows encoding units of quantum information in photonic modes;
(ii) measurement and classical feedforward in the style of measurement-based quantum computing (MBQC);
(iii) protocols with discrete-time dynamics involving both quantum and classical memory.
The graphical framework developed in this paper forms the basis of the Python package Optyx \cite{kupper_optyx_2025}, implemented within the DisCoPy ecosystem \cite{felice_discopy_2021, toumi_discopy_2022}.

\subsection{Applications}

\paragraph{Characterization of fusion measurements}

As a first application of our framework, in \cref{sec:characterization}, we derive formal representations of fusion measurements in the ZX calculus,
and we characterize all locally equivalent measurements whose Pauli byproducts can be corrected.
We classify as \enquote{green failure} all such measurements which preserve entanglement of any graph state under the fusion failure outcome.
We then narrow down which of these induce correctable Pauli errors in both success and failure outcomes; see \cref{prop:characterization}.
This class includes phase gadgets, an important family of non-Clifford entangling gates~\cite{kissingerReducingNumberNonClifford2020}.
By additionally restricting to stabilizer measurements, the two most commonly used fusion measurements
--- the Type II fusion of~\cite{browneResourceEfficientLinearOptical2005,bartolucciFusionbasedQuantumComputation2023} and the CZ fusion of~\cite{limRepeatUntilSuccessLinearOptics2005,degliniastySpinOpticalQuantumComputing2024} ---
naturally arise from our characterization as $\Xm$ and $\Ym$ fusions, respectively.

\paragraph{Proofs of determinism}

In MBQC, computation proceeds by measuring a graph state resource, with undesired outcomes corresponding to Pauli errors. 
Determinism is guaranteed when these errors can be corrected by classical control according to \emph{flow structure}~\cite{browneGeneralizedFlowDeterminism2007} on the underlying graph. 
Current formal models, however, do not handle two-qubit measurements such as fusions.

In \cref{sec:flow}, we introduce a notion of \emph{partially static flow} for fusion networks which enables correction of Pauli errors 
while avoiding active classical control on fusion nodes.
We prove that these fusion networks can be implemented by a deterministic pattern where fusion operations precede single-qubit measurements (\cref{thm:flow-pattern}), and that
any decomposition of an open graph as a fusion network of $\Xm$ and $\Ym$ fusions admits partially static flow provided the original graph has Pauli flow (\cref{thm:flow-simplified-graph}).
This yields the first formal analysis of photonic fusion errors, establishing flow conditions for determinism in fusion-based MBQC.

\paragraph{Proofs of universality}

Universality —-- the ability of a quantum protocol to simulate any other —-- is particularly challenging in linear optics, where entanglement can only be generated probabilistically. 
Our framework expresses entire photonic architectures as single diagrams, enabling their properties and expressivity to be investigated by inductive, graphical reasoning.
We give a powerful example of this form of reasoning in \cref{thm:honeycomb}, showing that the honeycomb lattice with measurements restricted to the $\YZm$-plane is universal for quantum computing.

In \cref{subsec:RUS}, we prove the correctness of new repeat-until-success (RUS) protocols that boost the success probability of fusion measurements with green failure, given access to entangled resource states. 
As special cases, this recovers the RUS optical CZ gates of~\cite{limRepeatUntilSuccessLinearOptics2005} and the boosted Bell measurement of~\cite{lee_nearly_2015}. The proof uses the unrolling technique of \cref{sec:stream} and, to our knowledge, constitutes the first entirely graphical proof by induction.
We end \cref{sec:universality} by comparing two approaches to universal photonic quantum computing: lattice-based and reprogrammable emitter-based architectures. 
We provide graphical proofs of determinism and universality for minimal examples of each, allowing us to appreciate the scaling of these approaches.
Our constructive proof for emitter-based architectures demonstrates how any given MBQC pattern can be compiled into a sequence of instructions for a concrete optical setup, without resorting to universal graph states.
This indicates an alternative path to near-term photonic quantum computing where programmable setups are enhanced and optimized by compilation techniques.

\part{Graphical framework for quantum protocols}\label{part1}

\section{Base language}\label{sec:base}

We introduce the graphical notations used in this paper, allowing us to relate
linear optical circuits and ZX-diagrams via the dual-rail encoding.
This fixes our base language of linear maps $\mathbf{ZXLO}$ generated by ZX diagrams, linear optical circuits and a triangle node introduced at the end of the section.
The graphical framework developed in subsequent sections is however independent of this choice of linear maps, 
and thus the reader may modify the base language as long as the new generators have an interpretation as bounded linear maps.

\subsection{ZX calculus}\label{sec:ZX}

The ZX calculus is a graphical language used to represent and reason about qubit quantum computation.
Its elementary building blocks are the green \emph{Z-spider} and the red \emph{X-spider} (therefore ZX).
\begin{align}
  \label{eq:z-spider-interp}
  \input{ZX/generators/z-spider.tikz}
  \quad &\overset{\interp{\cdot}}{\longmapsto} \quad
  \ket{0}^{\otimes n}\! \bra{0}^{\otimes m} + e^{i \alpha} \ket{1}^{\otimes n}\! \bra{1}^{\otimes m}\\
  \label{eq:x-spider-interp}
  \input{ZX/generators/x-spider.tikz}
  \quad &\overset{\interp{\cdot}}{\longmapsto} \quad
  \ket{+}^{\otimes n}\! \bra{+}^{\otimes m} + e^{i \alpha} \ket{-}^{\otimes n}\! \bra{-}^{\otimes m}
\end{align}
These spiders can have any number of input and output legs, corresponding to qubit ports, and they have a phase parameter $\alpha$.
Notice that since the $e^{i \alpha}$ function in the interpretation is $2 \pi$ periodic, we can take the parameter of spiders modulo $2 \pi$.
The last generator of the ZX calculus is the yellow \emph{Hadamard box} $\begin{tikzpicture}
	\begin{pgfonlayer}{nodelayer}
		\node [style=none] (0) at (-1, 0) {};
		\node [style=none] (1) at (1, 0) {};
		\node [style=hbox] (2) at (0, 0) {};
	\end{pgfonlayer}
	\begin{pgfonlayer}{edgelayer}
		\draw (0.center) to (2);
		\draw (2) to (1.center);
	\end{pgfonlayer}
\end{tikzpicture}
$ that corresponds to the Hadamard gate, $\ket{0}\bra{+}\ +\ \ket{1}\bra{-}$.
We define the star symbol as the diagram corresponding to $\frac{1}{\sqrt 2}$, that is, $\star \coloneqq \begin{tikzpicture}
	\begin{pgfonlayer}{nodelayer}
		\node [style=gn] (0) at (-0.625, 0) {};
		\node [style=rn] (1) at (0.625, 0) {};
	\end{pgfonlayer}
	\begin{pgfonlayer}{edgelayer}
		\draw (1) to (0);
		\draw [bend left=45, looseness=0.75] (0) to (1);
		\draw [bend left=45, looseness=0.75] (1) to (0);
	\end{pgfonlayer}
\end{tikzpicture}
$.
Using these building blocks, we are able to intuitively represent common quantum gates as well as any unitary.
\begin{figure}[ht]
  \noindent
  \begin{minipage}{.28\textwidth}
    \begin{align*}
      \input{ZX/elements/x-state-zero.tikz}
      \quad &\overset{\interp{\cdot}}{\longmapsto} \quad
      \ket{0} \\
      \input{ZX/elements/x-state-pi.tikz}
      \quad &\overset{\interp{\cdot}}{\longmapsto} \quad
      \ket{1} \\
      \input{ZX/elements/z-state-zero.tikz}
      \quad &\overset{\interp{\cdot}}{\longmapsto} \quad
      \ket{+} \\
      \input{ZX/elements/z-state-pi.tikz}
      \quad &\overset{\interp{\cdot}}{\longmapsto} \quad
      \ket{-}
    \end{align*}
  \end{minipage}
  \begin{minipage}{.28\textwidth}
    \begin{align*}
      \input{ZX/elements/not.tikz}
      \quad &\overset{\interp{\cdot}}{\longmapsto} \quad
      \input{QuantumCircuit/not.tikz} \\
      \input{ZX/elements/z-rotate-alpha.tikz}
      \quad &\overset{\interp{\cdot}}{\longmapsto} \quad
      \input{QuantumCircuit/z-rotate-alpha.tikz} \\
      \input{ZX/elements/t.tikz}
      \quad &\overset{\interp{\cdot}}{\longmapsto} \quad
      \input{QuantumCircuit/t.tikz} \\
      \input{ZX/elements/s.tikz}
      \quad &\overset{\interp{\cdot}}{\longmapsto} \quad
      \input{QuantumCircuit/s.tikz}
    \end{align*}
  \end{minipage}
  \begin{minipage}{.44\textwidth}
    \begin{align*}
      \input{ZX/elements/cnot.tikz}
      \quad &\overset{\interp{\cdot}}{\longmapsto} \quad
      \input{QuantumCircuit/cnot.tikz} \\
      \input{ZX/elements/cz.tikz}
      \quad &\overset{\interp{\cdot}}{\longmapsto} \quad
      \input{QuantumCircuit/cz.tikz} \\
      \input{ZX/elements/PhaseGadgetUnitary.tikz}
      \quad &\overset{\interp{\cdot}}{\longmapsto} \quad
      \ket{a\ b} \mapsto e^{i \alpha (a \oplus b)}\ket{a\ b}
    \end{align*}
  \end{minipage}
\end{figure}

Further to its universal representational power for unitary maps, the ZX calculus also comes equipped with a set of graphical rewrite rules, presented in \cref{fig:ZX-axioms}.
These rules provide a powerful mechanism for transforming diagrams, which can be used to establish equivalences between quantum circuits and to formally verify quantum protocols.
Crucially, this rule set is known to be complete for qubit stabilizer quantum mechanics, meaning any equality that holds between qubit maps can be derived diagrammatically~\cite{ngUniversalCompletionZXcalculus2017,jeandelDiagrammaticReasoningClifford2018}.

\begin{remark}
  \label{rem:scalar-zx}
  There are many different axiomatizations of the ZX calculus.
  The version presented in \cref{fig:ZX-axioms} maintains exact equality for the stabilizer fragment.
  Versions of the ZX calculus can be found in~\cite{jeandelCompleteAxiomatisationZXCalculus2018} for the Clifford+T fragment, and in~\cite{vilmartNearMinimalAxiomatisationZXCalculus2019} for the full language.
\end{remark}

A ZX-diagram with no inputs or outputs corresponds to some scalar, which we sometimes choose to omit from calculations.
We use $=$ when a rewrite preserves equality on the nose, and $\approx$ denotes equality \emph{up to some non-zero scalar}.


\subsection{MBQC graphs}

The ZX calculus is closely related to the measurement-based model of quantum computing.
In MBQC, computation is performed in two stages:
(i) a \emph{graph state} is prepared and
(ii) it is processed by a sequence of \emph{single-qubit measurements}.
A graph state associated with the graph $G = (V, E)$ is an entangled quantum state constructed by preparing a qubit for each vertex in the $\ket{+}$ state and applying $CZ$ gates for each edge.
We may depict graph states equivalently as ZX diagrams or qubit circuits.
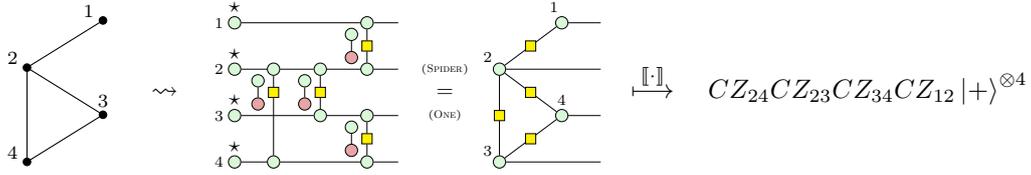
\begin{example}
  \label{example:graph-sate}
  \[
    \input{MBQC/simple-graph.tikz}
    \quad\leadsto
    \scalebox{.82}{\input{MBQC/graph-state.tikz}}
    \overset{\interp{\cdot}}{\longmapsto} \quad
    CZ_{24} CZ_{23} CZ_{34} CZ_{12} \ket{+}^{\otimes 4}
  \]
\end{example}
\noindent Qubits can be inputs or outputs, which we depict by connecting wires to the left or right boundaries of the diagram.
\begin{definition}[Open graph]
  An open graph is a tuple $(G, I, O)$, where $G = (V, E)$ is an undirected graph, and $I, O \subseteq V$ are (possibly overlapping) subsets representing inputs and outputs.
  We use the notations $\comp{O} \coloneqq V \backslash O$ for the non-output and $\comp{I} \coloneqq V \backslash I$ for the non-input vertices.
\end{definition}
\noindent During computation, every non-output vertex of the graph is measured in a certain basis specified by a measurement plane ($\lambda$) and angle ($\alpha$).
\begin{definition}[Labelled open graph]
  A labelled open graph is a tuple $\mathcal{M} = (G, I, O, \lambda, \alpha)$,
  where $(G, I, O)$ is an open graph,
  $\lambda: \comp{O} \to \set{\XYm, \XZm, \YZm}$ is an assignment of measurement planes,
  and $\alpha: \comp{O} \to [0, 2\pi)$ assigns measurement angles to each non-output qubit.
\end{definition}
\noindent We use the following notation to denote pure single qubit effects (the corresponding states are defined analogously).
\[
\bra{\pm_{\lambda,\alpha}} = \begin{dcases}
\frac{1}{\sqrt{2}}(\bra{0} \pm e^{i \alpha}\bra{1}) & \text{if } \lambda = \XYm \hspace{10mm} \begin{tikzpicture}
	\begin{pgfonlayer}{nodelayer}
		\node [style={gn_phase}] (0) at (1.075, 0) {$\alpha \pm \pi$};
		\node [style=none] (1) at (0, 0) {};
		\node [style=none] (2) at (2.25, 0) {$\star$};
	\end{pgfonlayer}
	\begin{pgfonlayer}{edgelayer}
		\draw (1.center) to (0);
	\end{pgfonlayer}
\end{tikzpicture}
 \\
\frac{1}{\sqrt{2}}(\bra{+} \pm e^{i \alpha}\bra{-}) & \text{if } \lambda = \YZm \hspace{10mm} \begin{tikzpicture}
	\begin{pgfonlayer}{nodelayer}
		\node [style={rn_phase}] (0) at (1.075, 0) {$\alpha \pm \pi$};
		\node [style=none] (1) at (0, 0) {};
		\node [style=none] (2) at (2.25, 0) {$\star$};
	\end{pgfonlayer}
	\begin{pgfonlayer}{edgelayer}
		\draw (1.center) to (0);
	\end{pgfonlayer}
\end{tikzpicture}
 \\
\frac{1}{\sqrt{2}}(\bra{i} \pm e^{i \alpha}\bra{-i}) & \text{if } \lambda = \XZm  \hspace{10mm} \begin{tikzpicture}
	\begin{pgfonlayer}{nodelayer}
		\node [style={gn_phase}] (0) at (1, 0) {$\frac{\pi}{2}$};
		\node [style=none] (1) at (0, 0) {};
		\node [style={rn_phase}] (2) at (2.325, 0) {$\alpha \pm \pi$};
		\node [style=none] (3) at (3.5, 0) {$\star$};
	\end{pgfonlayer}
	\begin{pgfonlayer}{edgelayer}
		\draw (1.center) to (0);
		\draw (2) to (0);
	\end{pgfonlayer}
\end{tikzpicture}

\end{dcases}
\]
\noindent Any labelled open graph defines a target linear map which corresponds to the quantum computation that we want to execute.
To ensure such a map is well defined, we provide a measurement pattern which contains a concrete sequence of instructions to generate the graph.

\begin{definition}\label{def:ogs-to-linear-map}
 Suppose $\mathcal{M} = (G, I, O, \lambda, \alpha)$ is a labelled open graph.
 The \emph{target linear map of $\mathcal{M}$} is given by
 \[
  T(\mathcal{M}) \coloneqq \left( \prod_{i\in\comp{O}} \bra{+_{\lambda(i),\alpha(i)}}_i \right) E_G N_{\comp{I}},
 \]
  where $E_G \coloneqq \prod_{i\sim j} CZ_{ij}$ and $N_{\comp{I}} \coloneqq \prod_{i\in\comp{I}} \ket{\plus}_i$.
\end{definition}

As shown in \cite{backensThereBackAgain2021}, a labelled open graph is completely characterized by a ZX diagram in MBQC form.

\begin{definition}\cite{backensThereBackAgain2021}
    A graph state diagram is a ZX diagram where all vertices are green, all the connections between vertices are Hadamard edges and a single output wire is incident on each vertex in the diagram.
    A ZX diagram is in MBQC form if it consists of a graph state diagram in which each vertex may additionally be connected to:
    \begin{itemize}
        \item an input (in addition to its output), and
        \item a measurement effect in one of the measurement planes $\{\XYm, \XZm, \YZm\}$ (instead of the output).
    \end{itemize}
    Given a ZX diagram $D$ in MBQC form, its \emph{underlying graph} $G(D)$ is the labelled open graph $(G, I, O, \lambda, \alpha)$
    where $G$ is the graph whose vertices are the green spiders of $D$ and whose edges are the Hadamard edges of $D$, the inputs (outputs) 
    $I, O \sub G$ are the spiders connected to the inputs (outputs) of $D$, and the pair $\lambda, \alpha$ are the measurement effects on each spider.
\end{definition}

We introduce a scalable notation for MBQC form ZX diagrams, which are completely characterized by an open graph $(G, I, O)$ and a choice of measurement 
planes and angles  $\lambda, \alpha$ on non-output qubits:
\[
  \input{MBQC/MBQC-form.tikz}
\]

\begin{proposition}\cite{backensThereBackAgain2021}
    For any labelled open graph $\mathcal{M} = (G, I, O, \lambda, \alpha)$ there is a ZX diagram $D$ in MBQC form with $G(D) = \mathcal{M}$ 
    and whose interpretation is the target linear map of $\mathcal{M}$, i.e. $T(\mathcal{M}) = \interp{D}$.
\end{proposition}

We can thus represent $T(\mathcal{M})$ in the ZX calculus by attaching the appropriate effects to the dangling qubits in the graph state.

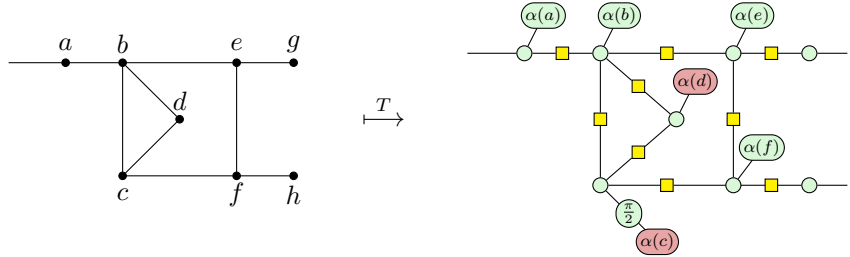
\begin{example}
  \[
    \input{figures/mbqc-example-pattern-graph.tikz}
    \qquad\overset{T}{\longmapsto}\qquad
    \input{figures/mbqc-example-pattern.tikz}
  \]
  where the input set is $I = \set{a}$, the set of outputs is $O = \set{g, h}$, and the measurement planes are $\lambda(v) = \XYm$ for all $v \in \set{a, b, e, f}$, $\lambda(d) = \YZm$, and $\lambda(c) = \XZm$.
\end{example}

The above description only discussed MBQC with post-selected measurement outcomes, i.e.\@ assuming determinism of measurements.
However, quantum measurements are fundamentally probabilistic processes: they may or may not induce Pauli errors upon observation.
In \cref{sec:channel}, we formalise measurements as completely positive maps or quantum channels.

\subsection{Linear optical circuits}\label{sec:LO}

Linear optical circuits are the basic ingredients of photonic computing.
They are generated by two physical gates: beam splitters and phase shifts.
When equipped with $n$-photon state preparations and number-resolving photon detectors, they give rise to quantum statistics~\cite{aaronsonComputationalComplexityLinear2011}.
We depict the above-mentioned components, respectively, as follows:
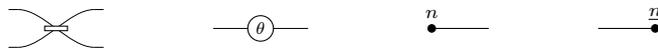
\[
  \begin{tikzpicture}
	\begin{pgfonlayer}{nodelayer}
		\node [style=none] (0) at (-1, 0.5) {};
		\node [style=none] (1) at (-1, -0.5) {};
		\node [style=none] (2) at (1, -0.5) {};
		\node [style=none] (3) at (1, 0.5) {};
		\node [style=beamSplitter] (4) at (0, 0) {};
		\node [style=none] (5) at (-1.25, 0.5) {};
		\node [style=none] (6) at (1.25, 0.5) {};
		\node [style=none] (7) at (1.25, -0.5) {};
		\node [style=none] (8) at (-1.25, -0.5) {};
	\end{pgfonlayer}
	\begin{pgfonlayer}{edgelayer}
		\draw [style=LOWire, in=0, out=-180, looseness=0.75] (3.center) to (1.center);
		\draw [style=LOWire, in=180, out=0, looseness=0.75] (0.center) to (2.center);
		\draw [style=LOWire] (8.center) to (1.center);
		\draw [style=LOWire] (5.center) to (0.center);
		\draw [style=LOWire] (6.center) to (3.center);
		\draw [style=LOWire] (7.center) to (2.center);
	\end{pgfonlayer}
\end{tikzpicture}

  \qquad \qquad
  \begin{tikzpicture}
	\begin{pgfonlayer}{nodelayer}
		\node [style=none] (0) at (-1.25, 0) {};
		\node [style=none] (1) at (1.25, 0) {};
		\node [style=phaseCircle] (2) at (0, 0) {$\theta$};
	\end{pgfonlayer}
	\begin{pgfonlayer}{edgelayer}
		\draw [style=LOWire] (1.center) to (2);
		\draw [style=LOWire] (2) to (0.center);
	\end{pgfonlayer}
\end{tikzpicture}

  \qquad \qquad
  \begin{tikzpicture}
	\begin{pgfonlayer}{nodelayer}
		\node [style=black node] (0) at (-0.75, 0) {};
		\node [style=none] (1) at (0.75, 0) {};
		\node [style=sLabel] (2) at (-0.75, 0.375) {$n$};
	\end{pgfonlayer}
	\begin{pgfonlayer}{edgelayer}
		\draw (1.center) to (0);
	\end{pgfonlayer}
\end{tikzpicture}

  \qquad \qquad
  \begin{tikzpicture}
	\begin{pgfonlayer}{nodelayer}
		\node [style=black node] (0) at (0.75, 0) {};
		\node [style=none] (1) at (-0.75, 0) {};
		\node [style=sLabel] (2) at (0.75, 0.375) {$\cvar n$};
	\end{pgfonlayer}
	\begin{pgfonlayer}{edgelayer}
		\draw (1.center) to (0);
	\end{pgfonlayer}
\end{tikzpicture}

\]
We call the graphical language generated by the above components $\bf{LO}$.
Graphical axiomatisations of these components can be found in~\cite{clementLOvCalculusGraphicalLanguage2022, defeliceQuantumLinearOptics2023, heurtelCompleteGraphicalLanguage2024}.
We use a special notation $\begin{tikzpicture}
	\begin{pgfonlayer}{nodelayer}
		\node [style=none] (0) at (0.75, 0) {};
		\node [style=wn] (1) at (-0.25, 0) {};
	\end{pgfonlayer}
	\begin{pgfonlayer}{edgelayer}
		\draw (1) to (0.center);
	\end{pgfonlayer}
\end{tikzpicture}
 \coloneqq \begin{tikzpicture}
	\begin{pgfonlayer}{nodelayer}
		\node [style=none] (0) at (0.75, 0) {};
		\node [style=black node] (1) at (-0.25, 0) {};
		\node [style=sLabel] (2) at (-0.25, 0.375) {$0$};
	\end{pgfonlayer}
	\begin{pgfonlayer}{edgelayer}
		\draw (1) to (0.center);
	\end{pgfonlayer}
\end{tikzpicture}
$ for the \enquote{empty} state.
Diagrams in $\bf{LO}$ have an interpretation as linear maps acting on tensor products of $\mathcal{F}(\mathbb{C})$
--- the bosonic Fock space with a single degree of freedom, also called \emph{bosonic mode}.
Formally, the bosonic Fock space over a vector space $\mathcal{H}$ is given by $\mathcal{F}(\mathcal{H}) = \bigoplus_n \mathcal{H}^{\tilde{\otimes} n}$
where $\tilde{\otimes}$ denotes the quotient of the tensor product by $x \otimes y \sim y \otimes x$ and $\bigoplus$ denotes the infinite direct sum.
$\mathcal{F}(\mathcal{H})$ is usually completed to a Hilbert space but we here focus on states containing finitely many photons and thus use the above definition.
For a finite number of modes $m$, the Fock space is isomorphic to the tensor product of $m$ copies of the single bosonic mode $\mathcal{F}(\mathbb{C}^m) \simeq \mathcal{F}(\mathbb{C})^{\otimes m}$~\cite{defeliceQuantumLinearOptics2023}.
Given an $\bf{LO}$ circuit on $m$ modes, the amplitudes of different input-output pairs can be computed by taking permanents of an underlying $m \times m$ unitary matrix.
We fix the interpretation of the phase shift of angle $\theta$ to be the matrix
$\begin{pmatrix}
   e^{i \theta}
\end{pmatrix}$,
and of the beam splitter to be the Hadamard matrix,
$\frac{1}{\sqrt{2}}\begin{pmatrix}
                     1 & 1  \\
                     1 & -1
\end{pmatrix}$.
The underlying matrix of arbitrary circuits is obtained by a simple block-diagonal matrix multiplication.

\begin{remark}
  Throughout the paper, we use two different types of classical variables.
  Underlined variables are \emph{outcomes} of a quantum measurement: they can take different values probabilistically.
  Other variables are \emph{controlled}: they have a fixed value, set before the computation is executed.
  For example, in the picture above, the number of prepared photons $\ccvar{n}$ is controlled while the number of detected
  photons $\cvar{n}$ is a probabilistic outcome. We formalise this notation in \cref{subsec:kraus-map}.
\end{remark}

Circuits in $\bf{LO}$ can be further decomposed using the generators of the QPath calculus~\cite{defeliceQuantumLinearOptics2023},
given by the following maps:
\begin{align*}
    \scalebox{0.7}{\begin{tikzpicture}
	\begin{pgfonlayer}{nodelayer}
		\node [style=none] (0) at (-0.5, 0) {};
		\node [style=none] (1) at (0, 0.25) {};
		\node [style=none] (2) at (0, -0.25) {};
		\node [style=none] (3) at (-1.5, 0) {};
		\node [style=none] (4) at (1.5, 1) {};
		\node [style=none] (5) at (1.5, -1) {};
	\end{pgfonlayer}
	\begin{pgfonlayer}{edgelayer}
		\draw (0.center) to (1.center);
		\draw (1.center) to (2.center);
		\draw (2.center) to (0.center);
		\draw (3.center) to (0.center);
		\draw (1.center) to (4.center);
		\draw (2.center) to (5.center);
	\end{pgfonlayer}
\end{tikzpicture}
} &\qquad \overset{\interp{\cdot}}{\longmapsto} \qquad \ket{n} \quad \mapsto \quad \sum_{k=0}^n \binom{n}{k}^{\frac{1}{2}} \ket{k}\ket{n - k} \\
    \scalebox{0.7}{\begin{tikzpicture}
	\begin{pgfonlayer}{nodelayer}
		\node [style=none] (0) at (-1.25, 0) {};
		\node [style=none] (1) at (1.25, 0) {};
		\node [style=phase] (2) at (0, 0) {c};
	\end{pgfonlayer}
	\begin{pgfonlayer}{edgelayer}
		\draw [style=LOWire] (1.center) to (2);
		\draw [style=LOWire] (2) to (0.center);
	\end{pgfonlayer}
\end{tikzpicture}
} &\qquad \overset{\interp{\cdot}}{\longmapsto} \qquad \ket{n} \quad \mapsto \quad c^n \ket{n}
\end{align*}
for $c \in \mathbb{C}$.
This enables a rewriting system given in \cref{sec:qpath}, and a more fine-grained diagram for the beam splitter,
\begin{equation*}
    \input{RUS/beamsplitter-qpath.tikz}
\end{equation*}
which we use in \cref{sec:mixedfusion}.
Note that the QPath maps are bounded on the Fock space as defined above, yet unbounded on its Hilbert space completion.
We must therefore be careful to use them only if there are finitely many photons. This assumption is valid throughout the paper and further discussed in \cref{subsec:cpmaps}.

\subsection{Dual-rail encoding}\label{sec:bgd-dual-rail}

In photonic quantum computing, qubits are usually encoded by a photon in a pair of bosonic modes, a method known as \emph{dual-rail encoding}~\cite{knillSchemeEfficientQuantum2001}.
These could be two possible positions of the photon (spatial modes), or any other binary degree of freedom of the photon such as polarisation.
We use a \enquote{double wire} \begin{tikzpicture}
	\begin{pgfonlayer}{nodelayer}
		\node [style=none] (0) at (-0.75, 0) {};
		\node [style=none] (1) at (0.75, 0) {};
	\end{pgfonlayer}
	\begin{pgfonlayer}{edgelayer}
		\draw [style=DRWire] (0.center) to (1.center);
	\end{pgfonlayer}
\end{tikzpicture}
 to denote dual-rail modes.
These wires are interpreted as $\mathcal{F}(\mathbb{C}^2)$, the bosonic Fock space over a qubit,
meaning that there can be any number of qubits in the same dual-rail mode.
Linear optical operations on these modes are defined from $\bf{LO}$ by using the following maps:
\[
  \input{dualrail/dual_rail_split.tikz}
  \qquad \qquad \qquad \qquad
  \input{dualrail/dual_rail_merge.tikz}
\]
These two maps are inverses of each other, corresponding to the natural isomorphism $\mathcal{F}(\mathbb{C}^2) \simeq \mathcal{F}(\mathbb{C}) \otimes \mathcal{F}(\mathbb{C})$.
We wish to represent processes acting on dual-rail qubits using the ZX calculus~\cite{coeckeInteractingQuantumObservables2008}.
However, ZX diagrams act on qubit spaces of the form $(\mathbb{C}^2)^{\otimes m}$ and there is no standard way of extending this action to $\mathcal{F}(\mathbb{C}^2)^{\otimes m}$.
There is, however, an isometry $\mathbb{C}^2 \to \mathcal{F}(\mathbb{C}^2)$, encoding a qubit state into its dual-rail representation.
We call it \enquote{triangle} and represent it as follows:
\begin{equation}
  \scalebox{1.5}{\begin{tikzpicture}
	\begin{pgfonlayer}{nodelayer}
		\node [style=ltriang] (0) at (0, 0) {};
		\node [style=none] (1) at (-1, 0) {};
		\node [style=none] (2) at (1, 0) {};
	\end{pgfonlayer}
	\begin{pgfonlayer}{edgelayer}
		\draw (1.center) to (0);
		\draw [style=DRWire] (2.center) to (0);
	\end{pgfonlayer}
\end{tikzpicture}
}
  \label{eq:triangle}
\end{equation}
Note that the adjoint of the triangle is a projector onto the qubit subspace, and we never use it in this paper.
We can now translate between dual-rail circuits and ZX diagrams using this graphical component.
For example, the qubit computational basis states are given by the dual-rail states.
\begin{restatable}{rewrite}{encodeKet}
  
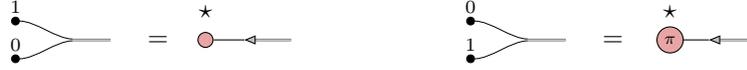
\label{eq:dual-rail-encoding}
  \[
  \input{dualrail/encode_ket_0.tikz}
  \qquad \qquad \qquad
  \input{dualrail/encode_ket_1.tikz}
  \]
\end{restatable}
Going further, the following equations imply that any single-qubit unitary can be realised on dual-rail qubits using only linear optical instruments.
\begin{restatable}{rewrite}{encodeHad}
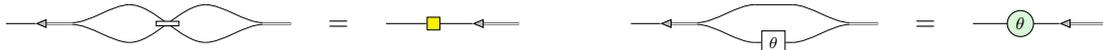
  \[
  \input{dualrail/EncodeHadamard.tikz}
  \qquad \qquad
  \input{dualrail/EncodeZSpider11.tikz}
  \]
\end{restatable}
Similarly, we may perform any single-qubit measurement on dual-rail qubits using photon detectors.
\begin{restatable}{rewrite}{encodeMeas}
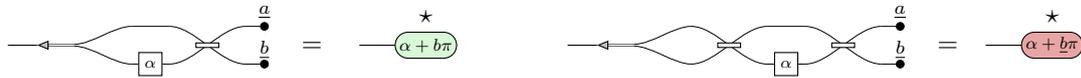
  \[
  \input{dualrail/EncodeZAplhaMeasurement.tikz}
  \qquad \qquad
  \input{dualrail/EncodeXAplhaMeasurement.tikz}
  \]
\end{restatable}
By pushing triangles from left to right in a dual-rail diagram, we compute the action of this diagram assuming
that a single photon is input in each dual-rail mode.
\begin{remark}
  Our results in this paper focus on dual-rail encoded qubits. 
  However, the only properties we use of the dual-rail encoding are the equations given in this subsection.
  For a straightforward generalisation to qubit encodings in $n$ (rather than $2$) optical modes, one may replace the beam splitter and tensored phase shift with $n \times n$ unitaries satisfying the same equations.
  For a more involved generalisation to \emph{qudit} encodings, similar equations could be used to relate linear optics and the qudit ZX calculus \cite{poorZXcalculusCompleteFinitedimensional2024}.
\end{remark}

\section{Channel construction}\label{sec:channel}

The graphical language $\mathbf{ZXLO}$ described so far has a standard interpretation in \emph{pure} quantum mechanics, 
consisting of an assignment of a linear map $\interp{D}$ to any diagram $D \in \mathbf{ZXLO}$.
In order to interpret these diagrams as \emph{quantum channels}, 
we need to define notions of environment and classical interface on the diagram.
In categorical quantum mechanics, the standard approach to modeling quantum channels uses Frobenius algebras~\cite{abramskyAlgebrasNonunitalFrobenius2012, coeckeCategoriesQuantumClassical2016} 
and employs a diagrammatic doubling procedure based on Stinespring dilation to capture mixed-state quantum mechanics. 
However, these approaches require the burden of either having twice as large diagrams, or keeping track of both quantum and classical wire types. 
Here, working with Kraus decompositions fixed by measurements, we utilize variables to represent classical data for cleaner syntax.
This approach can be found in the literature~\cite{duncanGraphicalApproachMeasurementbased2013},
but we here extend the notation to also model feedforward of classical information as variable annotations on the diagram. 
This gives us a premonoidal syntax~\cite{power_premonoidal_1997} for describing quantum channels where the order of boxes in the diagram defines the allowed
classical annotations.
We give a formal interpretation for this notation in terms of completely positive maps between vector spaces, focusing on their action on discrete-variable photonic states.
We then show that the determinism results of~\cite{browneGeneralizedFlowDeterminism2007,danosMeasurementCalculus2007} can be naturally expressed in the language
of $\mathtt{Channel}(\mathbf{ZX})$ diagrams, giving us tools for proving determinism of quantum channels which we will extend in subsequent sections.

\subsection{Kraus notation for abstract channels}\label{subsec:kraus-map}

Let $\bf{C}$ be any base graphical language of linear maps. We extend $\bf{C}$ with classical annotations to construct a language of classical-quantum maps $\mathtt{Channel}(\mathbf{C})$
Fix a global set of variables $x, y \in \chi$ and consider the abstract notion of Kraus map over a base graphical language. 

\begin{definition}[Kraus map]\label{def:kraus-map}
  A Kraus map $f: Q \to Q'$ over $\mathbf{C}$ is given by the following data:
  \begin{enumerate}
    \item a type $E$ in $\mathbf{C}$, called the environment,
    \item a pair of input and output typed variables $x : X$ and $y : Y$ where $X, Y$ are objects of $\mathbf{C}$, called the classical interface,
    \item a morphism of type $f: H \otimes X \to K \otimes E \otimes Y$ in $\mathbf{C}$. 
  \end{enumerate}
  and denoted as follows:
  $$\input{kraus-map.tikz}$$
  A Kraus map is called \emph{pure} if $X = Y = E = I$ are the unit of the tensor.
\end{definition}

We consider arbitrary compositions of Kraus maps over $\mathbf{C}$, in sequence and in parallel. 
This is equivalent to considering arbitrary diagrams built from $\mathbf{C}$ and the following 
additional generators on every wire:
\begin{equation}\label{eq:channel-generators}
  \input{discard.tikz} \quad \text{ (discard) }, \qquad \begin{tikzpicture}
	\begin{pgfonlayer}{nodelayer}
		\node [style=black node] (0) at (-0.5, 0) {};
		\node [style=none] (1) at (0.5, 0) {};
		\node [style=none] (2) at (-0.5, 0.5) {$\ccvar{x}$};
	\end{pgfonlayer}
	\begin{pgfonlayer}{edgelayer}
		\draw (0) to (1.center);
	\end{pgfonlayer}
\end{tikzpicture}
 \quad \text{ (prepare) }, \qquad \begin{tikzpicture}
	\begin{pgfonlayer}{nodelayer}
		\node [style=black node] (0) at (0.5, 0) {};
		\node [style=none] (1) at (-0.5, 0) {};
		\node [style=none] (2) at (0.5, 0.5) {$\cvar{y}$};
	\end{pgfonlayer}
	\begin{pgfonlayer}{edgelayer}
		\draw (0) to (1.center);
	\end{pgfonlayer}
\end{tikzpicture}
 \quad \text{ (measure) }
\end{equation}
We however distinguish between the two types of composition below:
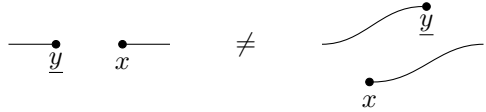
\[
  \input{interchanger-fail.tikz}
\]
Intuitively, measurements have classical outcomes modeled as side-effects.
In the first diagram the classical output
variable $y$ can influence the input variable $x$, while in the second diagram it cannot.
Thus, any diagram comes with a total order on the measurement and preparation commands,
which induces a strict ordering of the classical variables appearing in the diagram.
For a diagram $D$ we denote by $X_D \subset \chi$ the set of preparation variables in $D$, 
by $Y_D \subset \chi$ the set of measurement variables in $D$, and by $<_D$ the total order on $X_D \cup Y_D$.

\begin{definition}[Causal annotations]\label{def:annotations}
  Fix a set of function symbols $f, g$ and constant symbols $a, b$.
  Given a diagram $D$ built from $\mathbf{C}$ and the generators (\ref{eq:channel-generators}),
  a classical annotation $\mathcal{E}$ for $D$ is a set of typed \emph{coarse-graining variables} $Z_D \subset \chi \backslash (X_D \cup Y_D)$, 
  and a set of functional relations including:
  \begin{enumerate}
    \item \emph{coarse-graining annotations} for every $z \in Z_D$ and some $\nu_i \in Y_D \cup Z_D \backslash \set{z}$, of the form
    \begin{equation}\label{eq:coarse-graining}
      \cvar{z} = g(\nu_1, \dots, \nu_n) \, ,
    \end{equation}
    \item \emph{feedforward annotations} for some $x \in X_D$ and $\nu_i \in Y_D \cup Z_D$, of the form
    \begin{equation}\label{eq:feedforward-annotation}
      x = f(\nu_1, \dots \nu_n) \, ,
    \end{equation}
    \item \emph{post-selection annotations} for some constant symbol $a$ and $y \in Y_D$, of the form
    \begin{equation}\label{eq:post-selection-annotation}
      \cvar{y} = a \, .
    \end{equation}
  \end{enumerate}
  This defines a relation $\leq_{D, \mathcal{E}}$ on the set of variables $X_D \cup Y_D \cup Z_D$ where $\nu \leq_{D, \mathcal{E}} \mu$ 
  whenever either $\nu = \mu$, or $\nu <_D \mu$, or there is an annotation in $\mathcal{E}$ such that $\nu$ appears on the right-hand side and $\mu$ appears on the left-hand side.
  We say that a classical annotation $\mathcal{E}$ is \emph{causally compatible} with $D$ if $\leq_{D, \mathcal{E}}$ is a partial order, i.e. reflexive, antisymmetric and transitive.
\end{definition}

Annotations are interpreted as functions acting on classical data, allowing to specify the classical input variables of boxes with respect to previous output variables in the diagram.

\begin{definition}[Channel]
  A channel diagram $D: H \to K$ in $\mathtt{Channel}(\mathbf{C})$ consists of the following data:
  \begin{enumerate}
    \item a diagram $D$ built from $\mathbf{C}$ and the following additional generators 
    \[\left\{\input{discard.tikz} \,, \, \begin{tikzpicture}
	\begin{pgfonlayer}{nodelayer}
		\node [style=black node] (0) at (-0.5, 0) {};
		\node [style=none] (1) at (0.5, 0) {};
		\node [style=none] (2) at (-0.5, 0.5) {$\ccvar{x}$};
	\end{pgfonlayer}
	\begin{pgfonlayer}{edgelayer}
		\draw (0) to (1.center);
	\end{pgfonlayer}
\end{tikzpicture}
 \,, \, \begin{tikzpicture}
	\begin{pgfonlayer}{nodelayer}
		\node [style=black node] (0) at (0.5, 0) {};
		\node [style=none] (1) at (-0.5, 0) {};
		\node [style=none] (2) at (0.5, 0.5) {$\cvar{y}$};
	\end{pgfonlayer}
	\begin{pgfonlayer}{edgelayer}
		\draw (0) to (1.center);
	\end{pgfonlayer}
\end{tikzpicture}
\right\}_{x, y \in \chi}\]
    \item a classical annotation $\mathcal{E}_D$ causally compatible with $D$.
  \end{enumerate}
  The set of variables appearing in a channel diagram $D$ is $X_D \cup Y_D \cup Z_D$. 
  The classical input variables $I_D$ are the variables in $X_D$ that do not appear in any feedforward annotation.
  The output variables $O_D$ are the variables in $Y_D \cup Z_D$ that do not appear on the right-hand side of any coarse-graining or feedforward annotation, or on the left-hand side of a post-selection annotation.
\end{definition}

For example, the channel diagram below (left) has causally compatible coarse-graining and feedforward annotations with $X_D = \set{x, x'}$, $Y_D= \set{y, y'}$, $Z_D =\set{z}$ and $I_D = \set{x}$, $O_D = \set{z}$. 
Instead, the annotation on the right is not causally compatible with the diagram as $x$ appears before $y$, and thus gives an invalid channel diagram.
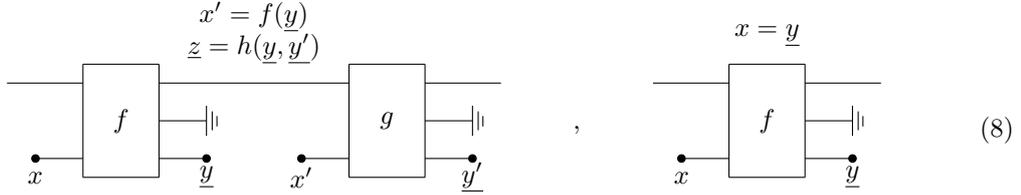
\begin{equation}\label{eq:channel-example}
  \input{channel-example.tikz}
\end{equation}
Parallel and sequential compositions of diagrams in $\mathtt{Channel}(\mathbf{C})$ are given by composing the underlying diagrams 
and appending the sets of equations, after appropriate relabeling with fresh variables in $\chi$. This ensures that 
the causal compatibility condition of \cref{def:annotations} is respected.

\begin{remark}
  Formally, $\mathtt{Channel}(\mathbf{C})$ is a premonoidal category~\cite{power_premonoidal_1997} with monoidal center $\mathbf{C} + \{\input{discard.tikz}\}$, that is,
  the interchange law holds for any pair of morphisms in $\mathtt{Channel}(\mathbf{C})$ that does not involve a controlled preparation or a measurement.
  The coarse-graining and feedforward annotations may be viewed as operations acting on an implicit ``runtime object'' for $\mathtt{Channel}(\mathbf{C})$~\cite{roman_string_2025}.
\end{remark}

\subsection{Interpretation as completely positive maps}\label{subsec:cpmaps}

Diagrams in $\bf{ZXLO}$ correspond to linear maps acting on tensor products of qubits and bosonic modes.
We now give an interpretation of $\mathtt{Channel}(\bf{ZXLO})$ diagrams as completely positive maps acting on mixed states with finitely many particles.
For any object $H$ of $\bf{ZXLO}$ consisting of $m$ bosonic modes and $k$ qubits, define the $N$-particle projector $P_{\leq N}: H \to H$ as the operator that acts as the identity on qubits and projects out all states with more than $N$ photons over the $m$ bosonic modes.
A \emph{finite-particle state} on $H$, also known as a regular state \cite{chaiken_finite-particle_1967}, 
can be defined as a positive operator $\rho: H \to H$ satisfying $\rho = P_{\leq N} \rho P_{\leq N}$ for some natural number $N$. 
We denote by $B_\ast(H)$ the space of all finite-particle states.
In finite dimensions this coincides with the full space of bounded operators.
For bosonic modes, it contains all discrete-variable states \cite{chaiken_finite-particle_1967, reck_experimental_1994,knillSchemeEfficientQuantum2001} and truncated optical states \cite{pegg_optical_1998,miranowicz_quantum-optical_2001, ozdemir_quantum-scissors_2001}, widely used in the literature.
Finite-particle states are finite-rank and trace-class so that $B_\ast(H)$ forms an inner-product space with the usual Hilbert-Schmidt inner product. 
We do not impose any normalisation condition on states, as this will allow us to extract physical quantities from our interpretation.

Fix a basis $\mathcal{B}_H$ of the vector space associated to any object in $\mathbf{ZXLO}$.
For the qubit space $\mathcal{B} = 2$ is the set with two elements,
while for single-mode bosonic space we have $\mathcal{B} = \mathbb{N}$.
For a channel diagram $D$, we denote by $\mathcal{B}_{X_D}, \mathcal{B}_{Y_D}, \mathcal{B}_{Z_D}$ the basis elements for the preparation, measurement and coarse-graining variables, respectively.
Similarly, the unbound input and output classical variables range over the basis elements $\mathcal{B}_{I_D}$ and $\mathcal{B}_{O_D}$.
We moreover fix an interpretation of each function (or constant) symbol as a function on the basis elements associated to its input and output variable types.
We group the environments $E_i$ of individual Kraus maps within a diagram $D$ into a single environment $E_D = \bigotimes_i E_i$.

\begin{definition}[Branches]
  The branches of a diagram $D: H \to K$ in $\mathtt{Channel}(\mathbf{ZXLO})$ with preparation commands in $X_D$ and measurement commands in $Y_D$,
  are the collection of diagrams $H \to K \otimes E_D$ in $\mathbf{ZXLO}$ obtained by setting different values for the variables 
  $x \in \mathcal{B}_{X_D}$ and $y \in \mathcal{B}_{Y_D}$.
  We sometimes view a branch as a diagram $D_{x, y}: H \to K$ in $\mathtt{Channel}(\mathbf{ZXLO})$ by attaching the discard map to the environment $E_D$.
\end{definition}

We can now give a direct interpretation of a channel $D$ as a completely positive map $B_\ast(H) \otimes X' \to B_\ast(K) \otimes Y'$ 
where $B_\ast(H)$ is the space of finite-particle states $\rho: H \to H$, 
and $X', Y'$ are the spaces of the input and output classical variables in $D$.

\begin{definition}[CQ interpretation]
  Given a diagram $D$ in $\mathtt{Channel}(\mathbf{ZXLO})$, the classical-quantum map (CQ) interpretation is given by the following operator on finite-particle states:
  \[
    \interp{D}_{CQ}(\rho) = \sum_{x, y, z \in \mathcal{B}_D \text{ satisfying } \mathcal{E}_D} \ket{y'} tr_E(\interp{D_{x, y}}^\dagger \rho \interp{D_{x, y}}) \bra{x'}
  \]
  \noindent where $x, y, z \in \mathcal{B}_D = \mathcal{B}_{X_D} \times \mathcal{B}_{Y_D} \times \mathcal{B}_{Z_D}$, 
  $E = E_D$ is the total environment of the diagram, 
  and $x', y' \in \mathcal{B}_{I_D} \times \mathcal{B}_{O_D}$ are the unbound input and output variables of the channel diagram.
\end{definition}

To show that the above interpretation is sound and does not result in infinite quantities we prove the following proposition.

\begin{proposition}[Soundness]
  For any diagram $D \in \bf{ZXLO}$ and positive finite-particle state $\rho$, the CQ interpretation $\interp{D}_{CQ}(\rho)$ is a positive finite-particle state.
\end{proposition}
\begin{proof}
  Let $\rho: H \to H$ be a positive finite-particle state. 
  Every generator $G$ introduced in the previous section, except for the $n$-photon state preparation, commutes with the $N$-particle projector, $P_{\leq N}G = G P_{\leq N}$.
  The $n$-photon state preparation creates $n$ additional photons and thus satisfies $(\ket{n} \otimes \mathtt{id}_H) P_{\leq N}  = P_{\leq N + n} (\ket{n} \otimes \mathtt{id}_H)$ for any reference space $H$.
  Therefore for any diagram $D$ in $\bf{ZXLO}$ we have $\interp{D} P_{\leq N} = P_{\leq N + a} \interp{D}$ where $a$ is the total number of photon creations in $D$.
  Therefore for any $x, y$,
  \[\interp{D_{x, y}}^\dagger \rho \interp{D_{x, y}} = \interp{D_{x, y}}^\dagger P_{\leq N} \rho P_{\leq N} \interp{D_{x, y}} = P_{\leq N + a} \interp{D_{x, y}}^\dagger \rho \interp{D_{x, y}} P_{\leq N + a}\]
  gives a finite-particle state which is moreover positive by positivity of $\rho$. 
  Any finite-particle state is trace-class, and the partial trace preserves positivity, so that $tr_E(\interp{D_{x, y}}^\dagger \rho \interp{D_{x, y}})$ is also positive and bounded.
  It also has finitely many particles since the trace induces a sum over finitely many finite-particle states.
  Moreover, finite-particle states can only result in finitely many measurement outcomes, and the causal compatibility condition ensures that all the feedforward 
  and coarse graining variables are deterministic functions of previous measurement variables and thus can only take finitely many values. 
  For example, the invalid diagram on the right of \cref{eq:channel-example} could result in infinitely many consistent assignments of values to $x$ and $y$, but this circularity is avoided in valid channel diagrams.
  Therefore the sum $\sum_{x, y, z \in \mathcal{B}_D \text{ satisfying } \mathcal{E}_D}tr_E(\interp{D_{x, y}}^\dagger \rho \interp{D_{x, y}})$ has only finitely many terms, 
  and thus gives a positive finite-particle state as required.
\end{proof}

Since all states considered are finite-rank and trace-class, the \emph{trace operator} 
$\rho \mapsto \sum_{x \in \mathcal{B}_H} \bra{x} \rho \ket{x}$ is well-defined and corresponds 
to the CQ interpretation of the discarding map $\input{discard.tikz}$.
We often restrict our attention to \emph{trace-preserving} channels, satisfying:
\begin{equation}\label{eq:causality}
  \input{causal-map.tikz}
\end{equation}
When $D$ is a pure map the equation above corresponds to $\interp{D}^\dagger \interp{D} = I$, i.e. $\interp{D}$ is an isometry.

We interpret formal sums of diagrams in $\mathtt{Channel}(\mathbf{ZXLO})$ as the sum of their CQ interpretation. 
With the scalars present in $\mathbf{ZXLO}$ we will only be able to form positive real linear combinations of diagrams, 
since $\interp{s}_{CQ} = \vert \interp{s} \vert^2$ for any scalar diagram $s$ in $\mathbf{ZXLO}$.
Note that this sum is akin to \emph{mixing}, rather than superposition, which instead corresponds to summing the Kraus maps in $\mathbf{ZXLO}$ before constructing the channel. 
In other words:
\[
  \interp{D}_{CQ} + \interp{D'}_{CQ} \text{ (mixing) } \, \neq \, \interp{D + D'}_{CQ} \text{ (superposition) }
\]
Causal maps are closed under taking \emph{probability distributions}: for a discrete probability distribution $p_i \in [0, 1]$ with $i \in X$, $X$ finite, 
and causal maps $\set{f_i}_{i \in X}$, the map defined by $\sum_{i \in X} p_i f_i$ is also causal.

We can represent a quantum channel with classical output as a diagram $D$ labelled by an output variable $\cvar{k}$.
Then, the probability of an outcome $e$, given an input state $\rho$, is obtained by setting $\cvar{k} = e$ in $D$
and tracing out the remaining outputs:
\[ 
P_D(\cvar{k} = e \, \vert \, \rho) = \quad \input{probability.tikz}
\]
We can use these post-selection annotations to define a notion of \emph{implementation} between channels.

\begin{definition}[Probabilistic implementation]\label{def:implementation}
  We say that a channel $D$ implements a channel $C$ with probability $p$ if there exists a function $\cvar{s} = f(\cvar{y_1}, \dots, \cvar{y_m})$
  of the classical outcome variables in $D$ such that $\interp{D_{\cvar{s}=1}}_{CQ} = p \interp{C}_{CQ}$.
  We say that $D$ implements $C$ deterministically if $\interp{D}_{CQ} = \interp{C}_{CQ}$.
\end{definition}

Measurements in quantum mechanics induce probabilistic branching and mixing. 
To result in a \emph{deterministic} computation, all the individual branches of the channel must contribute to the same process.

\begin{definition}[Determinism]\label{def:determinism}
  A channel $D$ is deterministic if all the branches are proportional to each other in the CQ interpretation. 
  It is strongly deterministic if all branches are equal up to a global phase, i.e. if all the branches are equal to each other in the CQ interpretation.
\end{definition}

Thus, a channel $D$ is deterministic if for any pattern of measurement outcomes, it results in the same pure computation.
It is strongly deterministic if these outcome patterns have equal probability.

\subsection{Qubit and optical channels}\label{subsec:channel-examples}

We now apply the $\mathtt{Channel}$ construction to our base language $\bf{ZXLO}$ of ZX diagrams and linear optical circuits, 
and give examples of how it can be used to compute properties of quantum channels.

\paragraph{Discarding and noise}

We start by considering the class of processes generated by only pure maps and discarding.
The discarding maps for each space together with the relations between them are as follows:
\[
  \input{discards.tikz}
\]
Using these maps we can represent common error channels from their Kraus decomposition.
For example, the photon loss channel with transmittivity $\eta$ is modeled as a beam splitter with an output discarded~\cite{oszmaniec_classical_2018}:
$$\input{photon-loss.tikz}$$
Similarly, the single-qubit bitflip error channel with probability $p$ can be written as:
$$\input{bitflip-error.tikz}$$
using the single-qubit state $\interp{\ket{\sqrt{p}}} = \sqrt{p} \ket{1} + \sqrt{1-p} \ket{0}$.
Our focus in this paper is however on \emph{noiseless channels} where the environment is fully observed by the experimenter.
These are simply defined as channels $D$ with a trivial environment $E_D = I$, i.e.\@ such that the discarding maps do not appear in the diagram.

\paragraph{Measurements and coarse-graining}

Even in a noiseless environment, destructive measurements can lead to branching and non-determinism in the CQ interpretation.
The $X$ and $Z$ single-qubit measurements can be written as:
\begin{gather*}
    \input{kraus-Zmeasurement.tikz} \quad \overset{\interp{\cdot}_{CQ}}{\longmapsto} \quad \rho \mapsto \ket{0} \bra{0}\rho\ket{0} + \ket{1} \bra{1}\rho\ket{1}\\
    \input{kraus-Xmeasurement.tikz} \quad \overset{\interp{\cdot}_{CQ}}{\longmapsto} \quad \rho \mapsto \ket{0} \bra{\plus}\rho\ket{\plus} + \ket{1} \bra{\minu}\rho\ket{\minu}\\
\end{gather*}
The photon-number resolving measurement on an optical mode has the following interpretation:
\[
  \begin{tikzpicture}
	\begin{pgfonlayer}{nodelayer}
		\node [style=black node] (0) at (0.75, 0) {};
		\node [style=none] (1) at (-0.75, 0) {};
		\node [style=sLabel] (2) at (0.75, 0.375) {$\cvar n$};
	\end{pgfonlayer}
	\begin{pgfonlayer}{edgelayer}
		\draw (1.center) to (0);
	\end{pgfonlayer}
\end{tikzpicture}
 \qquad \overset{\interp{\cdot}_{CQ}}{\longmapsto} \qquad \rho \mapsto \sum_{n \in \mathbb{N}} \ket{n} \bra{n} \rho \ket{n} \, .
\]
We can treat the diagrams built from these measurement commands as parametrised diagrams in $\mathbf{ZXLO}$.
For example, we can define spiders with a classical output variable, corresponding to entangling a qubit to an ancilla and measuring it.
\[
  \input{non-destructive-measurementZ.tikz}\, .
\]
While these spiders satisfy the usual spider fusion law, note that coarse-graining equations between output variables have an action on the scalars in the diagram.
Diagrams with coarse graining annotations satisfy the following rule.
\begin{restatable}{rewrite}{corseGraining}
  \label{eq:coarse-graining-rewrite}
  \[
  \input{lemma-measure-proof-3.tikz}
  \]
\end{restatable}
This rewrite removes a $\star$, a pair of variables, and an annotation from the diagram to correctly account for the coarse-grained probabilities: there are two possible outcomes for $\cvar{c}$, each happening with probability $\interp{\star}_{CQ} = \frac{1}{2}$.
Spiders labeled by a classical outcome allow to construct general projective-valued measures (POVMs), such as ancilla quantum measurements.
The outcomes correspond to Pauli byproducts and can be propagated through the diagram via ZX rewrites.

\paragraph{Feedforward and correction}

Feedforward is crucial to perform universal quantum computation with linear optics.
In this work, we only use bit-controlled operations generated by the optical controlled phase flip gate, given by:
$$\input{bit-controlled-phase.tikz}$$
This process acts as the identity if the control parameter $\ccvar{x} = 0$ and as the phase flip if $\ccvar{x} = 1$.
Formally, the Kraus map for this controlled process should be considered an additional generator of our base language, 
but the simplified syntax above directly captures the properties of its interpretation.
Switches can be built using bit-controlled phases:
\begin{equation}\label{eq:switch}
  \input{controlled-switch.tikz}
\end{equation}
and, vice versa, a composition of two switches can be used to build a controlled phase.
Probabilistic controlled phase gates can be obtained with passive linear optics~\cite{prevedel_high-speed_2007,lemr_experimental_2011}
but higher fidelities can be achieved with active optical switches~\cite{zaninFibercompatiblePhotonicFeedforward2021}.

On qubits, we use the classically controlled X and Z gates:
$$\input{controlled-XZ.tikz}$$
We moreover encode feedforward operations by reusing variables in the diagram.
Diagrams with feedforward annotations can be simplified using the following rule.
\begin{restatable}{rewrite}{feedForward}
  
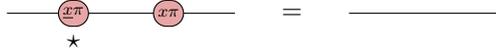
\label{eq:feedforward-rewrite}
  \[
  \input{controlled-correction.tikz}
  \]
\end{restatable}
As an example, we can use this to prove the teleportation protocol with a perfect Bell measurement:
$$\input{teleportation.tikz}$$
where the last step uses the feedforward rewrite twice.
This rewrite allows to remove measurement and correction commands from the diagram, while increasing the overall success probability (removing stars).
We can see the rewrite sequence above as showing \emph{determinism} for the teleportation protocol.
This correction argument can in fact be generalized to arbitrary ZX diagrams with flow, as we summarise in the remainder of this section.

\subsection{Pauli flow and determinism in qubitq patterns}

The notion of determinism of \cref{def:determinism} has been studied in detail for \emph{qubit channels}~\cite{browneGeneralizedFlowDeterminism2007,danosMeasurementCalculus2007}.
We now recast these results in our framework. \emph{Measurement patterns}~\cite{danosMeasurementCalculus2007} are a declarative language for MBQC, 
that describe how qubits are prepared, entangled, corrected, and measured.
They are usually defined as follows:

\begin{definition}[Measurement pattern~\cite{danosMeasurementCalculus2007}]\label{def:meas_pattern}
    A \emph{measurement pattern} consists of an $n$-qubit register $V$ with distinguished sets $I, O \subseteq V$ of input and output qubits and a sequence of commands consisting of the following operations:
    \begin{itemize}
      \item Preparations $N_i$, which initialise a qubit $i \in \comp{I}$ in the state $\ket{+}$.
      \item Entangling operators $E_{ij}$, which apply a $CZ$-gate to two distinct qubits $i$ and $j$.
      \item Destructive measurements $M_i^{\lambda,\alpha, \cvar{s}}$, which project a qubit $i\in \comp{O}$ onto the orthonormal basis $\{\ket{+_{\lambda,\alpha}},\ket{-_{\lambda,\alpha}}\}$, where $\lambda \in \{ \XYm, \XZm, \YZm \}$ is the measurement plane, $\alpha$ is the non-corrected measurement angle .
      The projector $\ket{+_{\lambda,\alpha}}\bra{+_{\lambda,\alpha}}$ corresponds to outcome $\cvar s = 0$ and $\ket{-_{\lambda,\alpha}}\bra{-_{\lambda,\alpha}}$ corresponds to outcome $\cvar s = 1$.
      \item Clifford operations $C_i$, which act on qubit $i$ by applying any Clifford unitary generated by the $S$ and $H$ gates.
      \item Corrections $[X_i]^t$, which depend on a measurement outcome (or a linear combination of measurement outcomes) $t\in\{0,1\}$ and act as the Pauli-$X$ operator on qubit $i$ if $t$ is $1$ and as the identity otherwise,
      \item Corrections $[Z_j]^s$, which depend on a measurement outcome (or a linear combination of measurement outcomes) $s\in\{0,1\}$ and act as the Pauli-$Z$ operator on qubit $j$ if $s$ is $1$ and as the identity otherwise.
    \end{itemize}
    A measurement pattern is \emph{runnable} if no command acts on a qubit already measured or not yet prepared (except preparation commands) and no correction depends on a qubit not yet measured.
\end{definition}

Runnable measurement patterns can be viewed as noiseless circuits in $\mathtt{Channel}(\mathbf{ZX})$.

\begin{proposition}
    Any runnable measurement pattern uniquely corresponds to a diagram in $\mathtt{Channel}(\mathbf{ZX})$ generated by the following operations:
    \[
    \input{pattern-def.tikz}    
    \]
    where $\lambda \in \set{I, S, H}$ is a Clifford map, and $\alpha, \omega \in [0, 2\pi)$ is an arbitrary angle.
\end{proposition}
\begin{proof}
  The defining operations of measurement patterns can all be built from the above generators, thus any runnable list of commands from \cref{def:meas_pattern}, 
  defines a unique premonoidal diagram.
\end{proof}

We can rewrite any measurement pattern to a $\mathtt{Channel}(\mathbf{ZX})$ diagram in a special form where all corrections appear before single-qubit measurements. 

\begin{definition}[Channel MBQC form]
  A diagram $D$ in $\mathtt{Channel}(\mathbf{ZX})$ is in MBQC form if there is a labelled open graph $(G, I, O, \lambda, \alpha)$ such that $D$ can be written as:
  \[ 
    \input{MBQC/MBQC-form-channel.tikz}
  \]
  where $x_i, z_i$ are control variables, $\cvar{k_i}$ are output variables, $\mathcal{E}$ is a set of feedforward annotations (exclusively)
  and $G$ denotes a graph state ZX diagram, called the \emph{underlying topology} of $D$.
\end{definition}

\begin{proposition}
    Any MBQC-form diagram in $\mathtt{Channel}(\mathbf{ZX})$ defines a measurement pattern with the same CQ interpretation.
\end{proposition}
\begin{proof}
  To go from an MBQC form diagram to a measurement pattern, we simply (1) decompose the graph state diagram with CZ gates and 
  (2) extract the control variables as X and Z correction gates using spider un-fusion. 
\end{proof}

In the context of measurement patterns, we require an even stronger form of determinism.

\begin{definition}[Uniform and stepwise determinism]
  A measurement pattern is (strongly) deterministic if it is (strongly) deterministic as a diagram in $\mathtt{Channel}(\bf{ZX})$.
  A measurement pattern is uniformly deterministic if it is deterministic for all choices of measurement angles $\alpha_i$.
  It is stepwise deterministic if all the patterns $\set{P_i}_{i = 1}^m$ obtained by truncating $P$ after the $i$th measurement command, 
  and adding back every correction command that depends on the first $i$ measurement variables, are deterministic.
\end{definition}

Flow structure gives sufficient (and sometimes necessary) conditions for a labelled open graph to be implementable by a deterministic measurement pattern.
It incorporates a time-ordering of the measurements and a function that indicates where to correct undesired measurement outcomes.
Gflow (or generalized flow) is a specific type of flow structure that ensures that the target linear map is an isometry for all choices of measurement angles.
\begin{definition}
  For a graph $G = (V, E)$ and a subset of its vertices $K \subseteq G$,
  let $\Odd(K) \coloneqq \set{u \in V : \abs{N(u) \cap K} \equiv 1  \!\mod 2}$ be the \emph{odd neighbourhood} of $K$ in $G$,
  where $N(u)$ is the set of neighbours of $u$.
\end{definition}
\begin{definition}[Gflow~\cite{browneGeneralizedFlowDeterminism2007}]
  \label{def:gflow}
  An open graph $(G, I, O)$ labelled with measurement planes $\lambda: \comp{O} \to \set{\XYm, \XZm, \YZm}$ has generalized flow (or gflow)
  if there exists a map $g: \comp{O} \to \cal{P}(\comp{I})$, where $\cal{P}$ is the power set function, and a strict partial order $<$ over $V$ such that for all $v \in \comp{O}$:
  \begin{enumerate}
    \item for all $w \in g(v)$ if $v \neq w$ then $v < w$
    \item for all $w \in \Odd(g(v))$ if $v \neq w$ then $v < w$
    \item $\lambda(v) = \XYm \, \implies \, v \notin g(v) \land v\in \Odd(g(v))$
    \item $\lambda(v) = \XZm \, \implies \, v \in g(v) \land v\in \Odd(g(v))$
    \item $\lambda(v) = \YZm \, \implies \, v \in g(v) \land v\notin \Odd(g(v))$
  \end{enumerate}
  The set $g(v)$ is called the \emph{correction set of $v$}.
\end{definition}

Extending the notion of gflow, \emph{Pauli flow} allows for the flow structure to take into account which vertices are measured in a Pauli basis..
In this setting, the function $\lambda$ defining measurement planes is of type $\lambda: \comp{O} \to \set{\XYm, \XZm, \YZm, \Xm, \Ym, \Zm}$,
while the function $\alpha$ is only defined for nodes $v \in G$ when $\lambda(v) \in \set{\XYm, \XZm, \YZm}$.
In other words, the pattern specifies vertices that are measured in the $X$, $Y$, or $Z$ basis.
For these specific measurements, the correction set is less restricted, and we obtain the conditions below.
\begin{definition}[Pauli flow~\cite{browneGeneralizedFlowDeterminism2007, simmonsRelatingMeasurementPatterns2021}]
  \label{def:pauli-flow}
  An open graph $(G, I, O)$ labelled with measurement planes $\lambda: \comp{O} \to \set{\XYm, \XZm, \YZm, \Xm, \Ym, \Zm}$ has Pauli flow
  if there exists a map $p: \comp{O} \to \cal{P}(\comp{I})$ and a strict partial order $<$ over $V$ such that:
  \begin{enumerate}
    \item for all $w \in p(v)$ if $\lambda(w) \notin \set{\Xm, \Ym} \land v \neq w$ then $v < w$
    \item for all $w \in \Odd(p(v))$ if $\lambda(w) \notin \set{\Ym, \Zm} \land v \neq w$ then $v < w$
    \item for all $w \leq v$ if $\lambda(w) = \Ym \land v \neq w$ then $(w \in p(v) \Longleftrightarrow w \in \Odd(p(v)))$
    \item $\lambda(v) = \XYm \, \implies \, v \notin p(v) \land v\in \Odd(p(v))$
    \item $\lambda(v) = \XZm \, \implies \, v \in p(v) \land v\in \Odd(p(v))$
    \item $\lambda(v) = \YZm \, \implies \, v \in p(v) \land v\notin \Odd(p(v))$
    \item $\lambda(v) = \Xm \, \implies \, v \in \Odd(p(v))$
    \item $\lambda(v) = \Zm \, \implies \, v \in p(v)$
    \item $\lambda(v) = \Ym \, \implies \, (v \notin p(v) \land v \in \Odd(p(v))) \lor (v \in p(v) \land v\notin \Odd(p(v)))$
  \end{enumerate}
\end{definition}

To understand the definition above, first note that for measurements in the planes $\set{\XYm, \XZm, \YZm}$, the conditions 
are the same as for gflow.
The above conditions $7-9$ are obtained by taking the pairwise disjunctions \enquote{$\lor$}
of conditions $4-6$, using the fact that each Pauli measurement belongs to a pair of planes.   
To obtain condition $1$, note that a Pauli $X$ error on a qubit measured in the $X$ basis only induces a global phase on the state.
Therefore, we must not correct $X$ errors on $X$ measurements.
Condition $2$ is the equivalent condition for $Z$ errors
and condition $3$ ensures that $Y$ measurements need only carry $Y = XZ$ corrections. 
A consequence is that $Y$ measurements in a graph with Pauli flow need not carry corrections, justifying conditions $1-2$.

We can now state the main result of~\cite{browneGeneralizedFlowDeterminism2007} which ensures that labelled open graphs with flow are implementable by deterministic patterns.
\begin{theorem}\cite{browneGeneralizedFlowDeterminism2007}\label{thm:flow-determinism}
  If a labelled open graph $\mathcal{M}$ has generalized flow, then the pattern defined by:
  \[\prod^<_i (X^{s_i}_{g(i) \cap \{j \vert i < j\}} Z^{s_i}_{\Odd(g(i))\cap \{j \vert i < j\}} M_i^{\lambda_i, \alpha_i, \cvar{s}_i}) E_G N_{\comp{I}}\]
  where $\prod^{<}$ denotes concatenation in the order $<$,
 is runnable, uniformly, strongly and stepwise deterministic and realises the target linear map $T(\mathcal{M})$, which is guaranteed to be an isometry.
\end{theorem}
The same result holds for Pauli flow where the measurement planes $\lambda_i$ can take the values $\set{\Xm, \Ym, \Zm}$.
However, a converse version of this theorem only holds for gflow, see~\cite{browneGeneralizedFlowDeterminism2007}.
The theorem indicates that $X$ corrections will be performed in $g(v) \backslash \set{v}$ and $Z$ corrections in $\Odd(g(v)) \backslash \set{v}$, for any qubit $v \in \comp{O}$.
Moreover, any qubit circuit can be turned into a labelled open graph satisfying the gflow conditions,
which ensures that MBQC can perform universal quantum computation~\cite{backensThereBackAgain2021}.

The above result can be expressed succinctly as an equation in $\mathbf{Channel}(\mathbf{ZX})$.
\begin{corollary}
  If a labelled open graph $(G, I, O, \lambda)$ has flow then, for any choice of measurement angles $\alpha$,
  the following equality holds in the CQ interpretation:
  \[
    \input{MBQC/determinism.tikz}
  \]
  where $\mathcal{E}$ is the set of feedforward equations $x_i =  \oplus_{j < i , \, i \in g(j)} \cvar{k_j}$ and $z_i =  \oplus_{j < i, \,i \in \Odd(g(j))} \cvar{k_j}$.
\end{corollary}

\section{Stream construction}\label{sec:stream}

Currently available optical setups are built using photon sources, optical routers and delay lines.
In order to represent these additional components, we need to add a time dimension to our diagrams, 
and allow classical and quantum feedback loops connecting different time-steps.

The stream processes introduced in this section can be seen as a generalisation of the notion of \emph{quantum comb}~\cite{chiribella_quantum_2008,chiribella_theoretical_2009}
to a recursive setting with infinitely many inputs and outputs. 
They enable formal descriptions of lattices, quantum optical setups and circuits studied in the context of quantum convolutional codes \cite{ollivier_description_2003,wilde_quantum-shift-register_2009}.
Our definition is based on the concept of intensional monoidal stream of~\cite{dilavoreMonoidalStreamsDataflow2022,lavore_coinductive_2025}.
The notion of observational equality for quantum protocols has already been studied in~\cite{caretteGraphicalLanguageDelayed2021},
but we here extend it with rewriting rules allowing us to reason by \emph{induction}. 
This is done while avoiding the regularity condition of~\cite[Section E]{caretteGraphicalLanguageDelayed2021} 
where processes are assumed to become constant after a finite number of time-steps.
Bridging theory and practice, we give multiple examples throughout the section of how this language can be used to reason about real-world experimental setups. 

\subsection{Recursive definition of stream processes}

We define a stream process recursively by what it does at time step zero, together with a stream describing what it does at future time steps.
We use letters $X, Y$ to denote infinite sequences $(X_0, X_1, X_2, \dots)$ of objects in a base graphical language $\bf{C}$,
and define $X \otimes Y$ as the sequence $(X_0 \otimes Y_0, X_1 \otimes Y_1, X_2 \otimes Y_2, \dots)$.
We use $X^\plus$ to denote the sequence obtained from $X$ by removing the head and by $\partial X$ the sequence $(I, X_0, X_1, \dots)$
obtained by adding the monoidal unit $I$ to $X$ as the head.
For $M_0$ an object of $\bf{C}$, we denote by $M_0 \cdot X$ the sequence $(M_0 \otimes X_0, X_1, X_2, \dots)$.

\begin{definition}[Stream]\label{def-stream}
  An intensional stream $\bf{f}: X \to Y$ in $\mathtt{Stream}(\bf{C})$ with \enquote{initial memory} $M_0$ is
  a process $f_0 : M_0 \otimes X_0 \to M_1 \otimes Y_0$ in $\bf{C}$ (called \enquote{now}) and a stream $\bf{f}^\plus : X^\plus \to Y^\plus$
  with initial memory $M_1$ (called \enquote{later}).
  \[
    \input{figures/stream-def.tikz}
  \]
  The wire labelled $M_0$ carries the initial state of the memory, $X_0$ and $Y_0$ are the input and output at time-step $0$,
  and $M_1$ is the memory created at time-step $0$ which serves as the initial memory for the stream $\bf{f}^\plus$.
\end{definition}

The recursive definition above defines the set of intensional streams with an initial memory as the final fixpoint~\cite{kozenPracticalCoinduction2017} of the following equation:
\[\mathtt{Stream}(\mathbf{C})(M_0 \cdot X, Y) = \sum_{M_1 \in \bf{C}} \mathbf{C}(M_0 \otimes X_0, M_1 \otimes Y_0) \times \mathtt{Stream}(\mathbf{C})(M_1 \cdot X^+, Y^+)\]
where $\sum$ denotes the disjoint union and $\times$ the cartesian product of sets.
As shown in~\cite{dilavoreMonoidalStreamsDataflow2022}, intensional streams are fully specified by their action at every time step and we have:
\[\mathbf{Stream}(\mathbf{C})(X, Y) \simeq \sum_{M \in \mathbf{C}^\mathbb{N}} \prod_{i \in \mathbb{N}} \mathbf{C}(M_{i} \otimes X_i, M_{i + 1} \otimes Y_i)\]
We say that two streams are \emph{intensionally equal} --- denoted by the equal sign $=$ --- if they have the same action at every time-step and are thus equal in the above set.

The simplest class of streams are \emph{constant streams} with no memory:
given any morphism $f: x \to y$ in $\bf{C}$ we obtain a stream $\bf{f}: X \to Y$ between constant objects $X=(x, x, \dots)$ and $Y= (y, y, \dots)$, 
with empty memory $M = I$, with $f_0 = f$ and $\bf{f}^\plus = \bf{f}$.
We denote the constant stream induced by a diagram $f$ simply by thickening its wires.
For example, the following constant stream defines the swap between any two constant objects:
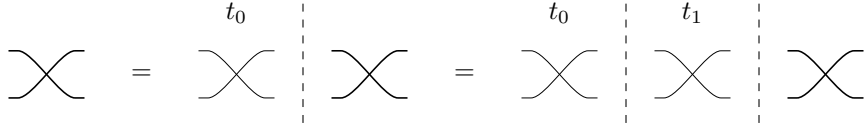
\[
  \input{figures/stream-constant.tikz}
\]
The equation above is read as a recursive definition: the swap stream is the \enquote{swap} now and itself later.
More generally we may consider the class of \emph{memoryless streams} where the memory type is the unit of the tensor $(I, I, \dots)$,
corresponding to sequences $\set{f_t: x_t \to y_t}_{t \in \mathbb{N}}$ in the base category.
For example, we can now define the swap operation between any two objects $X= (X_0, X_1, \dots)$ and $Y = (Y_0, Y_1, \dots)$
as the sequence $\set{\mathtt{swap}_t: X_t \otimes Y_t \to Y_t \otimes X_t}_{t \in \mathbb{N}}$.

Streams can be composed in sequence or in parallel, forming a symmetric premonoidal category.
\begin{proposition}\label{prop:stream-monoidal}
  Intensional streams over a base symmetric (pre)monoidal category $\mathbf{C}$ form a symmetric (pre)monoidal category $\mathtt{Stream}(\mathbf{C})$ where, for 
  streams $\bf{f}: X \to Y$, $\bf{g}: Y \to Z$ and $\bf{h}: A \to B$ with initial memory types $M^f_0$, $M^g_0$ and $M^h_0$:
  \begin{itemize}
    \item the \emph{sequential composition} $\bf{g} \circ \bf{f}: X \to Z$ is the stream with initial memory $M^f_0 \otimes M^g_0$, 
    acting \emph{now} as:
    \[\input{stream-sequential.tikz} \]
    and \emph{later} as {\normalfont $(\bf{g} \circ \bf{f})^\plus \coloneq \bf{g}^\plus \circ \bf{f}^\plus$} with initial memory $M^f_1 \otimes M^g_1$.
    \item the \emph{parallel composition} $\bf{f} \otimes \bf{h}: X \otimes A  \to Y \otimes B$ is the stream with initial memory $M^f_0 \otimes M^h_0$, 
    acting \emph{now} as:
    \[\input{stream-parallel.tikz} \]
    and \emph{later} as {\normalfont $(\bf{f} \otimes \bf{h})^\plus \coloneq \bf{f}^\plus \otimes \bf{h}^\plus$} with initial memory $M^f_1 \otimes M^h_1$.
  \end{itemize}
\end{proposition}
\begin{proof}
  The identity stream $X \to X$ is simply the memoryless stream $\set{\mathtt{id}_{X_t}: X_t \to X_t}_{t \in \mathbb{N}}$, and the symmetry is given by the swap stream defined above.
  It was shown in \cite[Section 7.1]{lavore_coinductive_2025} that the above composition is associative and unital with respect to the identity stream, 
  and that the above parallel composition is functorial in the monoidal case, i.e. it satisfies the interchange law $(f \circ g) \otimes (h \circ k) = (f \otimes h) \circ (g \otimes k)$. 
  These proofs are stated for a quotient of intensional streams, but the quotient condition is never used in the proofs.
  As shown in \cite[Theorem IV.9]{bonchi_effectful_2025}, the same results hold for premonoidal categories where the interchanger law is replaced by functoriality of \enquote{whiskering}, 
  i.e. $(\mathtt{id} \otimes f) \circ (\mathtt{id} \otimes g) = \mathtt{id} \otimes (f \circ g)$ and similarly for tensor product with $\mathtt{id}$ on the right.
\end{proof}

In order to link different time steps and model feedback of information, we need streams with a memory.
We can obtain these by taking the \emph{feedback} $\tt{fbk}_S(\bf{f}): X \to Y$ of a memoryless stream 
$\bf{f}: \partial S \otimes X \to S \otimes Y$.
This corresponds to adding $S$ to the memory of the stream by feeding back its output values to the inputs, as shown below.
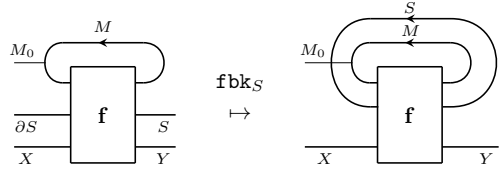
\[
  \scalebox{0.85}{\input{figures/stream-feedback.tikz}}
\]
We can use this to model \emph{delay}.
For example, the delay of length $1$ is defined as $\mathtt{delay}_X = \mathtt{fbk}_{X}(\mathtt{swap}) : X \to \partial X$, 
the feedback of $\mathtt{swap}: \partial X \otimes X \to X \otimes \partial X$:
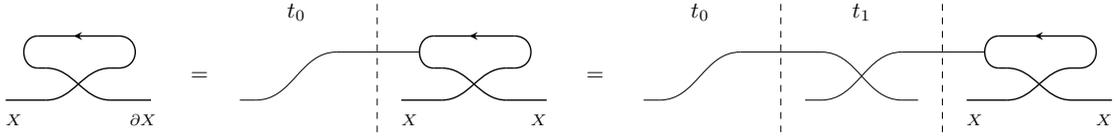
\[
  \scalebox{0.85}{\input{figures/stream-delay.tikz}}
\]

Streams represent infinite processes, but it is often useful to consider their execution for a finite number of time steps.
This is done by \emph{unrolling} the stream, giving us a family of functions $\mathtt{unroll}_n: \mathtt{Stream}(\bf{C}) \to \bf{C}$ parametrised by a natural number.

\begin{definition}[Unrolling]\label{def:unrolling}
  Given a stream $\bf{f} : X \to Y$ with memory $M$, the unrolling for $n$ time-steps of $\bf{f}$ is a process in $\bf{C}$ of the form:
  \[
    \mathtt{unroll}_n(\bf{f}) \colon M_0 \otimes X_0 \otimes \dots X_{n} \to Y_0 \otimes \dots Y_{n} \otimes M_{n + 1}
  \]
  defined by induction as follows:
  \begin{align*}
    \mathtt{unroll}_1(\mathbf{f}) = \; & \mathtt{swap}_{M_1, Y_0} \circ f_0 \\
    \mathtt{unroll}_n(\bf{f}) = \; & (\mathtt{id}_{Y_0} \otimes \mathtt{unroll}_{n - 1}(\bf{f}^\plus) ) \circ (\mathtt{unroll}_1(\mathbf{f}) \otimes \mathtt{id}_{Z})
  \end{align*} 
  where $Z = X_1 \otimes X_2 \otimes \dots \otimes X_{n - 1}$.
\end{definition}

For example, below are the three first unrollings of a generic stream.

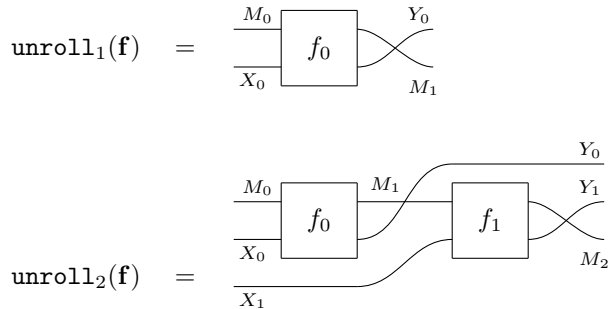
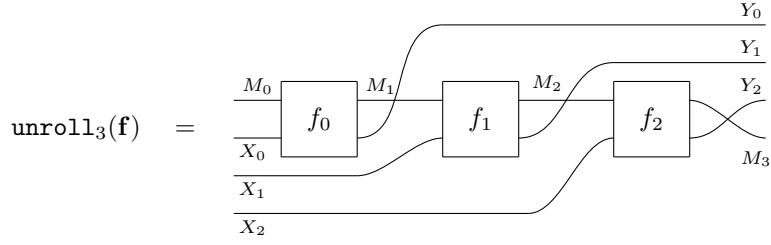
\begin{align*}
  \mathtt{unroll}_1(\bf{f}) \quad&= \quad \input{stream-unroll.tikz}\\
  \mathtt{unroll}_2(\bf{f}) \quad&= \quad \input{stream-unroll1.tikz}\\
  \mathtt{unroll}_3(\bf{f}) \quad&= \quad \input{stream-unroll2.tikz}
\end{align*}

We represent the unrolling graphically with a box surrounding the stream.
The following rules are the graphical equivalent of \cref{def:unrolling}, and allow us to reason inductively about finite protocol executions.

\begin{restatable}{rewrite}{unrollingStream}
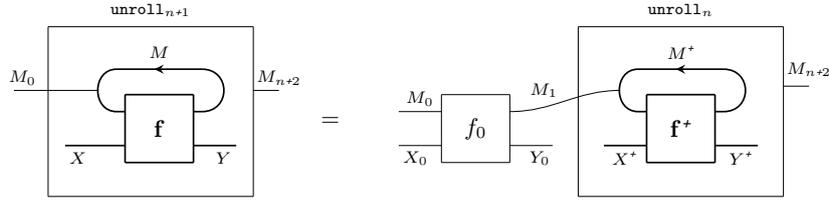
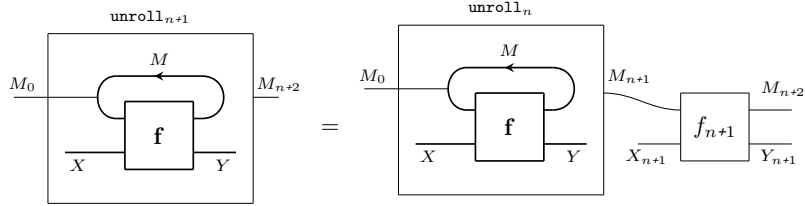
\label{rewrite:unroll-stream}
  \[
    \scalebox{0.9}{\input{unrolling.tikz}}
  \]
\end{restatable}

\begin{restatable}{rewrite}{unrollingFuture}\label{rewrite:unroll-future}
  \[
  \scalebox{0.9}{\input{unrolling-future.tikz}}
  \]
\end{restatable}

Note that these rules avoid the bureaucracy of ordering input and output wires which could be made more formal using the concept of \enquote{open diagrams} \cite{romanOpenDiagramsCoend2020}.

In streams we have access to a \enquote{followed by} operation, common in dataflow programming languages \cite{halbwachs_programming_1992}.
Given a process $r: M_0' \to M_0$ in $\bf{C}$ and a stream $\bf{f}: X \to Y$ with initial memory $M_0$, the 
stream $\mathtt{fby}(r, \bf{f})$ is given by $f_0 \circ (r \otimes \mathtt{id}_{X_0})$ \enquote{now} and $\bf{f}^+$ \enquote{later}.
It satisfies the following rule.

\begin{restatable}{rewrite}{followedBy}\label{rewrite:followed-by}
  For any $r: M_0' \to M_0$ and $\bf{f} : X \to Y$ with initial memory $M_0$, 
  \[
  \mathtt{unroll}_n(\mathtt{fby}(r, \bf{f})) = \mathtt{unroll}_n(\mathtt{fby}_{M_1}(\mathtt{id}_{M_1}, \bf{f})) \circ (r \otimes \mathtt{id}_{X_0})
  \]
\end{restatable}

This allows us to slide processes acting on the input memory outside of the unroll box, as in the following example.

\begin{example}
  Let $P: M \to M \in \mathbf{C}$ be a process such that $P^2 = \mathtt{id}_M$,
  then the streams $\bf{A} = \mathtt{fbk}_M(\bf{P})$ and 
  $\bf{B} = \mathtt{fby}(P, \mathtt{fbk}_M(\bf{I}))$ are intensionally different but $\mathtt{unroll}_n(\bf{A}) = \mathtt{unroll}_n(\bf{B})$ for all $n$.
  Indeed at time-step $0$ both streams act as $P$, then by induction we have:
  \begin{align*}
    \mathtt{unroll}_{n + 1}(\bf{P}) &= \mathtt{unroll}_n(\bf{P}) \circ P = \mathtt{unroll}_n(\bf{Q}) \circ P\\
    &= \mathtt{unroll}_n(\mathtt{fby}(P, \mathtt{fbk}_M(\bf{id}))) \circ P = \mathtt{unroll}_n(\mathtt{fbk}_M(\bf{id})) \circ P^2\\
    &= \mathtt{unroll}_n(\bf{I}) \circ P = \mathtt{unroll}_{n+1}(\bf{Q})
  \end{align*}
  using \cref{rewrite:unroll-stream} and \cref{rewrite:followed-by}.
\end{example}

Finally, we show the following using induction with \cref{rewrite:unroll-future}. 

\begin{lemma}\label{lemma:sliding}
  For any pair of sequences $r_t: M_t \to M_{t+1}$ and $f_t: M_{t+1} \otimes X_t \to M_{t+1} \otimes Y_t$, the following holds in $\bf{C}$:
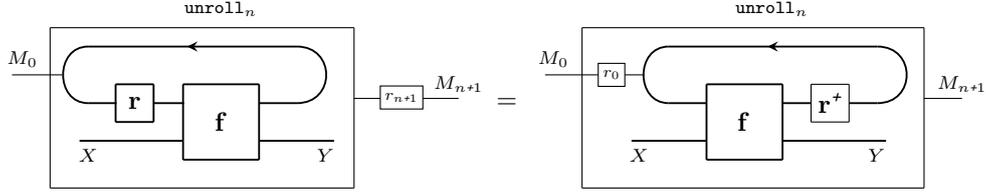
  \[\input{lemma-sliding.tikz}\]
\end{lemma}
\begin{proof}
  The statement for $n=0$ is easy to show, then we proceed by induction as follows:
  \[\scalebox{0.7}{\input{lemma-sliding-proof.tikz}}\]
  where the first step is \cref{rewrite:unroll-future}, the second step is induction and the last step uses both \cref{rewrite:followed-by} (twice) and \cref{rewrite:unroll-future}.
\end{proof}

\subsection{Interpretation as quantum protocols}

Constructing streams over our base language of quantum channels, we have access to the discarding map on every wire, as well as preparation and measurement commands taking values $x_t, \cvar{y_t}$ at each time step $t$.
The feedforward of classical information is again left implicit: the output variables at time step $t$ are input (and output) variables for time-step $t+1$.
This allows us to build annotations relating variables across different time-steps, while interpreting the unrolling direction as time.

\begin{definition}[Quantum protocol]
  A quantum protocol $\bf{f}: X \to Y \in \mathtt{Stream}(\mathtt{Channel}(\bf{C}))$ is stream of Kraus maps with a global set $\mathcal{E}$ of equations 
  of the form of \cref{def:annotations} on the set of input and output variables $x_t, \cvar{y_t}$, causally compatible with the order induced by $\cvar{y_t} < \cvar{y_{t + 1}}$.
\end{definition}

The discarding maps together with the unrolling operation allow us to define the notion of \emph{observational} equality from \cite{caretteGraphicalLanguageDelayed2021} for general quantum protocols.

\begin{definition}[Observational equality]
  We say that two quantum protocols $\bf{f}$ and $\bf{g}$ with the same initial memory $M_0$
  are observationally equivalent if for any $n \in \mathbb{N}$:
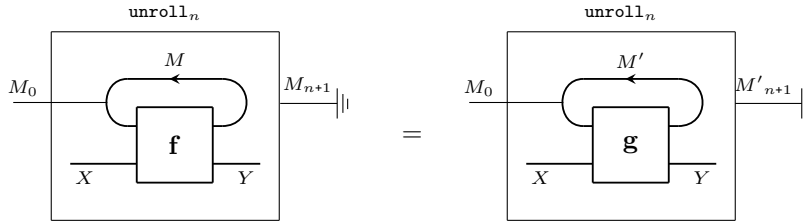
  \[
    \input{observational-equality.tikz}
  \]
  where $M_n$ and $M'_n$ denote the memory types of $\bf{f}$ and $\bf{g}$, respectively, at time step $n$.
\end{definition}

We use the symbol $\sim$ between streams for observational equivalence, and reserve the symbol $=$ for intensional equality.  
If two intensional streams are equal, then they are observationally equivalent; 
however, the converse isn't true in general.
The following \emph{sliding rule}, for example, equates streams with different memories. 
A proof is obtained by \cref{lemma:sliding} and using the fact that isometries commute with discards.

\begin{restatable}{rewrite}{slidingStream}\label{rewrite-sliding}
  For any sequence of isometries $r_t: M_t \to M_{t+1}$ and channels $f_t: M_{t+1} \otimes X_t \to M_{t+1} \otimes Y_t$, the following streams are observationally equal:
  \[\input{sliding.tikz}\]
\end{restatable}

\begin{remark}
  Similar sliding equations restricted to \emph{causal} maps $r_t$ appear in \cite{caretteGraphicalLanguageDelayed2021,bonchi_effectful_2025}. An important consequence of the $r_t$ being isometries is that the only scalar that can slide between different time steps is $1$, which ensures that every finite execution of the protocol has a well defined probability.
\end{remark}

We can now define causality and determinism for quantum protocols through their observational interpretation.

\begin{definition}[Causality]
  We say that a quantum protocol $\bf{f}$ is \emph{causal} if $\mathtt{unroll}_n(\bf{f})$ is causal for any $n \in \mathbb{N}$.
\end{definition}

\begin{definition}[Determinism]\label{def:determinism-stream} 
  We say that a quantum protocol $\bf{f}$ is \emph{deterministic} if every unrolling of $\bf{f}$ is deterministic in $\mathtt{Channel}(\bf{C})$.
  We say that it is uniformly deterministic if it is deterministic for every choice of sequence of angle parameters $\alpha_t$.
\end{definition}

The mixed sum $+$ of channels induces an infinite branching behaviour at the level of streams. 
We can however prove properties of this behaviour by induction over the unrolling, using the following rule.
\begin{restatable}{rewrite}{steamAddition}
  If $\bf{f} = \mathtt{fby}_M(A + B, \mathbf{g})$ then $\mathtt{unroll}_{n}(\bf{f}) = A \, \mathtt{unroll}_n(\bf{g}) + B \, \mathtt{unroll}_n(\bf{g})$
\end{restatable}
When branching occurs, we are interested in the probability of implementing a target channel $D$ in some time $n$, this motivates the following definition.

\begin{definition}[Probabilistic simulation]\label{def:stream-simulation}
  We say that a quantum protocol $\bf{f}$ implements a channel $D \in \mathtt{Channel}(\bf{C})$ in time $n$ with probability $p$, if 
  $\mathtt{unroll}_n(\bf{f})$ implements $D$ with probability $p$ (see \cref{def:implementation}).
  We say that $\bf{f}$ simulates a quantum protocol $\bf{g}$ with time-dependent probability $p(n)$ 
  if $\mathtt{unroll}_n(\bf{f})$ implements $\mathtt{unroll}_n(\bf{g})$ with probability $p(n)$ for all $n \in \mathbb{N}$.
\end{definition}

Universality then refers to the ability of a stream to implement any circuit in some set within a fixed probability.
We give two alternative definitions of universality. 
The strong notion requires that any circuit in the set can be implemented with probability arbitrarily close to $1$, 
while the weak notion allows for a constant probability of successful implementation, which may be more suited 
for intermediate-scale architectures.

\begin{definition}[Strong universality]\label{def:universality}
  We say that a parametrised quantum protocol $\bf{f}(\alpha)$ is strongly universal for a family of qubit circuits $\set{C_i}$ if 
  for any circuit $C_i$ and tolerance $\epsilon \in (0, 1)$ there exists an integer $n$ and parameter values $\alpha_i$ such that
  $\bf{f}(\alpha_i)$ implements $C_i$ in time $n$ with probability $p > 1 - \epsilon$.
\end{definition}

\begin{definition}[Weak universality]\label{def:weak-universality}
  We say that a parametrised quantum protocol $\bf{f}(\alpha)$ is weakly universal for a family of qubit circuits $\set{C_i}$ if 
  there exists an $\epsilon$ such that for any circuit $C_i$ there exists an integer $n$ and parameter values $\alpha_i$ such that
  $\bf{f}(\alpha_i)$ implements $C_i$ in time $n$ with probability $p > 1 - \epsilon$.
\end{definition}

We give examples of these two types of universality in \cref{thm:honeycomb} and in \cref{sec:universality}.

\subsection{Reasoning with streams of ZX diagrams}

Since any ZX diagram can be put in MBQC form, a general process in $\mathtt{Stream}(\bf{ZX})$ has the following form:
\begin{equation}\label{eq:zx-stream}
  \input{MBQC/MBQC-form-stream.tikz}
\end{equation}
for some sequence of open graphs $(G_t, I_t + M_t, O_t + M_{t + 1})$ with measurement planes $\lambda_t$ and angles $\alpha_t$.
Streams of ZX diagrams can thus be seen as \emph{foliations} of an infinite (open) graph state with specified measurement planes and angles
on internal nodes.

A key difference between $\mathtt{Stream}(\bf{LO})$ and $\mathtt{Stream}(\bf{ZX})$ is that the latter is \emph{compact closed}, that is,
it has Bell states and Bell effects on every object $X$.
These are induced from $\bf{ZX}$ by building the memoryless streams
$\mathtt{cup}_X = \set{\mathtt{cup}_{X_t}: X_t \otimes X_t \to I}_{t \in \mathbb{N}}$ and  $\mathtt{cap}_X = \set{\mathtt{cap}_{X_t}: I \to X_t \otimes X_t}_{t \in \mathbb{N}}$.
The snake equation lifts to $\mathtt{Stream}(\bf{ZX})$ as it holds at every time-step.
It allows us to relate feedback loops with delays. 

\begin{proposition}\cite{di_lavore_canonical_2021}\label{prop:zx-stream}
  Any stream in $\mathtt{Stream}(\bf{ZX})$ can be written as a composition of memoryless streams and the delay.
\end{proposition}
\begin{proof}
  This follows from the snake equation and the interchange law which holds by \cref{prop:stream-monoidal} since $\bf{ZX}$ is a symmetric monoidal category:
  \[\input{delay-feedback-proof.tikz}\]
\end{proof}
As a consequence, the slogan of the ZX calculus remains true after applying $\mathtt{Stream}$:
\begin{center}
  \enquote{Only connectivity matters.}
\end{center}
To see this, consider the sequence in \cref{eq:zx-stream} with a constant \enquote{generating graph} $G_t = G$. 
The resulting infinite graph has a periodic lattice structure, as illustrated in \cref{fig:lattices}.
In the figure, we use the following notation for delays on a single qubit initialised with $\ket{+}$ states:
\[
  \input{delay-ZX.tikz}
\]
and one can verify that only the connectivity of these delays to spiders in the diagram determines the generated infinite structure.

\begin{figure}[h]
  \scalebox{0.9}{
  \begin{minipage}{.5\textwidth}
    \centering
    \input{lattice-triangular.tikz}
    \caption*{Foliation of a triangular lattice on a cylinder, obtained by $\mathtt{unroll}_{11}$ with $d=3$.}
  \end{minipage}
  }\hfill
  \scalebox{0.9}{
  \begin{minipage}{.5\textwidth}
    \centering
    \input{lattice-rectangular.tikz}
    \caption*{Foliation of a rectangular lattice on a cylinder, obtained by $\mathtt{unroll}_{8}$ with $d=3$.}
  \end{minipage}
  }\vfill
  \scalebox{0.9}{
  \begin{minipage}{.5\textwidth}
    \centering
    \input{lattice-honeycomb.tikz}
    \caption*{Foliation of a honeycomb lattice on a cylinder, obtained by $\mathtt{unroll}_{11}$ with $d=3$.}
  \end{minipage}
  }\hfill
  \scalebox{0.9}{
  \begin{minipage}{.5\textwidth}
    \centering
    \input{lattice-raussendorf.tikz}
    \caption*{Foliation of a Raussendorf lattice on a duocylinder, obtained by $\mathtt{unroll}_{3}$ with $d=2$.}
  \end{minipage}
  }
  \caption{Lattice foliations as streams of ZX diagrams. 
  The labels indicate the time step at which the recurring part of the graph is generated. 
  The unrolled graph states are depicted on the plane but they are actually embedded in cylinders, 
  since the presence of a delay of $1$ implies that time step $t$ is always connected to time step $t+1$.}
  \label{fig:lattices}
\end{figure}

Streams of ZX diagrams enable finite representations of infinite graph states and allow us to prove equivalences between these infinite objects by graphical rewrite rules.
To illustrate the power of our approach, we now prove that the honeycomb lattice with only $\YZm$ measurements is universal for qubit circuits.
Indeed, the triangular lattice is known to be universal for qubit circuits using only $\XZm$ measurements \cite{mhalla_graph_2012}, 
and we show that the honeycomb lattice with periodic $\Ym$ measurements simulates the triangular lattice.  
While this was already observed in \cite{nest_universal_2006}, our one-line formal proof requires only local rewrites on a finite diagram (instead of an argument on infinite graphs with \enquote{$\dots$}). 
It moreover allows us to compute the planar measurements required to simulate the $\XZm$ measurements on the triangular lattice, and we find that $\YZm$ measurements suffice.

\begin{theorem}\label{thm:honeycomb}
  The hexagonal lattice with only $\YZm$ measurements is strongly universal for qubit circuits.
  More precisely, the honeycomb lattice with $\YZm$ measurements can simulate a triangular lattice with $\XZm$ measurements.
\end{theorem}
\begin{proof}
  \cref{fig:lattices} gives the ZX representation of triangular and hexagonal lattices. 
  We only use two flow-preserving ZX rules: Z insertion and local complementation. 
  These hold for constant streams which means we can apply them in $\mathtt{Stream}(\bf{ZX})$ as long as they don't act on delayed lines. 
  We thus obtain the following proof:
  \[
    \scalebox{0.9}{\input{honeycomb-universality.tikz}}
  \]
  where we used the notation for delays initialised differently with a single qubit state $\ket{\psi}$:
  \[\input{propagate-paulis.tikz}\] 
  and the observational equality $\sim$ follows from \cref{rewrite-sliding}:
  \[\input{initialisation.tikz} \]
  These local Clifford operations result only in measurements in the $\YZm$ plane. 
  Note that both lattices can be flattened to the plane by applying $\Zm$ measurements (which are both in the $\XZm$ and $\YZm$ measurement planes).
  It was shown in \cite[Theorem 17]{mhalla_graph_2012}, that a faultless triangular lattice of size $4n \cdot 4k$ is sufficient to implement 
  any qubit circuit of width $n$ and depth $k$, generated by $CZ$, Hadamard and $\Zm$-phase gates.
  For fixed $d$, we thus obtain strong universality in the sense of \cref{def:universality} for circuits of fixed width or depth.
  Since the rewrite above holds for any $d$, we conclude that the (infinite, faultless) hexagonal lattice is universal for arbitrary qubit circuits.
\end{proof}

As a second application of our calculus, we can use the inductive rules and sliding to compute the commutation relations between ZX streams and infinite Pauli sequences.
To illustrate this, we study a classical example from the literature on convolutional codes \cite{ollivier_description_2003, wilde_quantum-shift-register_2009}.

\begin{example}[Infinite-depth CNOT gate]
  \label{example:inf-cnot}
  Consider the ladder of CNOT gates defined by the following stream:
  \[ \input{infinite-cnot.tikz}\]
  Given any sequence of Pauli $\Z$ gates, we have the following commutation relation, using \cref{rewrite-sliding}:
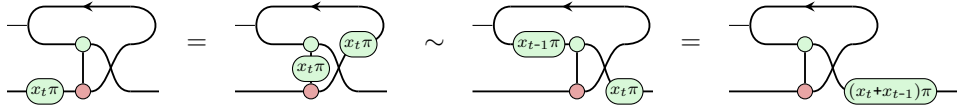
  \[ \input{infinite-cnot-Z.tikz}\]
  However, if we try to apply the same strategy to commute a sequence of Pauli $\X$ gates, we get stuck in a recursive loop:
  \[ \input{infinite-cnot-problem.tikz}\]
  In fact, the Pauli $\X$ gates accumulate on the output memory. 
  We can show the following equation about any finite unrolling of the stream:
  \[ \input{infinite-cnot-hypothesis.tikz}\]
  This is done by induction as follows:
  \[ \scalebox{0.8}{\input{infinite-cnot-X-proof.tikz}}\]
  After post-composing both sides with the discard map, we obtain the following observational equality:
  \[ \input{infinite-cnot-X.tikz}\]
\end{example}

\subsection{Routing and measurement modules}\label{subsec:LO-protocols}

Applying the $\mathtt{Stream}$ construction to linear optical circuit in $\mathtt{Channel}(\mathbf{LO})$ we can now construct 
linear optical protocols acting on photon time-bins, containing delays and feedback loops. 
These components are used throughout optical computing, see e.g. \cite{motes_scalable_2014, carolanScalableFeedbackControl2019, bartolucciSwitchNetworksPhotonic2021}.

The ability to prepare an empty optical mode gives rise to a useful class of memoryless processes in $\mathtt{Stream}(\bf{LO})$ called \emph{routers}.
For example, the binary oscillating router is defined by:
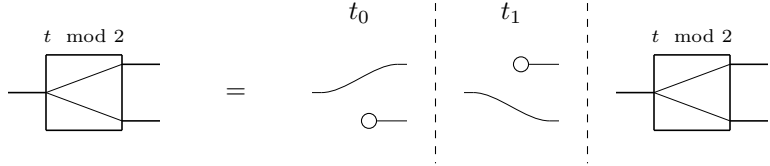
\[
  \input{figures/router-binary.tikz}
\]
And a similar 2 to 1 router is obtained using the discard map:
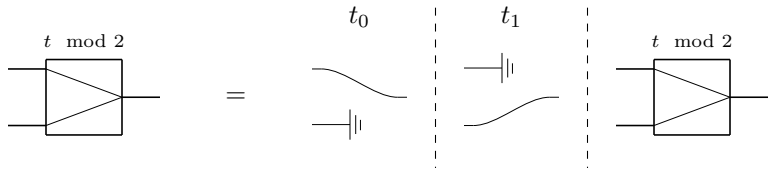
\[
  \input{figures/router-2-1.tikz}
\]
We assume that we have access to routers that can be controlled by a stream of classical variables $x_t$ with a predefined value at every time-step.
We may construct arbitrary routers from the binary router.
For example, the following setup implements an identity if $x_t = 0$ and a swap if $x_t = 1$.
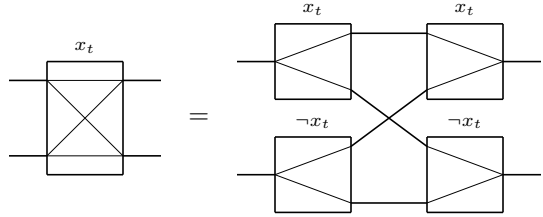
\[
  \input{figures/router-general.tikz}
\]
\begin{remark}
  Routers with more outputs than inputs can only be defined on \emph{optical modes} in $\mathtt{Stream}(\mathbf{ZXLO})$.
  This is because the qubit space $\mathbb{C}^2$ does not allow for the empty state.
\end{remark}
\begin{remark}
  We distinguish between routers and switches, defined in \cref{eq:switch}.
  The control parameters of a router  every time step are set before executing the program.
  With switches, the routing can be actively controlled by a measurement outcome or a classical variable computed at run-time.
  Such actively controlled switches are denoted as routers but with the control parameter drawn inside the box, as in \cref{eq:correction-module}.
\end{remark}

By composing delays and routers, we can route arbitrary permutations in time encoding.

\begin{lemma}\label{lemma:permutations}
  The following setup implements any permutation $\sigma$ of length $d$ in time encoding:
  \[
    \input{figures/router-permutation.tikz}
  \]
\end{lemma}
\begin{proof}
  Given a permutation $\sigma$, we set $x_t = t + \sigma(t) \mod 2d$ and $y_t = d + t \mod 2d$.
  A more efficient protocol for routing arbitrary permutation of size $d$ with $\mathtt{\log}(d)$ binary routers is given in~\cite{bartolucciSwitchNetworksPhotonic2021}.
\end{proof}

Working in $\mathtt{Stream}(\mathtt{Channel}(\bf{LO}))$, we are now ready to introduce the basic optical components used to perform universal MBQC: 
a measurement module and an active correction module.
For the measurement module, we need to be able to measure in the $\XYm$, $\XZm$, or $\YZm$ planes at different time-steps, typically fixed before running the experiment.
Given a choice of measurement plane $\lambda_t$ and angle $\alpha_t$ for each time-step $t$,
the measurement module $M_{\lambda, \alpha}$ is given by the following setup:
\begin{equation}\label{eq:measurement-module}
  \input{RUS/MeasurementModule.tikz}
\end{equation}
Finally, any sequence of Pauli corrections can be implemented using switches with classical control parameters $x_t, z_t \in \set{0, 1}$, 
as follows:
\begin{equation}\label{eq:correction-module}
  \input{RUS/CorrectionModule.tikz}
\end{equation}

Combining streams of linear optics and ZX diagrams, we can reason about optical protocols used in photonic quantum computing.
For example, the measurement and correction modules defined above with linear optical components, give rise to the expected streams of ZX diagrams.
\begin{lemma}\label{lemma-correct-measure}
  \[
    \input{RUS/MBQCModule.tikz}
  \]
\end{lemma}
\begin{proof}
  Since both the correction and measurement modules are memoryless streams,
  it is sufficient to prove that the equation above holds for any given time step $t$.
  In each case, the routers and switches define a specific path and the result follows from the equations of \cref{sec:bgd-dual-rail}.
\end{proof}

\subsection{Resource state generators}\label{subsec:RSG}

We now discuss different methods for the generation of photonic resource states (RSG) and how they can be modeled as streams of $\mathtt{Channel}(\bf{ZX})$ diagrams.
These can be broadly assigned to two classes --- (i) photonic and (ii) matter-based methods --- and we give two examples of the latter. 
These different procedures can in principle be used in conjunction.

\paragraph{Photonic RSG}

Linear optical methods begin with single photons, which are typically generated by spontaneous parametric down-conversion~\cite{harrisObservationTunableOptical1967}.
These photons are then entangled using linear optical Bell measurements or other heralded linear optical circuits~\cite{bartolucciCreationEntangledPhotonic2021}.
The advantage of this approach is that the resource states can in principle have arbitrary connectivity as photons
are not spatially or temporally restricted.
Moreover, photons of different resource states can be prepared with low distinguishability by active alignment of sources~\cite{carolanScalableFeedbackControl2019}.
Examples of photonic graph states used in the literature include the star graph~\cite{gimeno-segoviaThreePhotonGreenbergerHorneZeilingerStates2015, leeGraphtheoreticalOptimizationFusionbased2023},
rings~\cite{bartolucciFusionbasedQuantumComputation2023} and complete-like graphs~\cite{azumaAllphotonicQuantumRepeaters2015}.
The main drawback of these approaches is that, because of the fundamental limits of linear optics~\cite{stanisicGeneratingEntanglementLinear2017},
entanglement can only be generated probabilistically.
This drawback can be mitigated by using ancillary photons~\cite{bartolucciCreationEntangledPhotonic2021}
and \enquote{switch networks}~\cite{bartolucciSwitchNetworksPhotonic2021} to boost probabilities of success~\cite{ewertEfficientBellMeasurement2014}.

Assuming deterministic generation, we model photonic resource state generators of a graph state $\ket{G}$ as the constant stream over the underlying graph $G$ expressed as a ZX diagram:
\begin{equation}\label{def:RSG}
  \scalebox{0.75}{\input{stream-RSG.tikz}} \quad = \quad \scalebox{0.75}{\input{stream-bell.tikz}} \quad \text{ or } \quad \scalebox{0.75}{\input{stream-square.tikz}} \quad \text{ or } \quad \scalebox{0.75}{\input{stream-star.tikz}}
\end{equation}
where we omitted the normalisation scalar at every time step given by a constant stream $\bf{s}$ satisfying $\mathtt{unroll}_n(\bf{s}) = s^n$.

\begin{remark}
  Streams are \emph{synchronous} processes with a predefined global clock rate.
  It is sometimes useful to define streams acting at different clock rates, such as integer multiples of the global clock rate. 
  A possible application would be to formalise the switch networks of~\cite{bartolucciSwitchNetworksPhotonic2021} where multiple
  low-fidelity heralded resource states are actively routed to produce higher fidelity resource states in each time-bin.
\end{remark}

\paragraph{Emitter-based RSG}

Matter-based approaches for resource state generation rely on the emission of photons by excitation 
of a trapped ion~\cite{blinovObservationEntanglementSingle2004} or an artificial atom~\cite{ekimovQuantumSizeEffect1981, rossettiQuantumSizeEffects1983, murraySynthesisCharacterizationNearly1993, kastnerArtificialAtoms1993}.
Spin-based emitters, such as quantum dots, are able to keep the atom in a coherent superposition during the emission process.
They produce a stream of photons entangled with the atom and with each other~\cite{thomasEfficientGenerationEntangled2022}.
As such, they may be considered \enquote{photonic machine guns}~\cite{lindnerProposalPulsedOnDemand2009}, 
and represented accordingly:
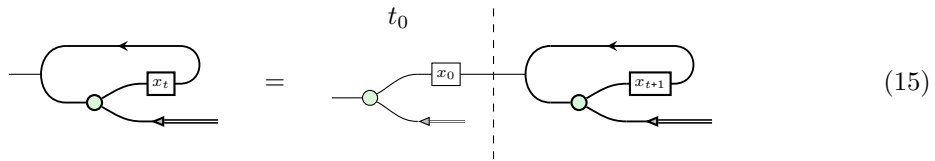
\begin{equation}\label{def:emitter}
  \input{figures/stream-emitter.tikz}
\end{equation}
The atom is the memory of the stream, entangled via a Z spider to the dual-rail states of the emitted photons.
At each time-step $t$ we may perform a single qubit unitary $x_t$ on the atom.
Special cases of interest are GHZ state generators when the white box is the identity,
linear clusters when it is a Hadamard gate and variable GHZ-linear clusters, or caterpillar states~\cite{huet_deterministic_2025},
when it can be programmed arbitrarily.
This technology has proved particularly effective for the generation of entangled photonic graph states~\cite{schwartzDeterministicGenerationCluster2016, costeHighrateEntanglementSemiconductor2023, coganDeterministicGenerationIndistinguishable2023}, and it has the advantage that resource states can be generated deterministically~\cite{schwartzDeterministicGenerationCluster2016, coganDeterministicGenerationIndistinguishable2023}.
Nevertheless, as photons are emitted one at a time, the resulting entanglement is restricted to \emph{linear} structure,
and photons need to be demultiplexed making them more susceptible to loss.
Photons emitted by non-identical atoms also suffer from distinguishability, although methods for mitigating this are being developed~\cite{yardOnchipQuantumInformation2022}.
A great advantage of emitters is that they can be used in repeat-until-success protocols, as shown in \cref{subsec:RUS},
that enables the near-deterministic implementation of entangling gates by fusion measurements~\cite{limRepeatUntilSuccessLinearOptics2005,degliniastySpinOpticalQuantumComputing2024}.

\paragraph{Ion-based RSG}

In the case of generation of photons from an ion trap, the emission process can be stimulated by a classically controlled laser pulse~\cite{nigmatullin_minimally_2016, main_distributed_2025}.
One obtains an ion-photon Bell pair, but this destructively measures the previous state of the ion which has to be prepared in a fixed state at each entanglement stage.
Ions with this property are more generally called \emph{network qubits}.
If we are further able to perform gates between network qubits and other more stable \emph{memory qubits}, 
we obtain a matter-based resource state generator of the form:
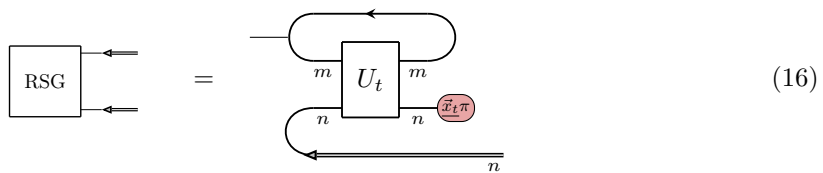
\begin{equation}\label{def:QPU}
  \scalebox{0.75}{\input{stream-RSG.tikz}} \quad = \quad \input{matter-QPU.tikz}
\end{equation}
where $U_t$ is a unitary on $m + n$ qubits, and a $k$-labeled wire is interpreted as the tensor product of $k$ single wires.
A general quantum processing unit (QPU), with both matter and photonic degrees of freedom can be represented as a variable 
stream with memory given by a tensor product of qubits and with photonic outputs at regular time intervals.

\paragraph{Simulations between RSGs}

The stream language which we have developed can be used to prove (bi)simulations between these different physical setups.

\begin{proposition}
  The quantum emitter can be simulated deterministically by a matter-based QPU and classical feedforward.
\end{proposition}
\begin{proof}
  We use a QPU with a single memory qubit entangled with a CNOT gate to the network qubit:
  \[\scalebox{0.9}{\input{emitter-QPU.tikz}}\]
\end{proof}

In order to simulate these matter-based memories with constant RSGs, we moreover need entangling measurements.

\begin{proposition}
  The quantum emitter can be simulated deterministically by a constant GHZ state, the delay, perfect Bell measurements and classical feedforward.
\end{proposition}
\begin{proof}
  This is shown by the derivation below which computes the feedforward instructions to correct the measurement-induced Pauli byproducts.
  \[
  \input{emitter-GHZ.tikz}
  \]
  The first step follows by \cref{prop:zx-stream} and the propagation of Pauli byproducts follows from \cref{rewrite-sliding}.
\end{proof}

Crucially, the above proposition relies on a perfect Bell measurement which is not available in linear optics. 
In the remainder of this paper we study the kind of entanglement afforded by linear optics and how it can be leveraged for photonic and distributed quantum computing.

\part{Application to distributed architectures}\label{part2}

\section{Characterization of correctable fusion measurements}\label{sec:characterization}

Fusion measurements are linear optical entangling measurements on dual-rail qubits.
They were originally introduced by Browne and Rudolf~\cite{browneResourceEfficientLinearOptical2005},
although several variations of this idea are present in the literature~\cite{limRepeatUntilSuccessLinearOptics2005,lee_nearly_2015}.
They are at the heart of recent \emph{fusion-based} proposals for performing quantum computation with photons~\cite{bartolucciFusionbasedQuantumComputation2023,  pankovichFlexibleEntangledState2023}, 
but are also used in \emph{distributed} quantum architectures to create entanglement between remote processors \cite{main_distributed_2025}.
Beyond \enquote{Type I} and \enquote{Type II fusion}~\cite{browneResourceEfficientLinearOptical2005},
linear optical circuits allow to construct different types of entangling gates on dual-rail qubits~\cite{limRepeatUntilSuccessLinearOptics2005, gimeno-segoviaThreePhotonGreenbergerHorneZeilingerStates2015, degliniastySpinOpticalQuantumComputing2024}.

In this section, we exploit the graphical language developed in \cref{part1} to obtain ZX representations of fusion measurements. 
We then generalise this concept by classifying all locally equivalent measurements whose Pauli byproducts \emph{can be corrected}.
Our notion of correctability requires that the non-determinism of measurement outcomes can be propagated as heralded Pauli errors in the input qubits.
The notion of planar fusion measurements introduced here (in the $\XYm$, $\XZm$ and $\YZm$ planes) allows to perform non-Clifford entangling gates 
in photonic circuits while reducing the total number of photons (compared to separate Type II fusions and non-Clifford single qubit measurements), see \cref{remark:XZfusion} for an example.
We also use it in \cref{subsec:emitter} to switch between $\Xm$ and $\Ym$ measurements with a single parameter shift.
The characterization results of this section are previewed in \cref{fig:zoo}.

\begin{figure}[h]
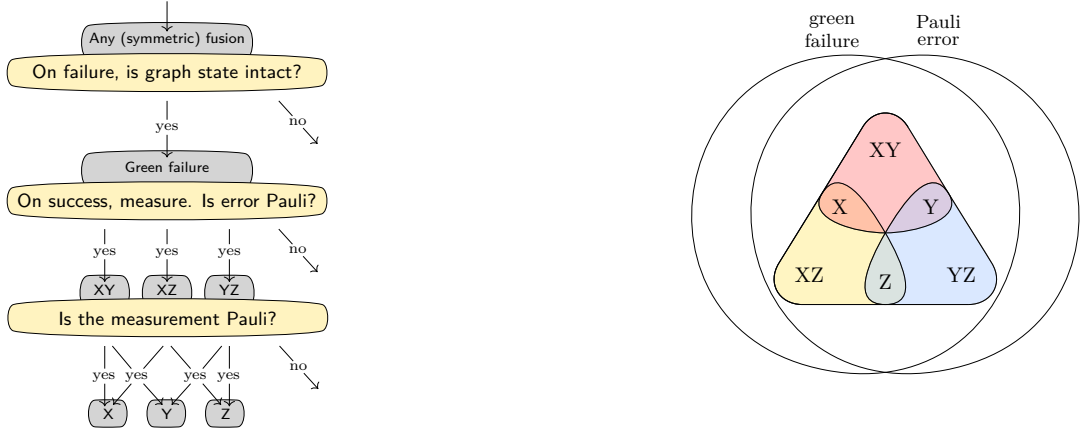

  \scalebox{0.8}{
  \begin{minipage}{.5\textwidth}
    \centering
    \input{char/flowchart.tikz}
  \end{minipage}
  }\hfill
  \scalebox{0.8}{
  \begin{minipage}{.5\textwidth}
    \centering
    \input{char/zoo.tikz}
  \end{minipage}
  }
  \label{fig:zoo}
  \caption{Overview of the results of this section. (a) Flowchart of the principles used to characterize correctability in fusion measurements (b) The zoo of correctable fusion measurements. Type II fusion is an instance of X fusion, CZ fusion is an instance of Y fusion, and two-qubit Z phase gadgets are a special case of YZ fusions.}
\end{figure}

\begin{remark}
  Ref.~\cite{loblTransformingGraphStates2024} also analyses the action of generalisations of fusion measurements on graph states.
  However, the generalisations of fusion they consider are local Clifford equivalent to Bell measurements, and they consider only the success case.
  Here instead, we consider all measurements local unitarily equivalent to Type I fusion followed by arbitrary single-qubit measurement.
  In particular, we obtain planar fusions and the $CZ$ fusion of~\cite{limRepeatUntilSuccessLinearOptics2005} that cannot be expressed in their setting.
\end{remark}

\subsection{Fusion measurements in ZX}\label{subsec:fusion-measurements}

In Ref~\cite{browneResourceEfficientLinearOptical2005}, fusions were expressed with quarter-wave plates and polarizing beam splitters, respectively, as follows:
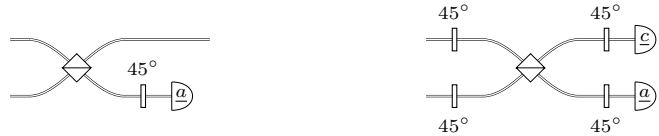
\[
  \begin{tikzpicture}
	\begin{pgfonlayer}{nodelayer}
		\node [style=none] (0) at (-1.25, 0.75) {};
		\node [style=none] (1) at (1.25, 0.75) {};
		\node [style=none] (2) at (1.25, -0.75) {};
		\node [style=none] (3) at (-1.25, -0.75) {};
		\node [style=none] (4) at (-1.75, -0.75) {};
		\node [style=none] (7) at (-1.25, 0.75) {};
		\node [style=none] (8) at (-1.75, 0.75) {};
		\node [style=none] (11) at (2.5, -0.75) {};
		\node [style=none] (12) at (1.25, -0.75) {};
		\node [style=beamSplitter, rotate=90] (13) at (1.75, -0.75) {};
		\node [style=sLabel] (14) at (1.75, 0) {$45^{\circ}$};
		\node [style=none] (15) at (3.5, 0.75) {};
		\node [style=none] (16) at (1.25, 0.75) {};
		\node [style=none] (19) at (0, -0.375) {};
		\node [style=none] (20) at (-0.375, 0) {};
		\node [style=none] (21) at (0, 0.375) {};
		\node [style=none] (22) at (0.375, 0) {};
		\node [style=none] (23) at (2.5, -0.375) {};
		\node [style=none] (24) at (2.5, -1.125) {};
		\node [style=sLabel] (27) at (2.75, -0.75) {$\cvar a$};
	\end{pgfonlayer}
	\begin{pgfonlayer}{edgelayer}
		\draw [style=DRWire, in=0, out=180, looseness=0.75] (2.center) to (0.center);
		\draw [style=DRWire] (3.center) to (4.center);
		\draw [style=DRWire, in=0, out=-180, looseness=0.75] (1.center) to (3.center);
		\draw [style=DRWire] (7.center) to (8.center);
		\draw [style=DRWire] (11.center) to (12.center);
		\draw [style=DRWire] (15.center) to (16.center);
		\draw [fill=white, style=LOWire] (22.center)
			 to (20.center)
			 to (19.center)
			 to cycle
			 to (21.center)
			 to (20.center);
		\draw [style=LOWire] (23.center) to (24.center);
		\draw [style=LOWire, bend left=90, looseness=2.50] (23.center) to (24.center);
	\end{pgfonlayer}
\end{tikzpicture}

  \qquad \qquad \qquad \qquad
  \begin{tikzpicture}
	\begin{pgfonlayer}{nodelayer}
		\node [style=none] (1) at (-1.25, -0.75) {};
		\node [style=none] (2) at (1.25, -0.75) {};
		\node [style=none] (3) at (1.25, 0.75) {};
		\node [style=none] (12) at (-1.25, 0.75) {};
		\node [style=none] (13) at (-2.75, 0.75) {};
		\node [style=beamSplitter, rotate=90] (14) at (-2, 0.75) {};
		\node [style=sLabel] (15) at (-2, 1.5) {$45^{\circ}$};
		\node [style=none] (16) at (-1.25, -0.75) {};
		\node [style=none] (17) at (-2.75, -0.75) {};
		\node [style=beamSplitter, rotate=90] (18) at (-2, -0.75) {};
		\node [style=sLabel] (19) at (-2, -1.5) {$45^{\circ}$};
		\node [style=none] (20) at (2.75, 0.75) {};
		\node [style=none] (21) at (1.25, 0.75) {};
		\node [style=beamSplitter, rotate=90] (22) at (2, 0.75) {};
		\node [style=sLabel] (23) at (2, 1.5) {$45^{\circ}$};
		\node [style=none] (24) at (2.75, -0.75) {};
		\node [style=none] (25) at (1.25, -0.75) {};
		\node [style=beamSplitter, rotate=90] (26) at (2, -0.75) {};
		\node [style=sLabel] (27) at (2, -1.5) {$45^{\circ}$};
		\node [style=none] (8) at (0, 0.375) {};
		\node [style=none] (9) at (-0.375, 0) {};
		\node [style=none] (10) at (0, -0.375) {};
		\node [style=none] (11) at (0.375, 0) {};
		\node [style=none] (28) at (2.75, 0.375) {};
		\node [style=none] (29) at (2.75, 1.125) {};
		\node [style=none] (30) at (2.75, -1.125) {};
		\node [style=none] (31) at (2.75, -0.375) {};
		\node [style=sLabel] (32) at (3, 0.75) {$\cvar c$};
		\node [style=sLabel] (33) at (3, -0.75) {$\cvar a$};
	\end{pgfonlayer}
	\begin{pgfonlayer}{edgelayer}
		\draw [style=DRWire, in=0, out=-180, looseness=0.75] (3.center) to (1.center);
		\draw [style=DRWire] (12.center) to (13.center);
		\draw [style=DRWire, in=0, out=180, looseness=0.75] (2.center) to (12.center);
		\draw [style=DRWire] (16.center) to (17.center);
		\draw [style=DRWire] (20.center) to (21.center);
		\draw [style=DRWire] (24.center) to (25.center);
		\draw [fill=white, style=LOWire] (11.center)
			 to (9.center)
			 to (8.center)
			 to cycle
			 to (10.center)
			 to (9.center);
		\draw [style=LOWire] (28.center) to (29.center);
		\draw [style=LOWire, bend right=90, looseness=2.50] (28.center) to (29.center);
		\draw [style=LOWire] (30.center) to (31.center);
		\draw [style=LOWire, bend right=90, looseness=2.50] (30.center) to (31.center);
	\end{pgfonlayer}
\end{tikzpicture}

\]
Note that Type I fusion is a partial measurement having two inputs and one output dual-rail mode,
while Type II is destructive and measures both qubits.
We can translate from polarization primitives to $\bf{LO}$ circuits using the following equalities:
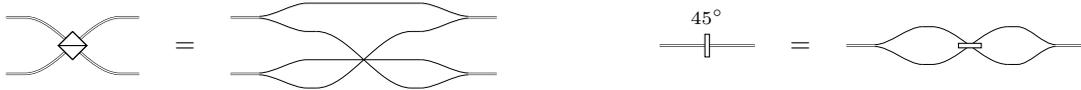
\[
  \begin{tikzpicture}
	\begin{pgfonlayer}{nodelayer}
		\node [style=none] (4) at (-1.25, 0.75) {};
		\node [style=none] (5) at (-1.25, -0.75) {};
		\node [style=none] (6) at (1.25, -0.75) {};
		\node [style=none] (7) at (1.25, 0.75) {};
		\node [style=none] (9) at (-1.75, 0.75) {};
		\node [style=none] (10) at (1.75, 0.75) {};
		\node [style=none] (11) at (1.75, -0.75) {};
		\node [style=none] (12) at (-1.75, -0.75) {};
		\node [style=none] (34) at (0, 0.375) {};
		\node [style=none] (35) at (-0.375, 0) {};
		\node [style=none] (36) at (0, -0.375) {};
		\node [style=none] (37) at (0.375, 0) {};
	\end{pgfonlayer}
	\begin{pgfonlayer}{edgelayer}
		\draw [style=DRWire] (10.center)
			 to (7.center)
			 to [in=0, out=-180, looseness=0.75] (5.center)
			 to (12.center);
		\draw [style=DRWire] (11.center)
			 to (6.center)
			 to [in=0, out=180, looseness=0.75] (4.center)
			 to (9.center);
		\draw [fill=white, style=LOWire] (37.center)
			 to (36.center)
			 to (35.center)
			 to cycle
			 to (34.center)
			 to (35.center);
		\draw [fill=white, style=LOWire] (37.center) to (35.center);
		\draw [fill=white, style=LOWire] (35.center) to (34.center);
		\draw [fill=white, style=LOWire] (34.center) to (37.center);
		\draw [fill=white, style=LOWire] (37.center) to (36.center);
		\draw [fill=white, style=LOWire] (36.center) to (35.center);
	\end{pgfonlayer}
\end{tikzpicture}

  \quad=\quad
  \begin{tikzpicture}
	\begin{pgfonlayer}{nodelayer}
		\node [style=none] (4) at (-1.25, 0.375) {};
		\node [style=none] (5) at (-1.25, -1.125) {};
		\node [style=none] (6) at (1.25, -1.125) {};
		\node [style=none] (7) at (1.25, 0.375) {};
		\node [style=none] (12) at (-2.75, 0.775) {};
		\node [style=none] (13) at (-2.75, 0.725) {};
		\node [style=none] (14) at (-2.75, 0.75) {};
		\node [style=none] (15) at (-1.5, 1.125) {};
		\node [style=none] (16) at (-1.5, 0.375) {};
		\node [style=none] (17) at (-2.75, -0.725) {};
		\node [style=none] (18) at (-2.75, -0.775) {};
		\node [style=none] (19) at (-2.75, -0.75) {};
		\node [style=none] (20) at (-1.5, -0.375) {};
		\node [style=none] (21) at (-1.5, -1.125) {};
		\node [style=none] (22) at (2.75, 0.775) {};
		\node [style=none] (23) at (2.75, 0.725) {};
		\node [style=none] (24) at (2.75, 0.75) {};
		\node [style=none] (25) at (1.5, 1.125) {};
		\node [style=none] (26) at (1.5, 0.375) {};
		\node [style=none] (27) at (2.75, -0.725) {};
		\node [style=none] (28) at (2.75, -0.775) {};
		\node [style=none] (29) at (2.75, -0.75) {};
		\node [style=none] (30) at (1.5, -0.375) {};
		\node [style=none] (31) at (1.5, -1.125) {};
		\node [style=none] (32) at (-3.5, 0.75) {};
		\node [style=none] (33) at (-3.5, -0.75) {};
		\node [style=none] (34) at (3.5, -0.75) {};
		\node [style=none] (35) at (3.5, 0.75) {};
	\end{pgfonlayer}
	\begin{pgfonlayer}{edgelayer}
		\draw [style=LOWire, in=0, out=-180, looseness=0.75] (7.center) to (5.center);
		\draw [style=LOWire, in=0, out=180, looseness=0.75] (6.center) to (4.center);
		\draw [style=LOWire, in=0, out=180, looseness=0.75] (15.center) to (12.center);
		\draw [style=LOWire, in=180, out=0, looseness=0.75] (13.center) to (16.center);
		\draw [style=LOWire] (4.center) to (16.center);
		\draw [style=LOWire, in=0, out=180, looseness=0.75] (20.center) to (17.center);
		\draw [style=LOWire, in=180, out=0, looseness=0.75] (18.center) to (21.center);
		\draw [style=LOWire, in=180, out=0, looseness=0.75] (25.center) to (22.center);
		\draw [style=LOWire, in=0, out=180, looseness=0.75] (23.center) to (26.center);
		\draw [style=LOWire, in=180, out=0, looseness=0.75] (30.center) to (27.center);
		\draw [style=LOWire, in=0, out=180, looseness=0.75] (28.center) to (31.center);
		\draw [style=LOWire] (20.center) to (30.center);
		\draw [style=LOWire] (26.center) to (7.center);
		\draw [style=LOWire] (5.center) to (21.center);
		\draw [style=LOWire] (31.center) to (6.center);
		\draw [style=LOWire] (15.center) to (25.center);
		\draw [style=DRWire] (35.center) to (24.center);
		\draw [style=DRWire] (34.center) to (29.center);
		\draw [style=DRWire] (32.center) to (14.center);
		\draw [style=DRWire] (33.center) to (19.center);
	\end{pgfonlayer}
\end{tikzpicture}

  \qquad \qquad \qquad
  \begin{tikzpicture}
	\begin{pgfonlayer}{nodelayer}
		\node [style=none] (0) at (1.25, 0) {};
		\node [style=none] (4) at (-1.25, 0) {};
		\node [style=beamSplitter, rotate=90] (25) at (0, 0) {};
		\node [style=sLabel] (35) at (0, 0.75) {$45^{\circ}$};
	\end{pgfonlayer}
	\begin{pgfonlayer}{edgelayer}
		\draw [style=DRWire] (0.center) to (4.center);
	\end{pgfonlayer}
\end{tikzpicture}

  \quad=\quad
  \begin{tikzpicture}
	\begin{pgfonlayer}{nodelayer}
		\node [style=none] (0) at (-1, 0.5) {};
		\node [style=none] (1) at (-1, -0.5) {};
		\node [style=none] (2) at (1, -0.5) {};
		\node [style=none] (3) at (1, 0.5) {};
		\node [style=beamSplitter] (4) at (0, 0) {};
		\node [style=none] (9) at (-2.5, 0.025) {};
		\node [style=none] (10) at (-2.5, -0.025) {};
		\node [style=none] (11) at (-2.5, 0) {};
		\node [style=none] (12) at (-1.25, 0.5) {};
		\node [style=none] (13) at (-1.25, -0.5) {};
		\node [style=none] (14) at (2.25, 0.025) {};
		\node [style=none] (15) at (2.25, -0.025) {};
		\node [style=none] (16) at (2.25, 0) {};
		\node [style=none] (17) at (1.25, 0.5) {};
		\node [style=none] (18) at (1.25, -0.5) {};
		\node [style=none] (19) at (-3.25, 0) {};
		\node [style=none] (20) at (3, 0) {};
	\end{pgfonlayer}
	\begin{pgfonlayer}{edgelayer}
		\draw [style=LOWire, in=0, out=-180, looseness=0.75] (3.center) to (1.center);
		\draw [style=LOWire, in=180, out=0, looseness=0.75] (0.center) to (2.center);
		\draw [style=LOWire, in=0, out=180, looseness=0.75] (12.center) to (9.center);
		\draw [style=LOWire, in=180, out=0, looseness=0.75] (10.center) to (13.center);
		\draw [style=LOWire, in=180, out=0, looseness=0.75] (17.center) to (14.center);
		\draw [style=LOWire, in=0, out=180, looseness=0.75] (15.center) to (18.center);
		\draw [style=DRWire] (20.center) to (16.center);
		\draw [style=DRWire] (19.center) to (11.center);
		\draw [style=LOWire] (0.center) to (12.center);
		\draw [style=LOWire] (1.center) to (13.center);
		\draw [style=LOWire] (17.center) to (3.center);
		\draw [style=LOWire] (18.center) to (2.center);
	\end{pgfonlayer}
\end{tikzpicture}

\]
Expressing the Type I and Type II fusions using this translation, we obtain the following diagrams, respectively:
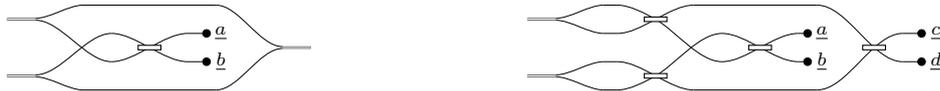
\[
  \begin{tikzpicture}
	\begin{pgfonlayer}{nodelayer}
		\node [style=none] (14) at (0.5, 0.375) {};
		\node [style=none] (15) at (0.5, -0.375) {};
		\node [style=beamSplitter] (16) at (-0.5, 0) {};
		\node [style=black node] (20) at (1, -0.375) {};
		\node [style=black node] (21) at (1, 0.375) {};
		\node [style=sLabel] (22) at (1.375, 0.375) {$\cvar a$};
		\node [style=sLabel] (23) at (1.375, -0.375) {$\cvar b$};
		\node [style=none] (34) at (-3.5, 0.775) {};
		\node [style=none] (35) at (-3.5, 0.725) {};
		\node [style=none] (36) at (-3.5, 0.75) {};
		\node [style=none] (37) at (-2.25, 1.125) {};
		\node [style=none] (38) at (-1.5, -0.375) {};
		\node [style=none] (39) at (3, 0.025) {};
		\node [style=none] (40) at (3, -0.025) {};
		\node [style=none] (41) at (3, 0) {};
		\node [style=none] (42) at (1.25, 1.125) {};
		\node [style=none] (43) at (1.25, -1.125) {};
		\node [style=none] (44) at (-4.25, 0.75) {};
		\node [style=none] (45) at (3.75, 0) {};
		\node [style=none] (46) at (-3.5, -0.725) {};
		\node [style=none] (47) at (-3.5, -0.775) {};
		\node [style=none] (48) at (-3.5, -0.75) {};
		\node [style=none] (49) at (-1.5, 0.375) {};
		\node [style=none] (50) at (-2.25, -1.125) {};
		\node [style=none] (51) at (-4.25, -0.75) {};
	\end{pgfonlayer}
	\begin{pgfonlayer}{edgelayer}
		\draw [style=LOWire] (15.center) to (20);
		\draw [style=LOWire] (14.center) to (21);
		\draw [style=LOWire, in=0, out=180, looseness=0.75] (37.center) to (34.center);
		\draw [style=LOWire, in=180, out=0, looseness=0.75] (35.center) to (38.center);
		\draw [style=LOWire, in=180, out=0, looseness=0.75] (42.center) to (39.center);
		\draw [style=LOWire, in=0, out=180, looseness=0.75] (40.center) to (43.center);
		\draw [style=DRWire] (45.center) to (41.center);
		\draw [style=DRWire] (44.center) to (36.center);
		\draw [style=LOWire, in=0, out=180, looseness=0.75] (49.center) to (46.center);
		\draw [style=LOWire, in=180, out=0, looseness=0.75] (47.center) to (50.center);
		\draw [style=DRWire] (51.center) to (48.center);
		\draw [style=LOWire, in=0, out=180, looseness=0.75] (15.center) to (49.center);
		\draw [style=LOWire, in=-180, out=0, looseness=0.75] (38.center) to (14.center);
		\draw [style=LOWire] (42.center) to (37.center);
		\draw [style=LOWire] (43.center) to (50.center);
	\end{pgfonlayer}
\end{tikzpicture}

  \qquad \qquad \qquad \qquad
  \begin{tikzpicture}
	\begin{pgfonlayer}{nodelayer}
		\node [style=none] (58) at (-0.5, -0.375) {};
		\node [style=none] (59) at (-0.5, 0.375) {};
		\node [style=none] (60) at (1.5, 0.375) {};
		\node [style=none] (61) at (1.5, -0.375) {};
		\node [style=beamSplitter] (62) at (0.5, 0) {};
		\node [style=black node] (63) at (1.75, -0.375) {};
		\node [style=black node] (64) at (1.75, 0.375) {};
		\node [style=sLabel] (65) at (2.125, 0.375) {$\cvar a$};
		\node [style=sLabel] (66) at (2.125, -0.375) {$\cvar b$};
		\node [style=none] (67) at (2, 1.125) {};
		\node [style=none] (68) at (2, -1.125) {};
		\node [style=none] (71) at (-3.25, 0.375) {};
		\node [style=none] (72) at (-3.25, 1.125) {};
		\node [style=none] (73) at (-1.25, 1.125) {};
		\node [style=none] (74) at (-2.25, 0.75) {};
		\node [style=beamSplitter] (75) at (-2.25, 0.75) {};
		\node [style=none] (76) at (-3.25, -1.125) {};
		\node [style=none] (77) at (-3.25, -0.375) {};
		\node [style=none] (78) at (-2.25, -0.75) {};
		\node [style=none] (79) at (-1.25, -1.125) {};
		\node [style=beamSplitter] (80) at (-2.25, -0.75) {};
		\node [style=beamSplitter] (83) at (3.5, 0) {};
		\node [style=none] (84) at (4.25, 0.375) {};
		\node [style=none] (85) at (4.25, -0.375) {};
		\node [style=black node] (86) at (4.75, -0.375) {};
		\node [style=black node] (87) at (4.75, 0.375) {};
		\node [style=sLabel] (88) at (5.125, -0.375) {$\cvar d$};
		\node [style=sLabel] (89) at (5.125, 0.375) {$\cvar c$};
		\node [style=none] (90) at (-4.875, 0.775) {};
		\node [style=none] (91) at (-4.875, 0.725) {};
		\node [style=none] (92) at (-4.875, 0.75) {};
		\node [style=none] (93) at (-3.625, 1.125) {};
		\node [style=none] (94) at (-3.625, 0.375) {};
		\node [style=none] (95) at (-5.625, 0.75) {};
		\node [style=none] (96) at (-4.875, -0.725) {};
		\node [style=none] (97) at (-4.875, -0.775) {};
		\node [style=none] (98) at (-4.875, -0.75) {};
		\node [style=none] (99) at (-3.625, -0.375) {};
		\node [style=none] (100) at (-3.625, -1.125) {};
		\node [style=none] (101) at (-5.625, -0.75) {};
	\end{pgfonlayer}
	\begin{pgfonlayer}{edgelayer}
		\draw [style=LOWire, in=0, out=180, looseness=0.75] (61.center) to (59.center);
		\draw [style=LOWire, in=-180, out=0, looseness=0.75] (58.center) to (60.center);
		\draw [style=LOWire] (61.center) to (63);
		\draw [style=LOWire] (60.center) to (64);
		\draw [style=LOWire, in=0, out=150, looseness=0.75] (74.center) to (72.center);
		\draw [style=LOWire, in=-180, out=0, looseness=0.75] (71.center) to (73.center);
		\draw [style=LOWire, in=0, out=180, looseness=0.75] (79.center) to (77.center);
		\draw [style=LOWire, in=-150, out=0, looseness=0.75] (76.center) to (78.center);
		\draw [style=LOWire] (67.center) to (73.center);
		\draw [style=LOWire] (68.center) to (79.center);
		\draw [style=LOWire] (85.center) to (86);
		\draw [style=LOWire] (84.center) to (87);
		\draw [style=LOWire, in=0, out=180, looseness=0.75] (93.center) to (90.center);
		\draw [style=LOWire, in=180, out=0, looseness=0.75] (91.center) to (94.center);
		\draw [style=DRWire] (95.center) to (92.center);
		\draw [style=LOWire, in=0, out=180, looseness=0.75] (99.center) to (96.center);
		\draw [style=LOWire, in=180, out=0, looseness=0.75] (97.center) to (100.center);
		\draw [style=DRWire] (101.center) to (98.center);
		\draw [style=LOWire] (72.center) to (93.center);
		\draw [style=LOWire] (71.center) to (94.center);
		\draw [style=LOWire] (77.center) to (99.center);
		\draw [style=LOWire] (76.center) to (100.center);
		\draw [style=LOWire, in=-30, out=180, looseness=0.75] (58.center) to (74.center);
		\draw [style=LOWire, in=180, out=30, looseness=0.75] (78.center) to (59.center);
		\draw [style=LOWire, in=0, out=180, looseness=0.75] (85.center) to (67.center);
		\draw [style=LOWire, in=-180, out=0, looseness=0.75] (68.center) to (84.center);
	\end{pgfonlayer}
\end{tikzpicture}

\]
This representation enables us to diagrammatically calculate the action of these circuits by representing them as a mixture of ZX diagrams.
Starting with Type I fusion, we get the following Kraus decomposition, proved in \cref{sec:mixedfusion}.
\begin{restatable}{proposition}{typeOneFusionProp} The following equation holds in the CP interpretation:
  \label{prop:type-I}
  \begin{gather*}
    \input{dualrail/Type1FusionFull.tikz}
  \end{gather*}
  after coarse-graining of the measurement operator by the equations $\cvar s = \cvar a \oplus \cvar b$ and $\cvar k = \cvar s \cvar b + \neg \cvar s (1 - \frac{\cvar a + \cvar b}{2})$.
  Here, $s$ is the Boolean value of success and $k$ is the Pauli measurement error.
\end{restatable}
This means that the error is $\cvar b$ in case of success and $1 - \frac{\cvar a + \cvar b}{2}$ in the failure case.
Note that, in case of failure, the pair of output modes is no longer in the qubit subspace defined in \cref{eq:dual-rail-encoding}.
Now, considering the Type II fusion, we see that this is just a Type I fusion
preceded by beam splitters and followed by a single qubit measurement in the Z-basis.
We can thus compute its action on the qubit subspace:
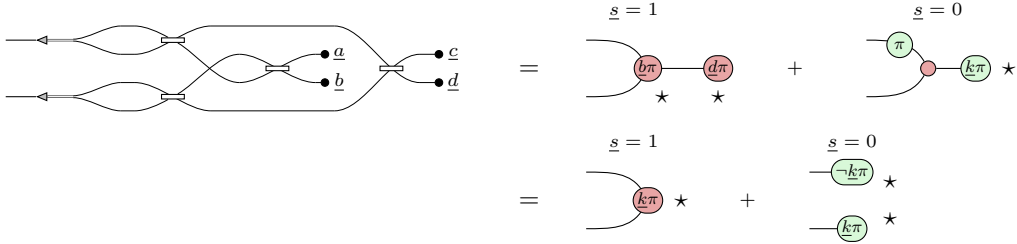
\[
  \input{dualrail/Type2FusionFull.tikz}
\]
where $\cvar s = \cvar a \oplus \cvar b$ is the Boolean value of success and $\cvar k = \cvar s (\cvar b + \cvar d) + \neg \cvar s (1 - \frac{\cvar a + \cvar b}{2})$ is the error.
This means that the error is $\cvar b + \cvar d$ in case of success and $1 - \frac{\cvar a + \cvar b}{2}$ in case of failure.
The scalars in the diagrams above are crucial for computing probabilities, but they can be disregarded in many cases. 
We use them in \cref{subsec:fusion-prob}, to compute the probability of success of fusion measurements for different input states.
However, in many other cases we can disregard them as in \cref{ex:XYpattern}.

As a first application, we consider a teleportation protocol where two parties share a dual-rail encoded photon Bell pair, and one of 
the parties performs a fusion measurement with a dual-rail encoded input qubit.
By \cref{prop:type-I}, the protocol can be rewritten as a distribution over
two causal maps:
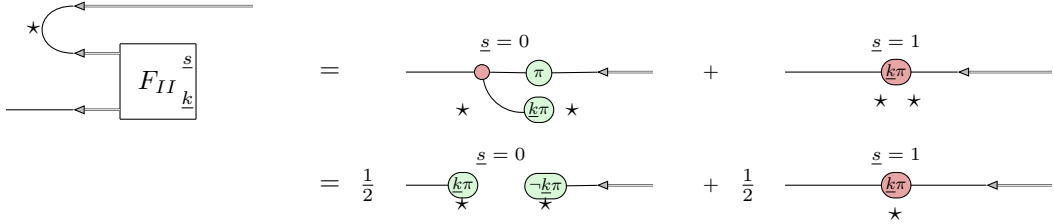
\[\input{teleportation-fusion.tikz}\]
If we apply a Pauli X correction conditioned on $\cvar{k}$ on the output qubit, this results in the identity channel with probability $\frac{1}{2}$, 
otherwise the input qubit is measured and the opposite state is prepared in the output.

\begin{remark}
  \label{remark:nType1}
  Both Type I and Type II fusions can be generalized to arbitrary number of inputs.
  A natural generalisation of Type II fusion gives the $n$-GHZ state analyser studied in~\cite{panGreenbergerHorneZeilingerstateAnalyzer1998, boseMultiparticleGeneralizationEntanglement1998, pankovichFlexibleEntangledState2023}
  in the success case.
  For the $3$-GHZ state analyser we obtain:
  \[
    \input{dualrail/Type2Fusion-ternary.tikz}
  \]
  where $\cvar{j} = r_2 + r_4 + r_6 \mod 2$, $\cvar{s} = r_1 + r_2 \mod 2$, $\cvar{s'} = r_3 + r_4 \mod 2$, 
  $\cvar{k} = 1 - \frac{r_{3} + r_{4}}{2}$ if $\cvar{s} = 1$ and $\cvar{k} = 1 - \frac{r_1 + r_2}{2}$ otherwise. 
\end{remark}

\subsection{General fusion measurements}

In \cref{subsec:fusion-measurements}, we saw that Type II fusion differs from Type I fusion only by single-qubit unitaries applied before and after the fusion, as well as an additional measurement.
Similarly, we can describe the success and failure outcomes of all possible entangling measurements unitarily equivalent to Type II fusion as follows:
\begin{equation}
  \label{eq:GF}
  \input{ZX/GeneralFusion.tikz}
\end{equation}
Here, $U_1$, $U_2$, and $U_3$ are arbitrary single-qubit unitaries which, up to some global phase, can be expressed by the Euler decomposition as three alternating rotations around the $Z$ and $X$ axes:
\begin{equation}
  \input{ZX/Euler.tikz}
  \label{eq:euler}
\end{equation}

Using this decomposition, we get a total of $9$ parameters to characterize a general fusion;
however, we are able to reduce this number using different observations that we support with calculations in ZX calculus.
First, one parameter can be eliminated from $U_3$ as it is followed by a single-qubit measurement that only contributes an irrelevant global phase:
\begin{equation*}
  \input{ZX/MeasureQubit.tikz}
\end{equation*}
Second, we observe that $Z(\gamma_1)$, $Z(\gamma_2)$, $Z(\alpha_3)$, and the fusion error $\cvar k \pi$ itself are simultaneously diagonalizable in the Z basis.
In other words, we can apply the spider fusion rule of the ZX calculus as follows:
\begin{equation*}
  \label{eq:Fusion6}
  \input{ZX/Fusion6.tikz}
\end{equation*}
where $\phi = \gamma_1 + \gamma_2 + \alpha_3$.

With this, we have reduced the number of parameters to describe an entangling measurement to $6$.
However, by only considering fusion with certain desirable properties, we can reduce this number even further.

\subsection{Green failure}

Consider the fusion given in \cref{eq:GF} the fusion with $U_1$, $U_2$, and $U_3$ all being the identity, corresponding to Type I fusion composed with a Z measurement.
An unsuccessful fusion in this case acts as a projector in the $Z$-basis.
This means that in addition to failing to fuse the two nodes, it also \emph{disconnects} them from their neighbours. For example:
\begin{equation*}
  \input{ZX/FusionFailDisconnect2.tikz}
\end{equation*}
In order to preserve the entanglement of the graph state, we want failures to be \enquote{green}.
\begin{definition}[Green failure]
  We say that a fusion measurement has green failure if its failure outcome satisfies:
  \begin{equation*}
    \input{ZX/GreenFailure.tikz}
  \end{equation*}
  for some $\theta_1, \theta_2$.
\end{definition}
\noindent This means that, upon failure, the underlying resource graph preserves its connectivity. For our running example:
\begin{equation*}
  \input{ZX/FusionFailGreen2.tikz}
\end{equation*}

We can characterize all types of fusions up to Z rotations on the nodes by the choices of $\beta_i$ in the Euler decomposition shown in \cref{eq:euler}.
We have the following measurement outcomes for different choices of $\beta_i$, with $i \in \{0,1\}$,
\begin{equation*}
  \input{ZX/FusionFailCases.tikz}
\end{equation*}
In other words, the failure is either green and keeps the connection of the graph, red and disconnects the graph, or it induces a non-unitary (and thus not correctable) error on the graph.
Asking for green failure reduces $U_1$ and $U_2$ to be of the following shape:
\begin{equation*}
  \label{eq:GreenU1U2}
  \input{ZX/GreenU1U2.tikz}
\end{equation*}
This reduces the $6$ parameters of a general fusion to $4$ if it has green failure:
\begin{proposition}
  \label{prop:green-failure}
  Any fusion measurement with green failure has the following form:
  \begin{equation}
    \input{ZX/Greencharacterisation.tikz}
  \end{equation}
  for some choice of angles $\alpha_1, \alpha_2, \beta, \phi \in [0, 2\pi)$ and measurement outcomes $\cvar{j}, \cvar{k} \in \set{0, 1}$.
\end{proposition}

Note that a similar characterization can be obtained for red fusion by reversing colours. 
However if the fusion is neither ``red'' nor ``green'' then we have shown that it induces 
a non-unitary and thus non-correctable error on the target qubits in the failure case.
We now analyse the success case and further refine our notion of correctability.

\subsection{Fusion with correctable Pauli error}\label{subsec:planar-fusions}

Since measurements in quantum computing are probabilistic processes, a fusion induces random errors in both its success and failure cases.
In MBQC, such errors can be propagated when the measurement pattern has flow.
To similarly propagate fusion measurement errors in our framework, they must be equivalent to local Pauli gates on the input qubits.
Diagrammatically, this means that the measurement errors $\cvar{k}\pi$ and $\cvar{j}\pi$ can be pulled out to the input wires while keeping them in X, Y or, Z basis.
\begin{definition}[Pauli error]
  \label{def-paulifusion}
  A fusion measurement has Pauli error when the success outcome satisfies:
  \begin{equation*}
    \input{ZX/PauliFusion.tikz}
  \end{equation*}
  for some bits $\cvar{w},\cvar{x},\cvar{y},\cvar{z} \in \set{0, 1}$.
\end{definition}
From the equation above, we deduce that either $U_1$ or $U_2$ must be Clifford, and that $U_3$ must be a gate locally equivalent to $H$, $S$, or $Id$ so that the measurement is in the \YZ, \XZ, or \XY plane, respectively.
Further requiring that failure is green gives us the following characterization.
\begin{restatable}{proposition}{characterization}
  \label{thm-characterization}
  Any fusion measurement with green failure and Pauli error has the following form:
  \begin{equation}
    \input{ZX/Paulicharacterisation.tikz}
  \end{equation}
  for a measurement plane $\lambda \in \set{\YZm, \XZm, \XYm}$, angles $\alpha, \omega \in [0, 2\pi)$, and a choice
  of Clifford parameter $d \in \set{0, 1, 2, 3}$.
\end{restatable}
\noindent The detailed proof is given in \cref{app:fusion}.

In practical applications, it is desirable that the action of a fusion measurement on its target qubits is symmetric,
so that errors can be propagated on either qubit at will.
\begin{definition}[Symmetric fusion]
  We say that a fusion measurement is symmetric if it is invariant under swap in the success case, that is,
  \begin{equation*}
    \input{fusion/SymmetricFusion.tikz}
  \end{equation*}
\end{definition}

\begin{theorem}
  \label{prop:characterization}
  Any symmetric fusion measurement with green failure and Pauli error has the following form:
  \begin{equation}
    \input{ZX/characterisation.tikz}
  \end{equation}
  where $c \in \set{0, 1}$, $\lambda \in \set{\YZm, \XZm, \XYm}$, and $\alpha \in [0, 2\pi)$.
\end{theorem}
\begin{proof}
  This directly follows from \cref{thm-characterization} proved in \cref{app:fusion}.
\end{proof}

\begin{definition}[Planar fusion]
  We call YZ, XZ, and XY fusion the three classes of symmetric fusion with green failure and Pauli error obtained by the choice of $\lambda$:\\
  \tabulinesep=3mm
  \begin{tabu}
    to \textwidth { X[$$c] | X[$$c] | X[$$c] }
    \lambda = \YZm             & \lambda = \XZm             & \lambda = \XYm             \\
    \input{fusion/Fusion-YZ.tikz} & \input{fusion/Fusion-XZ.tikz} & \input{fusion/Fusion-XY.tikz} \\
  \end{tabu}
\end{definition}

\begin{example}[Phase gadgets]
  A \emph{phase gadget} is an entangling non-Clifford gate that plays an important role in quantum circuit optimization~\cite{kissingerReducingNumberNonClifford2020, debeaudrapFastEffectiveTechniques2020, vandeweteringOptimalCompilationParametrised2024}
and quantum machine learning where they allow tuning the amount of entanglement between their inputs.
  A phase gadget is an instance of a YZ-fusion with $c = 0$ and $\alpha \in [0, 2\pi)$:
  \[
    \input{fusion/PhaseGadget.tikz}
  \]
\end{example}

\subsection{Stabilizer X and Y fusions}\label{sec:three-fusions}

The characterization that we obtained for fusions with green failure and Pauli error is three-fold, corresponding to
the three planes on the Bloch sphere.
We now consider measurements with the additional property of being Pauli measurements, in the X, Y or Z basis.

\begin{definition}[Stabilizer fusion]
  A fusion measurement is a stabilizer measurement if all its branches are stabilizer ZX diagrams, i.e.\@ only involve phases with multiples of $\frac{\pi}{2}$.
\end{definition}

Among symmetric green fusions, the characterization \cref{prop:characterization} together with the stabilizer assumption gives us only three possibilities
for whether the measurement in plane $\lambda$ is a Pauli $X$, $Y$ or $Z$ measurement.
Note first that the Z-fusion is trivial: it is a separable two-qubit measurement and leaves the connectivity
of the graph state unchanged.
\[
  \input{fusion/Fusion-Z-disconnects.tikz}
\]
X and Y fusions instead are entangling measurements that qualitatively change the connectivity of
the graph: they either fuse two nodes into one (X-fusion) or add a Hadamard edge between them (Y-fusion).

\begin{theorem}[Stabilizer green fusions]
  Up to local Clifford rotation on the target qubits, entangling stabilizer green fusions are either X or Y fusions.\\
  \begin{tabu}
    to \textwidth { X[c] | X[c] }
    \textnormal{X-Fusion}     & \textnormal{Y-Fusion}     \\
    \input{fusion/X-Fusion.tikz} & \input{fusion/Y-Fusion.tikz} \\
  \end{tabu}
\end{theorem}
\begin{proof}
  \noindent Suppose we write $\alpha = a \frac{\pi}{2}$ where $a \in \{0, 1, 2, 3\}$.
  Then from YZ-fusion we obtain $Z$-fusion when $a$ is even and $Y$-fusion when odd, from XZ-fusion we obtain $X$-fusion when $a$ is odd and $Z$-fusion when even, and from XY-fusion we obtain $X$-fusion when $a$ is even and $Y$-fusion when odd.
\end{proof}

The notion of stabilizer green fusion recovers the two main examples of fusion measurements used in the literature.

\begin{example}[Type II as X-fusion]
  The Type II fusion~\cite{browneResourceEfficientLinearOptical2005} is an instance of $X$-fusion with $c = 0$:
\begin{equation*}
  \input{fusion/XFusion.tikz}
\end{equation*}
Note that setting $c = 1$ is undesirable in this case as it only changes the errors from the X to the Y basis.
\end{example}

\begin{example}[CZ with Y-fusion]
  \label{ex:Y-fusion}
  The fusion measurement for performing CZ gates with linear optics, studied in~\cite{limRepeatUntilSuccessLinearOptics2005,degliniastySpinOpticalQuantumComputing2024},
  is an instance of $Y$-fusion.
  Indeed, up to Pauli errors, $Y$-fusion with $c = 1$ adds a Hadamard edge in the success case:
  \[
    \input{fusion/HFusion.tikz}
  \]
\end{example}

\begin{remark}
  Recall from \cref{remark:nType1} that Type II fusions can be generalized to an arbitrary number of input legs.
  Similarly, $Y$-fusion can also be generalized to any number of inputs.
  Its action corresponds to applying a CZ gate between each pair of qubits:
  \[
    \input{fusion/nHFusion.tikz}
  \]
  where the connections of the spiders form a complete graph on the right-hand side.
  This rewrite rule corresponds to toggling the CZ edges between all of the nodes being fused in the fusion of the underlying graph state; a formal proof can be derived from~\cite[Lemma 5.2.]{duncanGraphtheoreticSimplificationQuantum2020}.
\end{remark}

To conclude this section, we have carried out a thorough characterization of all fusion measurements that preserve graph states' connectivity on failure (\emph{green failure}) and induce correctable Pauli errors on success (\emph{Pauli error}).
Fusion measurements satisfying both these properties as well as being symmetric are XY, XZ, or YZ fusions.

\section{Flow structure for fusion networks}\label{sec:flow}
In the previous section we characterized fusion measurements that induce Pauli errors on their input qubits.
The aim of this section is to describe the flow structure that enables correction of these Pauli errors.
To do this, we extend the treatment of errors in standard MBQC to the fusion-based setting.
We study a notion of ``determinism on success'' for measurement patterns with additional destructive fusion measurements.
We define a notion of partially static flow for fusion networks which ensures errors can be propagated without applying corrections on fusion nodes.
We prove that any fusion network with partially static flow defines a fusion pattern that is deterministic on success.
Focusing then on fusion networks consisting of stabilizer X and Y fusions, 
we show that any decomposition of an open graph as an XY-fusion network has partially static flow, provided that the original graph has Pauli flow.

\begin{remark}
  The notion studied in this section only ensures determinism after post-selecting on fusion successes, it is focused at correcting errors 
  while reducing the feedforward required on fusion circuits.
  This assumption is justified in view of the repeat-until-success
  and boosting protocols studied in the following sections.
\end{remark}

\subsection{Fusion patterns}

In fusion-based quantum computing, photons may be prepared, entangled and measured by destructive fusions and single-qubit measurements.
We thus begin extending the definition of measurement patterns with the ability to perform \emph{correctable} fusion measurements.

\begin{definition}[Fusion pattern]
  A fusion pattern is a diagram in $\mathtt{Channel}(\mathbf{ZXLO})$ generated by the same commands as measurement patterns (\cref{def:meas_pattern})
  and additional two-qubit fusion commands defined according to \cref{thm-characterization} by:
  \begin{equation}
    \input{characterised-fusion.tikz}
  \end{equation}
  for some plane $\lambda \in \set{\XYm, \XZm, \YZm, X, Y, Z}$, angle $\alpha \in [0, 1)$ and Clifford parameter $c \in \set{0, 1}$.
\end{definition}

For a fusion pattern $P$ with $m$ measurement commands and $f$ fusion commands, the \emph{fusion branches} are the $2^f$ 
measurement patterns $\{P_{\vec{s}} \}_{\vec{s} \in \mathbb{B}^f}$ obtained by post selecting each success outcome according to $\vec{s}$, i.e.\@ $\cvar{s}_i = \vec{s}_i$.
This corresponds to replacing the $F_{\lambda, \alpha}$ command with the success or failure outcome, according to \cref{thm-characterization}.
The branch $P_{\vec{1}}$ where all fusions are successful is called the \emph{success branch}.

\begin{definition}[Determinism on success]
    We say that a fusion pattern $P$ is \emph{deterministic on success} if $P_{\vec{1}}$ is deterministic. Similarly, $P$ is uniformly, strongly or
    stepwise deterministic on success if $P_{\vec{1}}$ is uniformly, strongly or stepwise deterministic.
\end{definition}

\begin{example}
    \label{ex:XYpattern}
    As an example, consider the fusion pattern defined by the following diagram $P$ and its corresponding success pattern $P_{\vec{1}}$:
    \[
      \input{fpattern.tikz}
    \]
    where two stars have been cancelled by the scalars from the two entangling gates.
    The success pattern has $8$ branches obtained by setting the different values of $k, l, j \in \{0, 1 \}$.
    By rewriting the ZX diagram above, we can show that these $8$ branches are proportional to each other:
    \[
      \scalebox{.9}{\input{fpattern-rewrite.tikz}}
    \]
    where the last equation holds 
    Therefore, this specific pattern is deterministic on success.
    Moreover, each of the branches carries the same scalar, making the pattern strongly deterministic.
    Since the rewrite above holds for any angle $\alpha$, the pattern is also uniformly deterministic.
    By considering the pattern truncated at single-qubit measurement commands, a similar rewrite shows that it is also stepwise deterministic on success.
  \end{example}

\subsection{Fusion networks}\label{subsec:fusion-network}
 
Fusion networks represent configurations of single-qubit measurements and fusion measurements on a resource graph state.
Just as open graphs define the underlying topology of a measurement pattern,  fusion networks capture the underlying topology of fusion patterns.

\begin{definition}[Fusion network]
  A fusion network, denoted by $\mathcal{F} = (G, I, O, F, \lambda, \alpha, c)$, is given by the following:
  \begin{enumerate}
    \item an open graph $(G, I, O)$ (called \enquote{resource graph}),
    \item a set of fusions $F \sub \mathcal{M}(\comp{O} \times \comp{O})$,
    \item an assignment of measurement planes $\lambda_G: \comp{O} \to \set{\XYm, \XZm, \YZm, X, Y, Z, \ast}$ and $\lambda_F: F \to \set{\XYm, \XZm, \YZm, X, Y}$ (which we denote simply by $\lambda$ when the domain is unambiguous),
    \item an assignment of measurement angles $\alpha: \comp{O} + F \to [0, 2\pi)$ ($\alpha(v)$ is set to zero if $\lambda(v) \in \set{X, Y, Z}$), and
    \item a Clifford parameter for each measured qubit $c: \comp{O} \to \set{0, 1, 2, 3}$.
  \end{enumerate}
  where $\mathcal{M}$ denotes the multi-set construction, $+$ denotes the disjoint union, and $-$ denotes the set difference.
  An XY-fusion network is a fusion network where all fusions are either $X$ or $Y$ fusions.
\end{definition}

\begin{remark}
  The assignment of measurement planes to nodes in the resource graph $\lambda_G$ may return $\ast$ 
  which means that the node does not have a single-qubit measurement.
  We call these nodes ``unmeasured''.
  This makes the definition more flexible, as nodes may be part of a fusion and not of a single-qubit measurement.
  Moreover, the definition of fusion network given here references a single resource graph.
  In practice, the graph $G$ may be the disjoint union of multiple copies of the same basic resource state.
\end{remark}

Following \cref{prop:characterization}, a successful fusion has the effect of introducing an additional
node in the graph, measured in an arbitrary plane and angle.
Thus, any fusion network $\mathcal{F}$ defines a \emph{target open graph}, denoted $\mathcal{M_F}$, capturing the computation performed when we post select on fusion successes.

\begin{definition}[Target open graph]
  Given a fusion network $\mathcal{F} = (G, I, O, F, \lambda, \alpha, c)$ with $G = (V, E)$.
  Any fusion $f \in F$ contributes an extra vertex to the graph, labelled $v_f$.
  The \emph{target open graph} of $\mathcal{F}$ is $\mathcal{M_F} \coloneqq (G_\mathcal{F}, I, O, \lambda_\mathcal{F}, \alpha_\mathcal{F})$,
  where $G_\mathcal{F} = (V_\mathcal{F}, E_\mathcal{F})$,
  \begin{gather*}
      V_\mathcal{F} = V \cup \set{v_f | f \in F} \qquad
      E_\mathcal{F} = E \cup \set{ (v_f, w) | w \text{ belongs to } f \in F }\\
      \lambda_\mathcal{F}(u) =
    \begin{dcases}
      \lambda(f), & \text{if } u = v_f \text{ for some } f \in F \\
      X & \text{if } \lambda(u) = \ast \land c(u) \bmod 2 \equiv 0 \\
      Y, & \text{if } \lambda(u) = \ast \land c(u) \bmod 2 \equiv 1 \\
      \YZm, & \lambda(u) = \XZm \land c(u) \bmod 2 \equiv 1 \\
      \XZm, & \lambda(u) = \YZm \land c(u) \bmod 2 \equiv 1 \\
      \lambda(u), & \text{otherwise}
    \end{dcases}
    \qquad
    \alpha_\mathcal{F}(u) =
    \begin{dcases}
      \alpha(f), & \hspace*{-1.2cm}\text{if } u = v_f \text{ for some } f \in F \\
      \alpha(u) + \frac{c(u)\pi}{2}, & \lambda(u) = \XYm\\
      (-1)^{\left\lceil\! \frac{c(u)}{2}\! \right\rceil} \alpha(u), & \lambda(u) = XZ \\
      (-1)^{\left\lfloor\! \frac{c(u)}{2}\! \right\rfloor} \alpha(u), & \lambda(u) = YZ
    \end{dcases}
  \end{gather*}
  The \emph{target linear map} of the fusion network $T(\mathcal{F})$ is the target linear map of $\mathcal{M_F}$.
\end{definition}

In other words, nodes and edges of $G$ are extended with those coming from the set of fusions $F$.
The fusion measurement planes and angles are part of the new single-qubit measurement parameters.
We model ``unmeasured'' nodes in the resource graph (such that $\lambda(u) = \ast$) as measured in the $X$ basis.
Furthermore, some of the original single-qubit measurements are modified if their Clifford parameters are non-zero.
Clifford parameters on measured qubits correspond to changes in measurement planes.
We can capture these by the following equations:
\[
  \input{figures/fusion/clifford-parameter-change.tikz}
\]

\begin{example}\label{ex:network}
  \label{ex:fusion-network}
  Consider a fusion network with a pair of lines as the resource graph and two fusions.
  The target measurement graph is obtained by adding a new node in the graph for each fusion.
  \[
    \input{figures/network-graph.tikz}
    \qquad \longrightarrow \qquad
    \input{figures/network-target1.tikz}
  \]
\end{example}

\subsection{Partially static flow}

We now define a notion of flow for fusion networks that makes them deterministically implementable by a fusion pattern.
We assume that corrections can only be applied before a single-qubit measurement.
This ensures that fusion measurements are implemented by a \emph{static} setup.
Similarly, unmeasured nodes in the resource graph are modeled as $X$ measurements, but it is not possible to apply a correction on these nodes.
According to the definition of Pauli flow from \cref{def:pauli-flow},
$X$ corrections will be performed in $g(v) - \set{v}$ and $Z$ corrections in $\Odd(g(v)) - \set{v}$, for any non-output node $v$.
We thus obtain the following definition.

\begin{definition}[Partially static flow]
  \label{def:flow}
  A static flow for a fusion network $\mathcal{F} = (G, I, O, F, \lambda, \alpha)$ is a Pauli flow $(p, \leq)$ on the target open graph $\mathcal{M_F}$,
  such that no corrections need to be applied on fusion nodes.
  Concretely, for any node in the resource graph $u \in \comp{O}$ we must have:
  \begin{itemize}
    \item if $\lambda(u) = \ast$ then for any $v \in \mathcal{M_F}$, $ f\notin \Odd(p(v))$.
  \end{itemize}
  and for any fusion node $f \in F$ we must have:
  \begin{itemize}
    \item if $\lambda(f) = X$ then for any $v \in \mathcal{M_F}$, $ f\notin \Odd(p(v))$,
    \item if $\lambda(f) \in \set{\XYm, \XZm, \YZm}$ then for any $v \in \mathcal{M_F} - \set{f}$, $f\notin \Odd(p(v)) \cup p(v)$.
  \end{itemize}
\end{definition}

\begin{remark}
  Note that the condition above is precisely what is necessary to define a flow on the target open graph which does not require corrections on fused and unmeasured nodes.
  For $Y$-measured nodes this is already the case by conditions $1-3$ in \cref{def:pauli-flow}, so we do not need to impose additional conditions.
\end{remark}

Following~\cite{browneGeneralizedFlowDeterminism2007}, we prove that our notion of flow is both necessary and sufficient
for an XY-fusion pattern to be uniformly, strongly and stepwise deterministic on success.
Moreover, every such pattern can be factorized such that all fusions appear before single-qubit measurements.

\begin{theorem}\label{thm:flow-pattern}
  Given a fusion network $\mathcal{F}$ with static flow $(p, <)$ the fusion pattern defined by:
  \[
    \left( \prod^<_i X^{k_i}_{g(i)} Z^{k_i}_{\Odd(g(i))} M_i^{\lambda_i, \alpha_i, \cvar{k_i}}\right)
    \left(\prod_{f = (i, j) \in F} X^{k_f}_{g(f)} Z^{k_f}_{\Odd(g(f))} F_{ij}^{\lambda(f), \cvar{s_f}, \cvar{k_f}}\right)
    E_G N_{\comp{I}}
  \]
  is uniformly, strongly, and stepwise deterministic on success and implements the target linear map $T(\mathcal{F})$ when all fusions are successful.
  Here, $g(i) = p(i)\, \cap \{ j \, \vert \, i < j \}$, $\prod^{<}$ denotes concatenation in the order $<$ and
  $\prod$ denotes concatenation in any order.
\end{theorem}
\begin{proof}
  Since $\mathcal{M_F}$ has Pauli flow, we have a correction function $p: \comp{O} + F \to \mathcal{P}(\comp{O} + F)$
  satisfying the Pauli flow conditions.
  Since the errors in the success branches of $\mathcal{F}$ correspond exactly to the errors in $\mathcal{M_F}$,
  by~\cite[Theorem 4]{browneGeneralizedFlowDeterminism2007}, $\mathcal{F}$ is uniformly, strongly and stepwise deterministic on success
  and implements the target linear map $T(\mathcal{F}) = T(\mathcal{M_F})$.
  Moreover, since no corrections need to be applied on fusion nodes, the fusion nodes can be made initial in the order $<$.
  This gives us the factorisation required, where every fusion appears before single qubit measurements.
\end{proof}

\subsection{Decomposing open graphs as XY-fusion networks}\label{subsec:flow-preserving}

Suppose we wish to implement a given labelled open graph $\mathcal{G}$ as an XY-fusion network.
Several different fusion networks may exist that have $\mathcal{G}$ as their simplified target graph.
We now show that any such decomposition of $\mathcal{G}$ as a fusion network $\mathcal{F}$ is guaranteed to have static flow,
provided that $\mathcal{G}$ has Pauli flow. 
To prove this, we use rewrites that preserve the existence of Pauli flow~\cite{simmonsRelatingMeasurementPatterns2021,mcelvanneyCompleteFlowpreservingRewrite2023, mcelvanneyFlowpreservingZXcalculusRewrite2023}.

\begin{restatable}[X-fusion]{proposition}{xFusionFlowPreserving}\label{prop:Xflow-preserving}
  The following open graph rewrite preserves the existence of Pauli flow:
  \[
    \input{fusion/XFusionFlowPreservingGraph.tikz}
  \]
  where $\lambda(f) = \lambda(b) = X$, $\lambda(a) = \lambda(v_f)$, and $\alpha(a) = \alpha(v_f)$.
\end{restatable}

\begin{restatable}[Y-fusion]{proposition}{yFusionFlowPreserving}
  The following open graph rewrite preserves the existence of Pauli flow:
  \[
    \input{fusion/YFusionFlowPreservingGraph.tikz}
  \]
  where $\lambda(f) = Y$, $c(a) = c(b) = 0$ on the left and $c(a) = c(b) = 1$ on the right-hand side.
\end{restatable}
The proofs are in Appendix~\ref{app:flow}.

We can now show that Pauli flow on the target open graph $\mathcal{G_F}$ is both necessary and sufficient for $\mathcal{F}$ to have static flow.

\begin{theorem}\label{thm:flow-simplified-graph}
  An XY-fusion network $\mathcal{F}$ has static flow if and only if the simplified target graph $\mathcal{G_F}$ has Pauli flow.
\end{theorem}
\begin{proof}
  This follows from the two propositions above.
  For $X$-fusion we moreover need to show that when rewriting from $\mathcal{G_F}$ to $\mathcal{M_F}$ the newly introduced fusion node is not in the odd neighbourhood of some correction set.
  Using the notation of \cref{prop:Xflow-preserving},
  suppose that $v \in \Odd(g(u))$ for some node $u$ in $\mathcal{G_F}$ where $g$ is the Pauli flow on $\mathcal{G_F}$,
  then if $u$ is a neighbour of $v$, it is a neighbour of either $a$ or $b$ in $\mathcal{M_F}$ (and not a neighbour of both).
  Therefore, we can set $p(u) = g(u) - \{v\} + \{a, b\}$ as the correction function in $\mathcal{M_F}$ without changing the connectivity of
  $u$ to its correction set, and thus without violating the Pauli flow conditions for $u$.
  Then we have $f \notin \Odd(p(u))$, as required.
\end{proof}

\section{Universality in fusion-based architectures}\label{sec:universality}

Different architectures have recently been proposed to perform photonic and distributed quantum computing.
These include all-photonic approaches such as the FBQC proposal~\cite{bartolucciFusionbasedQuantumComputation2023}, 
and matter-based architectures such as spin-optical~\cite{degliniastySpinOpticalQuantumComputing2024} and ion-based~\cite{monroe_large_2014, main_distributed_2025} architectures.
The graphical framework developed in this paper allows us to represent, reason about and compare these architectures on an equal footing.

In this section, we study universality in distributed architectures, focusing on two examples:
\begin{enumerate}
  \item lattice-based photonic architectures based on a constant-size resource state generator and fusion measurements~\cite{bartolucciFusionbasedQuantumComputation2023}.
  \item reprogrammable architectures based on quantum emitters~\cite{degliniastySpinOpticalQuantumComputing2024, hilaire_enhanced_2024}
\end{enumerate}
We work in an idealised setting satisfying the following assumptions:
\begin{itemize}
  \item resource state generation is deterministic,
  \item photons are indistinguishable,
  \item all components are noiseless and have perfect efficiency.
\end{itemize}
While these assumptions are strong, a proof of universality in this setting needs to address 
both the probabilistic nature of fusion measurements and the correction of Pauli byproducts induced by undesired measurement outcomes. 
Relaxing any of the above assumptions gives an error model for these architectures 
which may be further analysed with the components of \cref{subsec:channel-examples}.

We start this section by proving the correctness of protocols that boost the probability of success of fusion measurements.
Exploiting these results we then provide two minimal examples of architectures --- lattice-based and reprogrammable --- 
giving graphical proofs of determinism and (weak) universality for each.
Compared to previous approaches~\cite{bartolucciFusionbasedQuantumComputation2023, degliniastySpinOpticalQuantumComputing2024} 
based on the generation of a universal lattice, our results show the potential for intermediate scale photonic quantum computing 
in programmable optical setups enhanced by compilation techniques.

\subsection{Boosting fusion with entangled resource states}\label{subsec:RUS}

Boosting the success probability of fusion measurements is an essential requirement for scaling FBQC\@.
Indeed the naive fusion measurement has a probability of success of only $\frac{1}{2}$ (see \cref{subsec:fusion-prob}).
Common approaches to boosting include the use of ancillary single photons~\cite{griceArbitrarilyCompleteBellstate2011, ewertEfficientBellMeasurement2014}, 
which however require exponentially many ancillas to arbitrarily approach probability $1$.
Here, we focus on boosting protocols that use entangled states as resource and reach probabilities arbitrarily close to $1$ with a linear number of attempts \cite{limRepeatUntilSuccessLinearOptics2005, lee_nearly_2015}.
Known methods apply only to $X$-fusion~\cite{lee_nearly_2015} and $Y$-fusion~\cite{limRepeatUntilSuccessLinearOptics2005}.
We give a formal graphical proof of their correctness and we generalise the protocol of~\cite{limRepeatUntilSuccessLinearOptics2005} to arbitrary fusions with green failure.

Let us consider the circuit of a non-destructive fusion measurement with green failure, parametrised by three phases $\theta_1$, $\theta_2$ and $\theta_3$.
\begin{equation}\label{fusion-green}
  \input{fusion-green.tikz}
\end{equation}
We define the \emph{passive fusion module}, implementing arbitrary correctable fusions with green failure, as the following stream:
\begin{equation}\label{eq:fusion-module}
  \input{RUS/FusionModule.tikz}
\end{equation}
This module produces two streams of classical outputs: the success values $\cvar{s_t}$ and the Pauli byproducts $\cvar{k_t}, \cvar{j_t}$.
We are now ready to state our results about repeat-until-success protocols. The proofs are given in \cref{app:RUS}.

\begin{definition}[RUS protocol]\label{def:RUS-protocol}
  The repeat-until-success fusion protocol is defined by the following setup:
  \[
    \input{repeat-until-success.tikz}
  \]
  Note that the control parameter of the switch takes the value $s_{t-1}$ of the previous success outcome.
\end{definition}

\begin{restatable}{theorem}{RUS}\label{thm:RUS}
  Any fusion with green failure can be boosted with a repeat-until-success protocol. More precisely, the following holds for $n \geq 1$:
  \[
  \scalebox{.8}{\input{repeat-boosted.tikz}}
  \]
  where $T$ is the time of the first successful fusion (if it exists) and:
  \[
    \cvar{\Sigma_t} \coloneqq \sum_{i = 0}^t \cvar{s_t}
    \quad \qquad
    c_{t} = c_{t - 1} \oplus (\neg \cvar{\Sigma_t}) k_t \oplus \cvar{\Sigma_{t - 1}} a_t
    \quad \qquad
    d_{t} = d_{t - 1} \oplus (\neg \cvar{\Sigma_t}) (\neg k_t) \oplus \cvar{\Sigma_{t - 1}} b_t
  \]
  with $s_{0} = 0$, $c_{0} = d_{0} = 1$.
\end{restatable}

As a consequence, the probability of success of a repeat-until-success fusion protocol after $n$ time-steps is $1 - \frac{1}{2^{n}}$.
We recover the RUS protocol of~\cite{limRepeatUntilSuccessLinearOptics2005} as the special case for $Y$-fusion.

\begin{restatable}[Y fusion RUS]{corollary}{YRUS}
  For $n \geq 1$ we have:
  \[
    \input{repeat-Y.tikz}
  \]
  where $z_t = (k_T \oplus j_T) \oplus c_t$ and $y_t = (k_T \oplus j_T) \oplus d_t$ if $T < t$ and $y_t = (k_T \oplus j_T) \oplus \neg c_t$ if $T=t$.
\end{restatable}

It was shown in~\cite{lee_nearly_2015} that $X$-fusion, corresponding to a Bell measurement in the success case,
can be boosted with entangled resource states.
As the $X$-fusion is an idempotent projection in the success case, and it has a correctable failure, 
we can boost it simply by repeating it on a GHZ resource state.
 
\begin{restatable}[Boosted X Fusion]{proposition}{XBoosted}
  For $n \geq 1$ we have:
  \[
    \input{boosted-X.tikz}
  \]
  where $c_{-1} = d_{-1} = 0$ and
  \[
    \cvar{c_{n + 1}} =
    \begin{cases}
      \cvar{c_{n}} & \text{ if } s_{n + 1} = 1 \\
      \cvar{c_{n}} + \cvar{k_{n + 1}} & \text{ if } s_{n + 1} = 0
    \end{cases}
    \qquad \qquad
    \cvar{d_{n + 1}} =
    \begin{cases}
      \cvar{d_{n}} & \text{ if } s_{n + 1} = 1 \\
      \cvar{d_{n}} + \neg \cvar{k_{n + 1}} & \text{ if } s_{n + 1} = 0
    \end{cases}
  \]
\end{restatable}

This has practical advantages compared to the RUS protocol in \cref{def:RUS-protocol}, as the above stream requires no feed forward.

\begin{corollary}\label{prop-boostedX}
  For $n \geq 1$ we have:
  \[
    \input{boosted-X-nooutput.tikz}
  \]
  where $e_n = c_n + d_n$.
\end{corollary}

\subsection{Lattice-based architecture}

Lattice-based architectures for universal quantum computing~\cite{gimeno-segoviaThreePhotonGreenbergerHorneZeilingerStates2015, bartolucciFusionbasedQuantumComputation2023} are structured as follows:
\begin{enumerate}
  \item a resource state generator (see \cref{subsec:RSG}) is used to construct a small constant-size entangled state at every time-step,
  \item the photons undergo a bounded-depth circuit built from linear optics and delays,
  \item photons from different time-bins are either fused or directly measured.
\end{enumerate}
To illustrate these architectures we focus on a minimal example involving a \emph{triangle resource state} and $\XZm$ measurements, inspired by~\cite{mhalla_graph_2012}.
Other examples of universal lattice-based architectures can be obtained by replicating the argument below for the streams in \cref{fig:lattices}.
We define our triangle architecture as the following protocol:
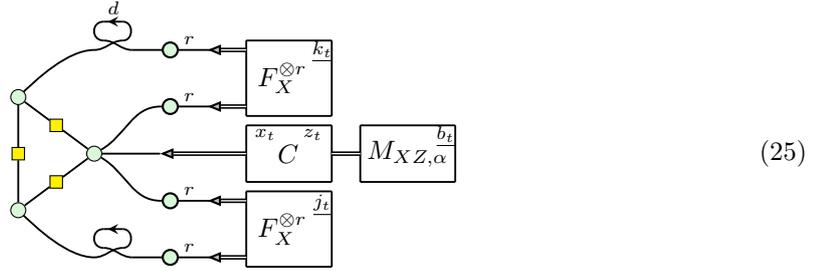
\begin{equation}\label{eq:photonic-architecture}
  \input{figures/triangle-architecture.tikz}
\end{equation}
where the delays are initialised with $\ket{+}$ states and the delay labelled $d$ is the $d$-fold composition of such delays (outputting $\ket{+}$ for the first $d$ time steps).
Using \cref{prop-boostedX} and post-selecting on fusion successes, we obtain a diagram in $\mathtt{Stream}(\mathtt{Channel}(\bf{ZX}))$ which 
we can simplify as follows:
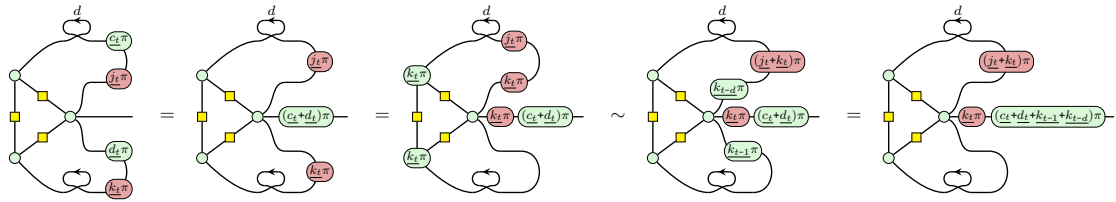
\[
  \scalebox{0.73}{\input{figures/triangle-architecture-rewrite-new.tikz}}
\]
Then, we show that:
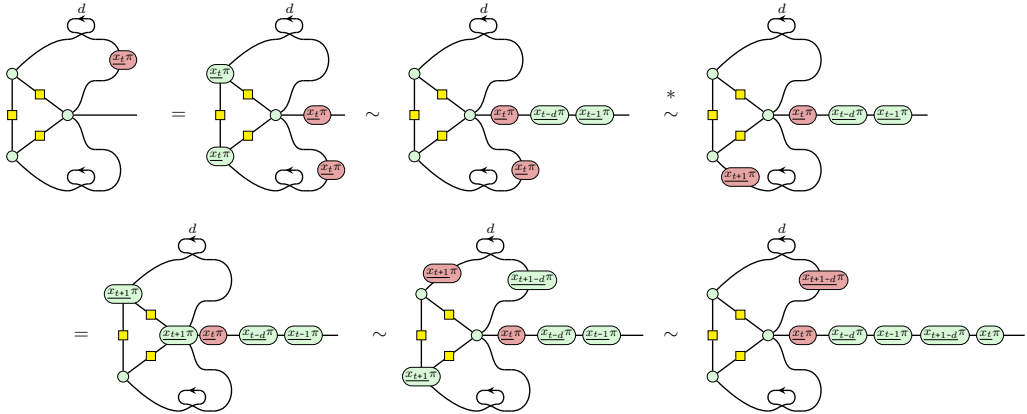
\[
  \scalebox{0.73}{\input{figures/triangle-architecture-rewrite-new-support.tikz}}
\]
where $(\ast)$ uses \cref{rewrite-sliding} and the fact that the delay is initialised with a $\ket{+}$ state satisfying $X\ket{+} = \ket{+}$.
Applying the above equality recursively, similarly to \cref{example:inf-cnot}, we obtain the observational equality:
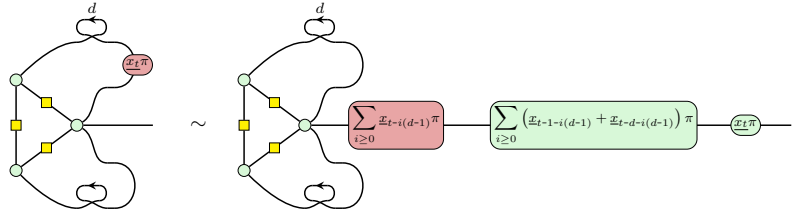
\[
  \scalebox{0.8}{\input{figures/triangle-architecture-rewrite-new-support-lem.tikz}}
\]
where $x_{t} \coloneqq 0$ for $t < 0$ so that the sums are over finitely many previous measurement outcomes.
Putting this all together, we need to apply $\X$ correction and $\Z$ correction summing the above.
Moreover, applying \cref{def:unrolling} to the resulting stream builds a triangular lattice where each node is measured in the $\XZm$ plane, see \cref{fig:lattices} and \cref{thm:honeycomb}.
As shown in~\cite{mhalla_graph_2012}, the triangular lattice with $\XZm$ measurements is universal for qubit circuits.
Therefore, we see that the protocol in \cref{eq:photonic-architecture} simulates a universal architecture with decreasing probability $(1 - \frac{1}{2^r})^{2t}$, in the sense of \cref{def:stream-simulation}.
This is however not sufficient to prove strong universality of the protocol, in the sense of \cref{def:universality}, since the parameter $r$ is constant in the setup.
Instead, we obtain a weak universality result, in the sense of \cref{def:weak-universality}, as follows.

\begin{theorem}\label{thm:photonic-universality}
  The architecture in \cref{eq:photonic-architecture} is weakly universal for qubit circuits generated by $\set{CZ, H, Z(\alpha)}$ of width $\frac{d - 1}{4}$ and depth bounded by a constant $k$.
\end{theorem}
\begin{proof}
  Note first that, according to~\cite{mhalla_graph_2012}, a triangular lattice of size $4 n \cdot 4k$ is required to compile a circuit of width $n$ and depth $k$.
  Thus, given a circuit as above, we need at most $4 d k$ time-steps to build a lattice of the correct size, using one time step every $d$ to flatten the lattice to the plane by applying $\Zm$ measurements.
  Moreover, by the rewrite above, all the Pauli byproducts can be corrected by classical feedforward and the execution of the protocol is deterministic on success.
  The lattice is obtained when all fusions succeed, with probability $(1 - \frac{1}{2^r})^{8dk}$. Therefore we can set $\epsilon \geq 1 - (1 - \frac{1}{2^r})^{8dk}$ to obtain a simulation of all the circuits in the set within probability $1 - \epsilon$.
\end{proof}

Stronger universality results may be obtained using percolation methods to avoid post-selecting on fusion successes.
In fact, the percolation threshold of the triangular lattice is $\frac{1}{2}$~\cite{malarz_site_2020}, which is reached when $r=2$.
While this ensures the existence of a path from the input to the output boundary of the lattice, it does not give us information about the size of this connected component and thus, 
given a circuit $C$, we cannot ensure that the connected component will be large enough to accommodate $C$.
Renormalisation ideas~\cite{herr_local_2018} may offer a way to show this rigorously, but we leave their formalisation in the ZX calculus to future work.

\begin{remark}\label{remark:XZfusion}
  The architecture of \cref{eq:photonic-architecture} can be simplified by exploiting the planar fusion measurements introduced in \cref{subsec:planar-fusions}.
Instead of performing
  two $X$-fusion measurements and an $\XZm$ single-qubit measurement, we could perform a single ternary $\XZm$-fusion measurement:
  \[
    \input{triangle-architecture-simplified.tikz}
  \]
  where the arrow denotes simulation according to \cref{def:stream-simulation} and we omitted the Pauli byproducts for simplicity.
  A RUS protocol for this ternary fusion can be obtained by a simple generalisation of \cref{def:RUS-protocol} although it requires more active switching compared to the static boosted $X$-fusion.
\end{remark}

\begin{remark}
  The architecture studied in this section relies on the deterministic generation of the triangle graph state. 
  For an all-photonic approach this could be achieved using, for example, switch networks \cite{bartolucciSwitchNetworksPhotonic2021}, although a convincing experimental demonstration has not been achieved to date.
  The other main source of noise in an all-photonic setup is photon loss, which can be heralded but will further impact the probability of successfully generating the lattice, and thus the 
  tolerance $\epsilon$ used in \cref{thm:photonic-universality}. The triangle resource state could also be constructed using other resource-state generation techniques such as quantum emitters.
\end{remark}

\subsection{Reprogrammable architecture}\label{subsec:emitter}

Quantum computing architectures based on quantum emitters~\cite{paesaniHighThresholdQuantumComputing2023,degliniastySpinOpticalQuantumComputing2024} 
can leverage variable-length resource states allowing constructive forms of universality where a given quantum computation is directly generated, rather than carved out of a lattice.
They are generally organised as follows:
\begin{enumerate}
  \item quantum emitters generate variable-length linear cluster states,
  \item photons from different quantum emitters and photon time-bins are routed into either fusions or single-qubit measurements.
\end{enumerate}
We now study a simple version of these architectures, based on a single quantum emitter, and show that it is universal for \emph{arbitrary}
qubit circuits. The protocol is built from a quantum emitter, the delay, photonic measurement and correction modules defined in \cref{subsec:LO-protocols},
and the repeat-until-success $\XYm$-fusion module of \cref{def:RUS-protocol}.
We define it as the following diagram:
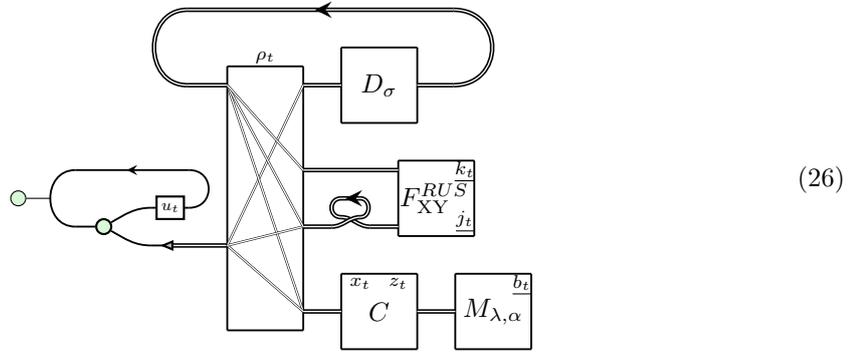
\begin{equation}
  \label{architecture-1}
  \input{figures/architecture-1.tikz}
\end{equation}
By recursively applying \cref{def:unrolling}, we produce a diagram in $\bf{ZXLO}$ which is structured as follows:
\begin{itemize}
  \item the top part of the diagram consists of a variable GHZ-linear cluster produced by the emitter,
  \item the middle part is obtained by unrolling the delay module and equates to a permutation of the qubits,
  \item the bottom part is a sequence of fusions and single qubit measurements of arbitrary types.
\end{itemize}
As an example, the following is the success term of a possible unrolling of the architecture:
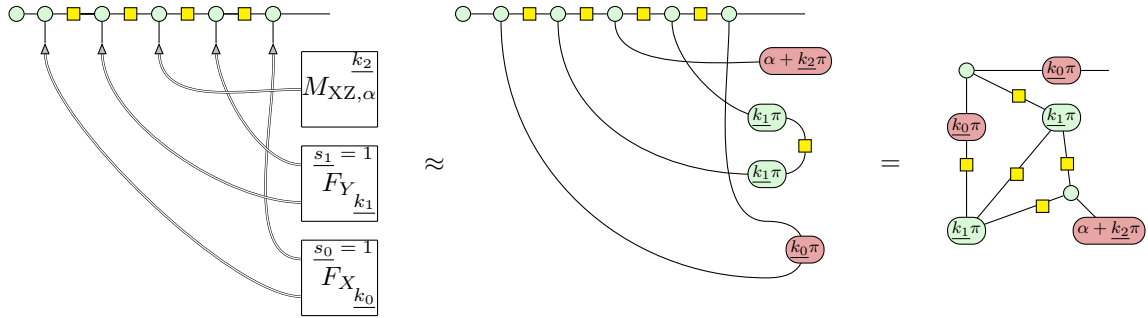
\[
  \input{figures/architecture-unrolling.tikz}
\]
To prove universality of the architecture, we show that it can be used to implement any XY-fusion network where the resource graph is a line.
\begin{definition}
  A linear XY-fusion pattern is an XY-fusion pattern where the $n$-qubit register is totally ordered $V = \set{1, \dots, n}$ and the entangling commands are restricted to be of the form $E_{i, i \plus 1}$.
  A linear XY-fusion network is an XY-fusion network where the resource graph is a disjoint union of lines.
\end{definition}
\begin{proposition}
  \label{prop:hamiltonian-path}
  For any labeled open graph $\mathcal{G}$ with flow, there is a linear XY-fusion network $\mathcal{F}$ with flow with the same target linear map $T(\mathcal{G}) = T(\mathcal{F})$.
\end{proposition}
\begin{proof}
  We use exclusively the $Y$-fusion measurement which adds a hadamard edge between nodes.
  Given any labeled open graph $\mathcal{G} = (G, I, O, \lambda, \alpha)$ with flow, we may extend it to an equivalent labeled open graph $(G', I, O, \lambda, \alpha)$
  such that $G'$ has a Hamiltonian path by finding a Hamiltonian completion of $G$.
  The additional edges in the completion are constructed by introducing nodes measured in the Z-basis, an operation which preserves the existence of Pauli flow~\cite{mcelvanneyCompleteFlowpreservingRewrite2023}.
  Then we may construct a linear XY fusion network $(L, F)$ where $L \sub G'$ is the Hamiltonian path and $F$ is the set of all remaining edges.
  It is easy to check that this fusion network has the same target linear map as $G$.
\end{proof}
\begin{proposition}
  \label{prop:pattern-implementability}
  The protocol in (\ref{architecture-1}) has settings $\lambda, \alpha, \sigma, \rho, u$ that implement any runnable linear XY-fusion pattern, with probability arbitrarily close to $1$.
\end{proposition}
\begin{proof}
  Fix a linear XY-fusion pattern.
  Let $f_i$ be the number of fusions applied to qubit $i$.
  The total number of fusion operations is $f = \frac{1}{2} \sum_{i = 1}^n f_i$.
  For each qubit $i$, we emit $kf_i + 1$ photons entangled as a GHZ state with the atom, where $r$ is a positive integer.
  Between rounds, we either apply a hadamard gate on the atom if the command $E_{i, i \plus 1}$ is present,
  or else we emit an additional photon to be measured in the $X$ basis.
  By setting the parameters $\rho$ and $\sigma$ we may route these photons arbitrarily in either a RUS fusion or a single-qubit measurement, following \cref{lemma:permutations}.
  We use $r$ photons for each node in a RUS fusion operation, giving us a probability of success of $(1 - \frac{1}{2^r})$ for each of the $f$ fusions.
  If the RUS fusion fails after $k$ rounds we restart the whole computation.
  Finally, we can apply any sequence of single qubit measurements and corrections on the remaining $n$ photons, following \cref{lemma-correct-measure}.
  In order to achieve a total success probability $\epsilon$ close to $1$ we just have to set $r$ an odd integer such that $(1 - \frac{1}{2^r})^f > 1 - \epsilon$.
\end{proof}

Note that the RUS $X$ and $Y$ fusions defined above may induce additional $Z$ errors on the target qubits.
Even for $Y$-fusion, it is sufficient to set $n$ to be even (i.e.\@ repeat an odd number of times) to ensure that the error is Pauli.
Thus, in order to correct these errors in an XY-fusion pattern, we must add $Z$ corrections on the target qubit.
This is always possible with the factorisation given in \cref{thm:flow-pattern}, since fusion nodes precede their target qubits in the partial order.

\begin{theorem}
  \label{thm:universality}
  The protocol in \cref{architecture-1} is strongly universal for any qubit unitary.
\end{theorem}
\begin{proof}
  Given any qubit unitary, we may represent it as a labeled open graph $\mathcal{G}$ with flow. By \cref{prop:hamiltonian-path}, there is a linear XY fusion pattern $\mathcal{F}$
  with flow and the same target linear map.
  This gives rise to a runnable linear XY fusion pattern and the result follows by \cref{prop:pattern-implementability}.
\end{proof}

The proof of strong universality above crucially relies on a linear cluster state of length arbitrary.
In practice however, the length of the resource state is limited by the coherence time of the atom.
If $R$ is the rate of emission and $T$ is the maximum coherence time, we can have at most $N = RT$ photons in each resource state.
With this restriction, we cannot show strong universality, as this relied on the ability to boost the fusion success probability arbitrarily.
We thus obtain the following weak universality result, which can be directly compared with \cref{thm:photonic-universality}.

\begin{theorem}\label{thm:emitter-universality}
  The protocol in \cref{architecture-1}, with maximum resource state length $N$, is weakly universal for MBQC patterns with number of edges bounded by a constant $E$.
\end{theorem}
\begin{proof}
  Given a set of MBQC patterns with number of edges bounded by a constant $E$.
  We follow the same procedure as in \cref{prop:pattern-implementability} but we now can attempt each fusion for only a bounded number of times.
  In the worst case scenario, each node in the graph state is implemented by a single round of $N$ photon emissions and we can thus repeat each fusion 
  at most $\frac{N}{m}$ times, where $m$ is the maximum vertex order of the MBQC patterns.
  The total number of photons for this worst-case scenario is $NE$, which is proportional to the total time required to implement any pattern in the set.
  Then we can set $\epsilon \geq 1 - (1 - \frac{1}{2^{\frac{N}{m}}})^E$ to get the desired approximation over all patterns in the set.
\end{proof}

The scaling obtained applies to a worst-case scenario and may be improved for optimized MBQC patterns by
using the linear cluster states to cover a larger part of the graph, thus reducing the number of fusions which are the main limiting factor.
Using multiple quantum emitters would allow to reduce delays and the total execution time.
The programmable nature of the setup in \cref{architecture-1} allows for these types of optimization which are not available in 
the setup of \cref{eq:photonic-architecture} that relies on a fixed compilation process onto the lattice.
Moreover, note that implementing an MBQC graph $G$ with the gate set of \cref{eq:photonic-architecture} may require a circuit with up to $O(E^3)$ gates~\cite{mhalla_graph_2012}.
This suggests that programmable setups with variable delays may offer a more feasible route to short-term photonic computing, 
although the additional number of components required could have an impact on their long-term scalability.

\begin{remark}
  The architecture studied in this section can be implemented using optical components currently available in advanced photonic labs. 
  In particular, as discussed in \cref{subsec:RSG}, the quantum emitter gives a high level of determinism to the resource state generation part. 
  The other sources of error will be photon loss and distinguishability. The first can be heralded but will impact the tolerance $\epsilon$ of \cref{thm:emitter-universality}. 
  The latter cannot be heralded and thus will have an impact on the fidelity of the resulting graph state. 
  Using a \emph{single} quantum emitter can however mitigate this latter source of error, at the cost of longer delay lines and thus photon loss.
\end{remark}

\section{Conclusion}

We have introduced a graphical framework for networked quantum computing that integrates linear optics, the ZX calculus, and dataflow programming. 
This framework connects photonic hardware models with formal methods from programming languages, 
supporting verification, optimization, and compilation in distributed quantum architectures.

At the \textbf{base layer}, we combined ZX diagrams with linear optical circuits to analyze qubit-photon interactions (\cref{sec:base}). 
We characterized all fusion measurements with failure outcomes in the $\XYm$ plane and whose success outcomes induce correctable Pauli errors, 
and identified the class of correctable measurements that give rise to useful entangling operations (\cref{sec:characterization}). 

At the \textbf{channel layer}, we extended the calculus to capture measurement, feedforward, and classical–quantum interaction (\cref{sec:channel}). 
Within this setting, we formally distinguished fusion patterns from fusion networks, and introduced the notion of partially static flow, 
showing that decompositions of open graphs with Pauli flow yield fusion networks with guaranteed correctability (\cref{sec:flow}). 
These results support recent work on fusion networks~\cite{zilkCompilerUniversalPhotonic2022,zhangOneQCompilationFramework2023,leeGraphtheoreticalOptimizationFusionbased2023} 
by providing the first systematic account of determinism in measurement-driven photonic protocols.

At the \textbf{stream layer}, we enriched the formalism with discrete-time dynamics to describe protocols acting on photon streams with routers, delays, switches, and time-controlled emitters (\cref{sec:stream}). 
This extension enables inductive proofs about dynamic quantum processes, including the correctness of new repeat-until-success protocols that boost the probability of success of planar fusion measurements (\cref{subsec:RUS}). 
Using these tools, we established universality results for photonic architectures (\cref{sec:universality}), 
showing how MBQC patterns can be compiled into concrete optical instructions without recourse to universal graph states.

Several avenues for future work naturally emerge from this work. 
Among them are the characterization of higher-arity fusion measurements, the development of Pauli-flow techniques for stream processes, and the use of diagrammatic methods to capture renormalization ideas \cite{herr_local_2018} in universality proofs.
On the practical side, optimization strategies that reduce fusion counts, time delays, and optical component depth will be crucial to bridging theory and experiment. 
More broadly, the framework establishes a foundation for systematic methods in distributed quantum computing, with applications ranging from distributed error correction to advanced quantum communication protocols.

\section*{Acknowledgements}

GdF would like to thank Mario Rom\'an for helping elucidate multiple aspects of the stream formalism.
We would like to thank Will Simmons for reviewing an earlier version of this manuscript.
We benefited from discussions with Ross Duncan, Pierre-Emmanuel Emeriau, Paul Hilaire, Dan Mills, Luc\'ia Tormo Ba\~nuelos, Razin A.\@ Shaikh, and Richie Yeung.

\bibliographystyle{plainnat}
\bibliography{preamble/references_doi}

\begin{thebibliography}{110}
\providecommand{\natexlab}[1]{#1}
\providecommand{\url}[1]{\texttt{#1}}
\expandafter\ifx\csname urlstyle\endcsname\relax
  \providecommand{\doi}[1]{doi: #1}\else
  \providecommand{\doi}{doi: \begingroup \urlstyle{rm}\Url}\fi

\bibitem[Aaronson and Arkhipov(2011)]{aaronsonComputationalComplexityLinear2011}
Scott Aaronson and Alex Arkhipov.
\newblock The computational complexity of linear optics.
\newblock In \emph{Proceedings of the Forty-Third Annual {{ACM}} Symposium on {{Theory}} of Computing}, {{STOC}} '11, pages 333--342, New York, NY, USA, June 2011. Association for Computing Machinery.
\newblock ISBN 978-1-4503-0691-1.
\newblock \doi{10.1145/1993636.1993682}.

\bibitem[Abramsky and Heunen(2012)]{abramskyAlgebrasNonunitalFrobenius2012}
Samson Abramsky and Chris Heunen.
\newblock H*-algebras and nonunital {{Frobenius}} algebras: First steps in infinite-dimensional categorical quantum mechanics.
\newblock In Samson Abramsky and Michael Mislove, editors, \emph{Mathematical {{Foundations}} of {{Information Flow}}}, volume~71 of \emph{Proceedings of {{Symposia}} in {{Applied Mathematics}}}, pages 1--24. American Mathematical Society, 2012.
\newblock ISBN 978-0-8218-4923-1.
\newblock \doi{10.1090/psapm/071}.

\bibitem[Andres-Martinez et~al.(2024)Andres-Martinez, Forrer, Mills, Wu, Henaut, Yamamoto, Murao, and Duncan]{andres-martinez_distributing_2024}
Pablo Andres-Martinez, Tim Forrer, Daniel Mills, Jun-Yi Wu, Luciana Henaut, Kentaro Yamamoto, Mio Murao, and Ross Duncan.
\newblock Distributing circuits over heterogeneous, modular quantum computing network architectures.
\newblock \emph{Quantum Science and Technology}, 9\penalty0 (4):\penalty0 045021, October 2024.
\newblock ISSN 2058-9565.
\newblock \doi{10.1088/2058-9565/ad6734}.

\bibitem[Azuma et~al.(2015)Azuma, Tamaki, and Lo]{azumaAllphotonicQuantumRepeaters2015}
Koji Azuma, Kiyoshi Tamaki, and Hoi-Kwong Lo.
\newblock All-photonic quantum repeaters.
\newblock \emph{Nature Communications}, 6\penalty0 (1):\penalty0 6787, April 2015.
\newblock ISSN 2041-1723.
\newblock \doi{10.1038/ncomms7787}.

\bibitem[Backens et~al.(2021)Backens, {Miller-Bakewell}, de~Felice, Lobski, and van~de Wetering]{backensThereBackAgain2021}
Miriam Backens, Hector {Miller-Bakewell}, Giovanni de~Felice, Leo Lobski, and John van~de Wetering.
\newblock There and back again: {{A}} circuit extraction tale.
\newblock \emph{Quantum}, 5:\penalty0 421, March 2021.
\newblock \doi{10.22331/q-2021-03-25-421}.

\bibitem[Bartolucci et~al.(2021{\natexlab{a}})Bartolucci, Birchall, Bonneau, Cable, {Gimeno-Segovia}, Kieling, Nickerson, Rudolph, and Sparrow]{bartolucciSwitchNetworksPhotonic2021}
Sara Bartolucci, Patrick Birchall, Damien Bonneau, Hugo Cable, Mercedes {Gimeno-Segovia}, Konrad Kieling, Naomi Nickerson, Terry Rudolph, and Chris Sparrow.
\newblock Switch networks for photonic fusion-based quantum computing, September 2021{\natexlab{a}}.
\newblock URL \url{https://doi.org/10.48550/arXiv.2109.13760}.

\bibitem[Bartolucci et~al.(2021{\natexlab{b}})Bartolucci, Birchall, {Gimeno-Segovia}, Johnston, Kieling, Pant, Rudolph, Smith, Sparrow, and Vidrighin]{bartolucciCreationEntangledPhotonic2021}
Sara Bartolucci, Patrick~M. Birchall, Mercedes {Gimeno-Segovia}, Eric Johnston, Konrad Kieling, Mihir Pant, Terry Rudolph, Jake Smith, Chris Sparrow, and Mihai~D. Vidrighin.
\newblock Creation of {{Entangled Photonic States Using Linear Optics}}, June 2021{\natexlab{b}}.
\newblock URL \url{https://doi.org/10.48550/arXiv.2106.13825}.

\bibitem[Bartolucci et~al.(2023)Bartolucci, Birchall, Bomb{\'i}n, Cable, Dawson, {Gimeno-Segovia}, Johnston, Kieling, Nickerson, Pant, Pastawski, Rudolph, and Sparrow]{bartolucciFusionbasedQuantumComputation2023}
Sara Bartolucci, Patrick Birchall, Hector Bomb{\'i}n, Hugo Cable, Chris Dawson, Mercedes {Gimeno-Segovia}, Eric Johnston, Konrad Kieling, Naomi Nickerson, Mihir Pant, Fernando Pastawski, Terry Rudolph, and Chris Sparrow.
\newblock Fusion-based quantum computation.
\newblock \emph{Nature Communications}, 14\penalty0 (1):\penalty0 912, February 2023.
\newblock ISSN 2041-1723.
\newblock \doi{10.1038/s41467-023-36493-1}.

\bibitem[Blinov et~al.(2004)Blinov, Moehring, Duan, and Monroe]{blinovObservationEntanglementSingle2004}
B.~B. Blinov, D.~L. Moehring, L.-M. Duan, and C.~Monroe.
\newblock Observation of entanglement between a single trapped atom and a single photon.
\newblock \emph{Nature}, 428\penalty0 (6979):\penalty0 153--157, March 2004.
\newblock ISSN 1476-4687.
\newblock \doi{10.1038/nature02377}.

\bibitem[Bombin et~al.()Bombin, Litinski, Nickerson, Pastawski, and Roberts]{bombinUnifyingFlavorsFault2023}
Hector Bombin, Daniel Litinski, Naomi Nickerson, Fernando Pastawski, and Sam Roberts.
\newblock Unifying flavors of fault tolerance with the {ZX} calculus.
\newblock 8:\penalty0 1379.
\newblock \doi{10.22331/q-2024-06-18-1379}.

\bibitem[Bonchi et~al.()Bonchi, Di~Lavore, and Román]{bonchi_effectful_2025}
Filippo Bonchi, Elena Di~Lavore, and Mario Román.
\newblock Effectful mealy machines: Coalgebraic and causal traces (invited talk).
\newblock In \emph{11th Conference on Algebra and Coalgebra in Computer Science ({CALCO} 2025)}. Schloss Dagstuhl - Leibniz-Zentrum für Informatik.
\newblock \doi{10.4230/LIPIcs.CALCO.2025.1}.

\bibitem[Bose et~al.(1998)Bose, Vedral, and Knight]{boseMultiparticleGeneralizationEntanglement1998}
S.~Bose, V.~Vedral, and P.~L. Knight.
\newblock Multiparticle generalization of entanglement swapping.
\newblock \emph{Physical Review A}, 57\penalty0 (2):\penalty0 822--829, February 1998.
\newblock \doi{10.1103/PhysRevA.57.822}.

\bibitem[Browne et~al.()Browne, Kashefi, Mhalla, and Perdrix]{browneGeneralizedFlowDeterminism2007}
D.~E. Browne, E.~Kashefi, M.~Mhalla, and S.~Perdrix.
\newblock Generalized flow and determinism in measurement-based quantum computation.
\newblock 9\penalty0 (8):\penalty0 250--250.
\newblock ISSN 1367-2630.
\newblock \doi{10.1088/1367-2630/9/8/250}.

\bibitem[Browne and Rudolph(2005)]{browneResourceEfficientLinearOptical2005}
Daniel~E. Browne and Terry Rudolph.
\newblock Resource-{{Efficient Linear Optical Quantum Computation}}.
\newblock \emph{Physical Review Letters}, 95\penalty0 (1):\penalty0 010501, June 2005.
\newblock \doi{10.1103/PhysRevLett.95.010501}.

\bibitem[Carette et~al.(2021)Carette, de~Visme, and Perdrix]{caretteGraphicalLanguageDelayed2021}
Titouan Carette, Marc de~Visme, and Simon Perdrix.
\newblock Graphical {{Language}} with {{Delayed Trace}}: {{Picturing Quantum Computing}} with {{Finite Memory}}.
\newblock In \emph{2021 36th {{Annual ACM}}/{{IEEE Symposium}} on {{Logic}} in {{Computer Science}} ({{LICS}})}, pages 1--13. IEEE Computer Society, June 2021.
\newblock ISBN 978-1-66544-895-6.
\newblock \doi{10.1109/LICS52264.2021.9470553}.

\bibitem[Carolan et~al.(2019)Carolan, Chakraborty, Harris, Pant, {Baehr-Jones}, Hochberg, and Englund]{carolanScalableFeedbackControl2019}
Jacques Carolan, Uttara Chakraborty, Nicholas~C. Harris, Mihir Pant, Tom {Baehr-Jones}, Michael Hochberg, and Dirk Englund.
\newblock Scalable feedback control of single photon sources for photonic quantum technologies.
\newblock \emph{Optica}, 6\penalty0 (3):\penalty0 335--340, March 2019.
\newblock ISSN 2334-2536.
\newblock \doi{10.1364/OPTICA.6.000335}.

\bibitem[Chaiken(1967)]{chaiken_finite-particle_1967}
Jan~M. Chaiken.
\newblock Finite-particle representations and states of the canonical commutation relations.
\newblock \emph{Annals of Physics}, 42\penalty0 (1):\penalty0 23--80, March 1967.
\newblock ISSN 0003-4916.
\newblock \doi{10.1016/0003-4916(67)90186-8}.

\bibitem[Chiribella et~al.(2008)Chiribella, D'Ariano, and Perinotti]{chiribella_quantum_2008}
Giulio Chiribella, Giacomo~Mauro D'Ariano, and Paolo Perinotti.
\newblock Quantum {Circuits} {Architecture}.
\newblock \emph{Physical Review Letters}, 101\penalty0 (6):\penalty0 060401, August 2008.
\newblock ISSN 0031-9007, 1079-7114.
\newblock \doi{10.1103/PhysRevLett.101.060401}.

\bibitem[Chiribella et~al.(2009)Chiribella, D'Ariano, and Perinotti]{chiribella_theoretical_2009}
Giulio Chiribella, Giacomo~M. D'Ariano, and Paolo Perinotti.
\newblock Theoretical framework for quantum networks.
\newblock \emph{Physical Review A}, 80\penalty0 (2):\penalty0 022339, August 2009.
\newblock ISSN 1050-2947, 1094-1622.
\newblock \doi{10.1103/PhysRevA.80.022339}.

\bibitem[Clément et~al.(2022)Clément, Heurtel, Mansfield, Perdrix, and Valiron]{clementLOvCalculusGraphicalLanguage2022}
Alexandre Clément, Nicolas Heurtel, Shane Mansfield, Simon Perdrix, and Benoît Valiron.
\newblock Lo\textsubscript{v}-calculus: A graphical language for linear optical quantum circuits.
\newblock In Stefan Szeider, Robert Ganian, and Alexandra Silva, editors, \emph{47th international symposium on mathematical foundations of computer science (MFCS 2022)}, volume 241 of \emph{Leibniz international proceedings in informatics (LIPIcs)}, page 35:1–35:16, Dagstuhl, Germany, 2022. Schloss Dagstuhl -- Leibniz-Zentrum f{\"u}r Informatik.
\newblock ISBN 978-3-95977-256-3.
\newblock \doi{10.4230/LIPIcs.MFCS.2022.35}.

\bibitem[Coecke and Duncan(2008)]{coeckeInteractingQuantumObservables2008}
Bob Coecke and Ross Duncan.
\newblock Interacting {{Quantum Observables}}.
\newblock In Luca Aceto, Ivan Damg{\aa}rd, Leslie~Ann Goldberg, Magn{\'u}s~M. Halld{\'o}rsson, Anna Ing{\'o}lfsd{\'o}ttir, and Igor Walukiewicz, editors, \emph{Automata, {{Languages}} and {{Programming}}}, Lecture {{Notes}} in {{Computer Science}}, pages 298--310, Berlin, Heidelberg, 2008. Springer.
\newblock ISBN 978-3-540-70583-3.
\newblock \doi{10.1007/978-3-540-70583-3_25}.

\bibitem[Coecke and Kissinger(2017)]{coeckePicturingQuantumProcesses2017}
Bob Coecke and Aleks Kissinger.
\newblock \emph{Picturing Quantum Processes}.
\newblock Cambridge University Press, March 2017.
\newblock ISBN 978-1-107-10422-8.
\newblock \doi{10.1017/9781316219317}.

\bibitem[Coecke et~al.(2016)Coecke, Heunen, and Kissinger]{coeckeCategoriesQuantumClassical2016}
Bob Coecke, Chris Heunen, and Aleks Kissinger.
\newblock Categories of quantum and classical channels.
\newblock \emph{Quantum Information Processing}, 15\penalty0 (12):\penalty0 5179--5209, December 2016.
\newblock ISSN 1573-1332.
\newblock \doi{10.1007/s11128-014-0837-4}.

\bibitem[Cogan et~al.(2023)Cogan, Su, Kenneth, and Gershoni]{coganDeterministicGenerationIndistinguishable2023}
Dan Cogan, Zu-En Su, Oded Kenneth, and David Gershoni.
\newblock Deterministic generation of indistinguishable photons in a cluster state.
\newblock \emph{Nature Photonics}, 17\penalty0 (4):\penalty0 324--329, April 2023.
\newblock ISSN 1749-4893.
\newblock \doi{10.1038/s41566-022-01152-2}.

\bibitem[Coste et~al.(2023)Coste, Fioretto, Belabas, Wein, Hilaire, Frantzeskakis, Gundin, Goes, Somaschi, Morassi, Lema{\^i}tre, Sagnes, Harouri, Economou, Auffeves, Krebs, Lanco, and Senellart]{costeHighrateEntanglementSemiconductor2023}
N.~Coste, D.~A. Fioretto, N.~Belabas, S.~C. Wein, P.~Hilaire, R.~Frantzeskakis, M.~Gundin, B.~Goes, N.~Somaschi, M.~Morassi, A.~Lema{\^i}tre, I.~Sagnes, A.~Harouri, S.~E. Economou, A.~Auffeves, O.~Krebs, L.~Lanco, and P.~Senellart.
\newblock High-rate entanglement between a semiconductor spin and indistinguishable photons.
\newblock \emph{Nature Photonics}, 17\penalty0 (7):\penalty0 582--587, July 2023.
\newblock ISSN 1749-4893.
\newblock \doi{10.1038/s41566-023-01186-0}.

\bibitem[Covey et~al.(2023)Covey, Weinfurter, and Bernien]{covey_quantum_2023}
Jacob~P. Covey, Harald Weinfurter, and Hannes Bernien.
\newblock Quantum networks with neutral atom processing nodes.
\newblock \emph{npj Quantum Information}, 9\penalty0 (1):\penalty0 1--12, September 2023.
\newblock ISSN 2056-6387.
\newblock \doi{10.1038/s41534-023-00759-9}.

\bibitem[Danos et~al.(2007)Danos, Kashefi, and Panangaden]{danosMeasurementCalculus2007}
Vincent Danos, Elham Kashefi, and Prakash Panangaden.
\newblock The measurement calculus.
\newblock \emph{Journal of the ACM}, 54\penalty0 (2):\penalty0 8--es, April 2007.
\newblock ISSN 0004-5411.
\newblock \doi{10.1145/1219092.1219096}.

\bibitem[de~Beaudrap and Horsman(2020)]{beaudrapZXCalculusLanguage2020}
Niel de~Beaudrap and Dominic Horsman.
\newblock The {{ZX}} calculus is a language for surface code lattice surgery.
\newblock \emph{Quantum}, 4:\penalty0 218, January 2020.
\newblock \doi{10.22331/q-2020-01-09-218}.

\bibitem[{de Beaudrap} et~al.(2020){de Beaudrap}, Bian, and Wang]{debeaudrapFastEffectiveTechniques2020}
Niel {de Beaudrap}, Xiaoning Bian, and Quanlong Wang.
\newblock Fast and {{Effective Techniques}} for {{T-Count Reduction}} via {{Spider Nest Identities}}.
\newblock In \emph{15th {{Conference}} on the {{Theory}} of {{Quantum Computation}}, {{Communication}} and {{Cryptography}} ({{TQC}} 2020)}, volume 158 of \emph{Leibniz International Proceedings in Informatics ({{LIPIcs}})}, pages 11:1--11:23, Dagstuhl, Germany, 2020. Schloss-Dagstuhl - Leibniz Zentrum f{\"u}r Informatik.
\newblock ISBN 978-3-95977-146-7.
\newblock \doi{10.4230/LIPIcs.TQC.2020.11}.

\bibitem[{de Felice} and Coecke(2023)]{defeliceQuantumLinearOptics2023}
Giovanni {de Felice} and Bob Coecke.
\newblock Quantum {{Linear Optics}} via {{String Diagrams}}.
\newblock In Stefano Gogioso and Matty Hoban, editors, \emph{Proceedings 19th International Conference on Quantum Physics and Logic}, volume 394 of \emph{Electronic Proceedings in Theoretical Computer Science}, pages 83--100, Wolfson College, Oxford, UK, 2023. Open Publishing Association.
\newblock \doi{10.4204/EPTCS.394.6}.

\bibitem[{de Felice} et~al.(2023){de Felice}, Shaikh, Po{\'o}r, Yeh, Wang, and Coecke]{defeliceLightMatterInteractionZXW2023}
Giovanni {de Felice}, Razin~A. Shaikh, Boldizs{\'a}r Po{\'o}r, Lia Yeh, Quanlong Wang, and Bob Coecke.
\newblock Light-{{Matter Interaction}} in the {{ZXW Calculus}}.
\newblock In Shane Mansfield, Benoit Val{\^i}ron, and Vladimir Zamdzhiev, editors, \emph{Proceedings of the Twentieth International Conference on Quantum Physics and Logic, Paris, France, 17-21st July 2023}, volume 384 of \emph{Electronic Proceedings in Theoretical Computer Science}, pages 20--46. Open Publishing Association, July 2023.
\newblock \doi{10.4204/EPTCS.384.2}.

\bibitem[Di~Lavore et~al.()Di~Lavore, Gianola, Román, Sabadini, and Sobociński]{di_lavore_canonical_2021}
Elena Di~Lavore, Alessandro Gianola, Mario Román, Nicoletta Sabadini, and Paweł Sobociński.
\newblock A canonical algebra of open transition systems.
\newblock In \emph{Formal Aspects of Component Software: 17th International Conference, {FACS} 2021, Virtual Event, October 28–29, 2021, Proceedings}, pages 63--81. Springer-Verlag.
\newblock ISBN 978-3-030-90635-1.
\newblock \doi{10.1007/978-3-030-90636-8_4}.

\bibitem[Di~Lavore et~al.(2022)Di~Lavore, {de Felice}, and Rom{\'a}n]{dilavoreMonoidalStreamsDataflow2022}
Elena Di~Lavore, Giovanni {de Felice}, and Mario Rom{\'a}n.
\newblock Monoidal {{Streams}} for {{Dataflow Programming}}.
\newblock In \emph{Proceedings of the 37th {{Annual ACM}}/{{IEEE Symposium}} on {{Logic}} in {{Computer Science}}}, {{LICS}} '22, pages 1--14, New York, NY, USA, August 2022. Association for Computing Machinery.
\newblock ISBN 978-1-4503-9351-5.
\newblock \doi{10.1145/3531130.3533365}.

\bibitem[Duncan(2013)]{duncanGraphicalApproachMeasurementbased2013}
Ross Duncan.
\newblock A graphical approach to measurement-based quantum computing.
\newblock In Chris Heunen, Mehrnoosh Sadrzadeh, and Edward Grefenstette, editors, \emph{Quantum Physics and Linguistics: {{A}} Compositional, Diagrammatic Discourse}, pages 50--89. Oxford University Press, February 2013.
\newblock ISBN 978-0-19-964629-6.
\newblock \doi{10.1093/acprof:oso/9780199646296.001.0001}.

\bibitem[Duncan and Perdrix(2010)]{duncanRewritingMeasurementBasedQuantum2010}
Ross Duncan and Simon Perdrix.
\newblock Rewriting {{Measurement-Based Quantum Computations}} with {{Generalised Flow}}.
\newblock In Samson Abramsky, Cyril Gavoille, Claude Kirchner, Friedhelm {Meyer auf der Heide}, and Paul~G. Spirakis, editors, \emph{Automata, {{Languages}} and {{Programming}}}, Lecture {{Notes}} in {{Computer Science}}, pages 285--296, Berlin, Heidelberg, 2010. Springer.
\newblock ISBN 978-3-642-14162-1.
\newblock \doi{10.1007/978-3-642-14162-1_24}.

\bibitem[Duncan et~al.(2020)Duncan, Kissinger, Perdrix, and van~de Wetering]{duncanGraphtheoreticSimplificationQuantum2020}
Ross Duncan, Aleks Kissinger, Simon Perdrix, and John van~de Wetering.
\newblock Graph-theoretic {{Simplification}} of {{Quantum Circuits}} with the {{ZX-calculus}}.
\newblock \emph{Quantum}, 4:\penalty0 279, June 2020.
\newblock \doi{10.22331/q-2020-06-04-279}.

\bibitem[Ekimov and Onushchenko(1981)]{ekimovQuantumSizeEffect1981}
A.~I. Ekimov and A.~A. Onushchenko.
\newblock Quantum {{Size Effect}} in {{Three-Dimensional Microscopic Semiconductor Crystals}}.
\newblock \emph{JETP Letters}, 118\penalty0 (1):\penalty0 S15--S17, August 1981.
\newblock ISSN 1090-6487.
\newblock \doi{10.1134/S0021364023130040}.

\bibitem[Ewert and van Loock(2014)]{ewertEfficientBellMeasurement2014}
Fabian Ewert and Peter van Loock.
\newblock $3/4$-efficient bell measurement with passive linear optics and unentangled ancillae.
\newblock \emph{Physical Review Letters}, 113\penalty0 (14):\penalty0 140403, September 2014.
\newblock \doi{10.1103/PhysRevLett.113.140403}.

\bibitem[Felice et~al.()Felice, Toumi, and Coecke]{felice_discopy_2021}
Giovanni~de Felice, Alexis Toumi, and Bob Coecke.
\newblock {DisCoPy}: Monoidal categories in python.
\newblock 333:\penalty0 183--197.
\newblock ISSN 2075-2180.
\newblock \doi{10.4204/EPTCS.333.13}.

\bibitem[Ferrari et~al.(2021)Ferrari, Cacciapuoti, Amoretti, and Caleffi]{ferrari_compiler_2021}
Davide Ferrari, Angela~Sara Cacciapuoti, Michele Amoretti, and Marcello Caleffi.
\newblock Compiler {Design} for {Distributed} {Quantum} {Computing}.
\newblock \emph{IEEE Transactions on Quantum Engineering}, 2:\penalty0 1--20, 2021.
\newblock ISSN 2689-1808.
\newblock \doi{10.1109/TQE.2021.3053921}.

\bibitem[{Gimeno-Segovia} et~al.(2015){Gimeno-Segovia}, Shadbolt, Browne, and Rudolph]{gimeno-segoviaThreePhotonGreenbergerHorneZeilingerStates2015}
Mercedes {Gimeno-Segovia}, Pete Shadbolt, Dan~E. Browne, and Terry Rudolph.
\newblock From {{Three-Photon Greenberger-Horne-Zeilinger States}} to {{Ballistic Universal Quantum Computation}}.
\newblock \emph{Physical Review Letters}, 115\penalty0 (2):\penalty0 020502, July 2015.
\newblock \doi{10.1103/PhysRevLett.115.020502}.

\bibitem[Gliniasty et~al.()Gliniasty, Hilaire, Emeriau, Wein, Salavrakos, and Mansfield]{degliniastySpinOpticalQuantumComputing2024}
Grégoire~de Gliniasty, Paul Hilaire, Pierre-Emmanuel Emeriau, Stephen~C. Wein, Alexia Salavrakos, and Shane Mansfield.
\newblock A spin-optical quantum computing architecture.
\newblock 8:\penalty0 1423.
\newblock \doi{10.22331/q-2024-07-24-1423}.

\bibitem[Grice(2011)]{griceArbitrarilyCompleteBellstate2011}
W.~P. Grice.
\newblock Arbitrarily complete {{Bell-state}} measurement using only linear optical elements.
\newblock \emph{Physical Review A}, 84\penalty0 (4):\penalty0 042331, October 2011.
\newblock \doi{10.1103/PhysRevA.84.042331}.

\bibitem[Halbwachs et~al.(1992)Halbwachs, Lagnier, and Ratel]{halbwachs_programming_1992}
Nicolas Halbwachs, Fabienne Lagnier, and Christophe Ratel.
\newblock Programming and {Verifying} {Real}-{Time} {Systems} by {Means} of the {Synchronous} {Data}-{Flow} {Language} {LUSTRE}.
\newblock \emph{IEEE Trans. Softw. Eng.}, 18\penalty0 (9):\penalty0 785--793, September 1992.
\newblock ISSN 0098-5589.
\newblock \doi{10.1109/32.159839}.

\bibitem[Harris et~al.(1967)Harris, Oshman, and Byer]{harrisObservationTunableOptical1967}
S.~E. Harris, M.~K. Oshman, and R.~L. Byer.
\newblock Observation of {{Tunable Optical Parametric Fluorescence}}.
\newblock \emph{Physical Review Letters}, 18\penalty0 (18):\penalty0 732--734, May 1967.
\newblock \doi{10.1103/PhysRevLett.18.732}.

\bibitem[Herr et~al.(2018)Herr, Paler, Devitt, and Nori]{herr_local_2018}
Daniel Herr, Alexandru Paler, Simon~J. Devitt, and Franco Nori.
\newblock A local and scalable lattice renormalization method for ballistic quantum computation.
\newblock \emph{npj Quantum Information}, 4\penalty0 (1):\penalty0 27, June 2018.
\newblock ISSN 2056-6387.
\newblock \doi{10.1038/s41534-018-0076-0}.

\bibitem[Heurtel()]{heurtelCompleteGraphicalLanguage2024}
Nicolas Heurtel.
\newblock A complete graphical language for linear optical circuits with finite-photon-number sources and detectors.
\newblock In Jörg Endrullis and Sylvain Schmitz, editors, \emph{33rd {EACSL} Annual Conference on Computer Science Logic ({CSL} 2025)}, volume 326 of \emph{Leibniz International Proceedings in Informatics ({LIPIcs})}, pages 38:1--38:23. Schloss Dagstuhl – Leibniz-Zentrum für Informatik.
\newblock ISBN 978-3-95977-362-1.
\newblock \doi{10.4230/LIPIcs.CSL.2025.38}.

\bibitem[Hilaire et~al.(2024)Hilaire, Dessertaine, Bourdoncle, Denys, Gliniasty, Valentí-Rojas, and Mansfield]{hilaire_enhanced_2024}
Paul Hilaire, Théo Dessertaine, Boris Bourdoncle, Aurélie Denys, Grégoire~de Gliniasty, Gerard Valentí-Rojas, and Shane Mansfield.
\newblock Enhanced {Fault}-tolerance in {Photonic} {Quantum} {Computing}: {Floquet} {Code} {Outperforms} {Surface} {Code} in {Tailored} {Architecture}, October 2024.
\newblock URL \url{https://doi.org/10.48550/arXiv.2410.07065}.

\bibitem[Huang et~al.(2023)Huang, Li, Yeh, Kissinger, Mosca, and Vasmer]{huangGraphicalCSSCode2023}
Jiaxin Huang, Sarah~Meng Li, Lia Yeh, Aleks Kissinger, Michele Mosca, and Michael Vasmer.
\newblock Graphical {{CSS}} code transformation using {{ZX}} calculus.
\newblock In Shane Mansfield, Benoit Val{\^i}ron, and Vladimir Zamdzhiev, editors, \emph{Proceedings of the Twentieth International Conference on Quantum Physics and Logic}, volume 384 of \emph{Electronic Proceedings in Theoretical Computer Science}, pages 1--19. Open Publishing Association, 2023.
\newblock \doi{10.4204/EPTCS.384.1}.

\bibitem[Huet et~al.()Huet, Ramesh, Wein, Coste, Hilaire, Somaschi, Morassi, Lemaître, Sagnes, Doty, Krebs, Lanco, Fioretto, and Senellart]{huet_deterministic_2025}
H.~Huet, P.~R. Ramesh, S.~C. Wein, N.~Coste, P.~Hilaire, N.~Somaschi, M.~Morassi, A.~Lemaître, I.~Sagnes, M.~F. Doty, O.~Krebs, L.~Lanco, D.~A. Fioretto, and P.~Senellart.
\newblock Deterministic and reconfigurable graph state generation with a single solid-state quantum emitter.
\newblock 16\penalty0 (1):\penalty0 4337.
\newblock ISSN 2041-1723.
\newblock \doi{10.1038/s41467-025-59693-3}.

\bibitem[Jeandel et~al.(2018{\natexlab{a}})Jeandel, Perdrix, and Vilmart]{jeandelCompleteAxiomatisationZXCalculus2018}
Emmanuel Jeandel, Simon Perdrix, and Renaud Vilmart.
\newblock A {{Complete Axiomatisation}} of the {{ZX-Calculus}} for {{Clifford}}+{{T Quantum Mechanics}}.
\newblock In \emph{Proceedings of the 33rd {{Annual ACM}}/{{IEEE Symposium}} on {{Logic}} in {{Computer Science}}}, {{LICS}} '18, pages 559--568, New York, NY, USA, July 2018{\natexlab{a}}. Association for Computing Machinery.
\newblock ISBN 978-1-4503-5583-4.
\newblock \doi{10.1145/3209108.3209131}.

\bibitem[Jeandel et~al.(2018{\natexlab{b}})Jeandel, Perdrix, and Vilmart]{jeandelDiagrammaticReasoningClifford2018}
Emmanuel Jeandel, Simon Perdrix, and Renaud Vilmart.
\newblock Diagrammatic {{Reasoning}} beyond {{Clifford}}+{{T Quantum Mechanics}}.
\newblock In \emph{Proceedings of the 33rd {{Annual ACM}}/{{IEEE Symposium}} on {{Logic}} in {{Computer Science}}}, {{LICS}} '18, pages 569--578, New York, NY, USA, July 2018{\natexlab{b}}. Association for Computing Machinery.
\newblock ISBN 978-1-4503-5583-4.
\newblock \doi{10.1145/3209108.3209139}.

\bibitem[Jiang et~al.(2007)Jiang, Taylor, Sørensen, and Lukin]{jiang_distributed_2007}
Liang Jiang, Jacob~M. Taylor, Anders~S. Sørensen, and Mikhail~D. Lukin.
\newblock Distributed quantum computation based on small quantum registers.
\newblock \emph{Physical Review A}, 76\penalty0 (6):\penalty0 062323, December 2007.
\newblock \doi{10.1103/PhysRevA.76.062323}.

\bibitem[Kastner(1993)]{kastnerArtificialAtoms1993}
Marc~A. Kastner.
\newblock Artificial {{Atoms}}.
\newblock \emph{Physics Today}, 46\penalty0 (1):\penalty0 24--31, 1993.
\newblock ISSN 0031-9228.
\newblock \doi{10.1063/1.881393}.

\bibitem[Kissinger(2022)]{kissingerPhasefreeZXDiagrams2022}
Aleks Kissinger.
\newblock Phase-free zx diagrams are css codes (...or how to graphically grok the surface code), April 2022.
\newblock URL \url{https://doi.org/10.48550/arXiv.2204.14038}.

\bibitem[Kissinger and {van de Wetering}(2020)]{kissingerReducingNumberNonClifford2020}
Aleks Kissinger and John {van de Wetering}.
\newblock Reducing the number of non-{{Clifford}} gates in quantum circuits.
\newblock \emph{Physical Review A}, 102\penalty0 (2):\penalty0 022406, August 2020.
\newblock \doi{10.1103/PhysRevA.102.022406}.

\bibitem[Knill et~al.(2001)Knill, Laflamme, and Milburn]{knillSchemeEfficientQuantum2001}
E.~Knill, R.~Laflamme, and G.~J. Milburn.
\newblock A scheme for efficient quantum computation with linear optics.
\newblock \emph{Nature}, 409\penalty0 (6816):\penalty0 46--52, January 2001.
\newblock ISSN 1476-4687.
\newblock \doi{10.1038/35051009}.

\bibitem[Kozen and Silva(2017)]{kozenPracticalCoinduction2017}
Dexter Kozen and Alexandra Silva.
\newblock Practical coinduction.
\newblock \emph{Mathematical Structures in Computer Science}, 27\penalty0 (7):\penalty0 1132--1152, October 2017.
\newblock ISSN 0960-1295, 1469-8072.
\newblock \doi{10.1017/S0960129515000493}.

\bibitem[Kupper et~al.()Kupper, Yeung, Poór, Toumi, Cashman, and Felice]{kupper_optyx_2025}
Mateusz Kupper, Richie Yeung, Boldizsár Poór, Alexis Toumi, William Cashman, and Giovanni~de Felice.
\newblock Optyx: A {ZX}-based python library for networked quantum architectures.
\newblock URL \url{https://doi.org/10.48550/arXiv.2512.09648}.

\bibitem[Lavore et~al.(2025)Lavore, Felice, and Román]{lavore_coinductive_2025}
Elena~Di Lavore, Giovanni~de Felice, and Mario Román.
\newblock Coinductive {Streams} in {Monoidal} {Categories}.
\newblock \emph{Logical Methods in Computer Science}, Volume 21, Issue 3:\penalty0 10759, August 2025.
\newblock ISSN 1860-5974.
\newblock \doi{10.46298/lmcs-21(3:18)2025}.

\bibitem[Lee and Jeong(2023)]{leeGraphtheoreticalOptimizationFusionbased2023}
Seok-Hyung Lee and Hyunseok Jeong.
\newblock Graph-theoretical optimization of fusion-based graph state generation.
\newblock \emph{Quantum}, 7:\penalty0 1212, December 2023.
\newblock \doi{10.22331/q-2023-12-20-1212}.

\bibitem[Lee et~al.(2015)Lee, Park, Ralph, and Jeong]{lee_nearly_2015}
Seung-Woo Lee, Kimin Park, Timothy~C. Ralph, and Hyunseok Jeong.
\newblock Nearly deterministic {Bell} measurement with multiphoton entanglement for efficient quantum-information processing.
\newblock \emph{Physical Review A}, 92\penalty0 (5):\penalty0 052324, November 2015.
\newblock \doi{10.1103/PhysRevA.92.052324}.

\bibitem[Lemr et~al.(2011)Lemr, Cernoch, Soubusta, Kieling, Eisert, and Dusek]{lemr_experimental_2011}
Karel Lemr, Antonin Cernoch, Jan Soubusta, Konrad Kieling, Jens Eisert, and Miloslav Dusek.
\newblock Experimental implementation of the optimal linear-optical controlled phase gate.
\newblock \emph{Physical Review Letters}, 106\penalty0 (1):\penalty0 013602, January 2011.
\newblock ISSN 0031-9007, 1079-7114.
\newblock \doi{10.1103/PhysRevLett.106.013602}.

\bibitem[Lim et~al.(2005)Lim, Beige, and Kwek]{limRepeatUntilSuccessLinearOptics2005}
Yuan~Liang Lim, Almut Beige, and Leong~Chuan Kwek.
\newblock Repeat-{{Until-Success Linear Optics Distributed Quantum Computing}}.
\newblock \emph{Physical Review Letters}, 95\penalty0 (3):\penalty0 030505, July 2005.
\newblock \doi{10.1103/PhysRevLett.95.030505}.

\bibitem[Lindner and Rudolph(2009)]{lindnerProposalPulsedOnDemand2009}
Netanel~H. Lindner and Terry Rudolph.
\newblock Proposal for {{Pulsed On-Demand Sources}} of {{Photonic Cluster State Strings}}.
\newblock \emph{Physical Review Letters}, 103\penalty0 (11):\penalty0 113602, September 2009.
\newblock \doi{10.1103/PhysRevLett.103.113602}.

\bibitem[L{\"o}bl et~al.(2024)L{\"o}bl, Pettersson, Paesani, and S{\o}rensen]{loblTransformingGraphStates2024}
Matthias~C. L{\"o}bl, Love~A. Pettersson, Stefano Paesani, and Anders~S. S{\o}rensen.
\newblock Transforming graph states via {{Bell}} state measurements, May 2024.

\bibitem[Madsen et~al.(2022)Madsen, Laudenbach, Askarani, Rortais, Vincent, Bulmer, Miatto, Neuhaus, Helt, Collins, Lita, Gerrits, Nam, Vaidya, Menotti, Dhand, Vernon, Quesada, and Lavoie]{madsenQuantumComputational2022}
Lars~S. Madsen, Fabian Laudenbach, Mohsen~Falamarzi Askarani, Fabien Rortais, Trevor Vincent, Jacob F.~F. Bulmer, Filippo~M. Miatto, Leonhard Neuhaus, Lukas~G. Helt, Matthew~J. Collins, Adriana~E. Lita, Thomas Gerrits, Sae~Woo Nam, Varun~D. Vaidya, Matteo Menotti, Ish Dhand, Zachary Vernon, Nicol{\'a}s Quesada, and Jonathan Lavoie.
\newblock Quantum computational advantage with a programmable photonic processor.
\newblock \emph{Nature}, 606\penalty0 (7912):\penalty0 75--81, June 2022.
\newblock ISSN 1476-4687.
\newblock \doi{10.1038/s41586-022-04725-x}.

\bibitem[Main et~al.(2025)Main, Drmota, Nadlinger, Ainley, Agrawal, Nichol, Srinivas, Araneda, and Lucas]{main_distributed_2025}
D.~Main, P.~Drmota, D.~P. Nadlinger, E.~M. Ainley, A.~Agrawal, B.~C. Nichol, R.~Srinivas, G.~Araneda, and D.~M. Lucas.
\newblock Distributed quantum computing across an optical network link.
\newblock \emph{Nature}, 638\penalty0 (8050):\penalty0 383--388, February 2025.
\newblock ISSN 1476-4687.
\newblock \doi{10.1038/s41586-024-08404-x}.

\bibitem[Malarz(2020)]{malarz_site_2020}
Krzysztof Malarz.
\newblock Site percolation thresholds on triangular lattice with complex neighborhoods.
\newblock \emph{Chaos: An Interdisciplinary Journal of Nonlinear Science}, 30\penalty0 (12):\penalty0 123123, December 2020.
\newblock ISSN 1054-1500.
\newblock \doi{10.1063/5.0022336}.

\bibitem[McElvanney and Backens(2023{\natexlab{a}})]{mcelvanneyCompleteFlowpreservingRewrite2023}
Tommy McElvanney and Miriam Backens.
\newblock Complete flow-preserving rewrite rules for {{MBQC}} patterns with pauli measurements.
\newblock In Stefano Gogioso and Matty Hoban, editors, \emph{Proceedings 19th International Conference on Quantum Physics and Logic, Wolfson College, Oxford, {{UK}}, 27 June - 1 July 2022}, volume 394 of \emph{Electronic Proceedings in Theoretical Computer Science}, pages 66--82. Open Publishing Association, 2023{\natexlab{a}}.
\newblock \doi{10.4204/EPTCS.394.5}.

\bibitem[McElvanney and Backens(2023{\natexlab{b}})]{mcelvanneyFlowpreservingZXcalculusRewrite2023}
Tommy McElvanney and Miriam Backens.
\newblock Flow-preserving {{ZX-calculus}} rewrite rules for optimisation and obfuscation.
\newblock In Shane Mansfield, Benoit Val{\^i}ron, and Vladimir Zamdzhiev, editors, \emph{Proceedings of the Twentieth International Conference on Quantum Physics and Logic, Paris, France, 17-21st July 2023}, volume 384 of \emph{Electronic Proceedings in Theoretical Computer Science}, pages 203--219. Open Publishing Association, 2023{\natexlab{b}}.
\newblock \doi{10.4204/EPTCS.384.12}.

\bibitem[Mhalla and Perdrix(2012)]{mhalla_graph_2012}
Mehdi Mhalla and Simon Perdrix.
\newblock Graph {States}, {Pivot} {Minor}, and {Universality} of ({X},{Z})-measurements, February 2012.
\newblock URL \url{https://doi.org/10.48550/arXiv.1202.6551}.

\bibitem[Miranowicz et~al.()Miranowicz, Leoński, and Imoto]{miranowicz_quantum-optical_2001}
Adam Miranowicz, Wiesław Leoński, and Nobuyuki Imoto.
\newblock Quantum-optical states in finite-dimensional hilbert space. i. general formalism.
\newblock In \emph{Modern Nonlinear Optics}, pages 155--193. John Wiley \& Sons, Ltd.
\newblock ISBN 978-0-471-23147-9.
\newblock \doi{10.1002/0471231479.ch3}.

\bibitem[Monroe et~al.(2014)Monroe, Raussendorf, Ruthven, Brown, Maunz, Duan, and Kim]{monroe_large_2014}
C.~Monroe, R.~Raussendorf, A.~Ruthven, K.~R. Brown, P.~Maunz, L.-M. Duan, and J.~Kim.
\newblock Large {Scale} {Modular} {Quantum} {Computer} {Architecture} with {Atomic} {Memory} and {Photonic} {Interconnects}.
\newblock \emph{Physical Review A}, 89\penalty0 (2):\penalty0 022317, February 2014.
\newblock ISSN 1050-2947, 1094-1622.
\newblock \doi{10.1103/PhysRevA.89.022317}.

\bibitem[Motes et~al.(2014{\natexlab{a}})Motes, Gilchrist, Dowling, and Rohde]{motesScalableBoson2014}
Keith~R. Motes, Alexei Gilchrist, Jonathan~P. Dowling, and Peter~P. Rohde.
\newblock Scalable {{Boson Sampling}} with {{Time-Bin Encoding Using}} a {{Loop-Based Architecture}}.
\newblock \emph{Physical Review Letters}, 113\penalty0 (12):\penalty0 120501, September 2014{\natexlab{a}}.
\newblock \doi{10.1103/PhysRevLett.113.120501}.

\bibitem[Motes et~al.(2014{\natexlab{b}})Motes, Gilchrist, Dowling, and Rohde]{motes_scalable_2014}
Keith~R. Motes, Alexei Gilchrist, Jonathan~P. Dowling, and Peter~P. Rohde.
\newblock Scalable {Boson} {Sampling} with {Time}-{Bin} {Encoding} {Using} a {Loop}-{Based} {Architecture}.
\newblock \emph{Physical Review Letters}, 113\penalty0 (12):\penalty0 120501, September 2014{\natexlab{b}}.
\newblock \doi{10.1103/PhysRevLett.113.120501}.

\bibitem[Murray et~al.(1993)Murray, Norris, and Bawendi]{murraySynthesisCharacterizationNearly1993}
C.~B. Murray, D.~J. Norris, and M.~G. Bawendi.
\newblock Synthesis and characterization of nearly monodisperse {{CdE}} ({{E}} = sulfur, selenium, tellurium) semiconductor nanocrystallites.
\newblock \emph{Journal of the American Chemical Society}, 115\penalty0 (19):\penalty0 8706--8715, September 1993.
\newblock ISSN 0002-7863.
\newblock \doi{10.1021/ja00072a025}.

\bibitem[Nest et~al.(2006)Nest, Miyake, Dür, and Briegel]{nest_universal_2006}
Maarten Van~den Nest, Akimasa Miyake, Wolfgang Dür, and Hans~J. Briegel.
\newblock Universal resources for measurement-based quantum computation.
\newblock \emph{Physical Review Letters}, 97\penalty0 (15):\penalty0 150504, October 2006.
\newblock ISSN 0031-9007, 1079-7114.
\newblock \doi{10.1103/PhysRevLett.97.150504}.

\bibitem[Ng and Wang(2017)]{ngUniversalCompletionZXcalculus2017}
Kang~Feng Ng and Quanlong Wang.
\newblock A universal completion of the {{ZX-calculus}}, June 2017.

\bibitem[Nigmatullin et~al.(2016)Nigmatullin, Ballance, Beaudrap, and Benjamin]{nigmatullin_minimally_2016}
Ramil Nigmatullin, Christopher~J Ballance, Niel~de Beaudrap, and Simon~C Benjamin.
\newblock Minimally complex ion traps as modules for quantum communication and computing.
\newblock \emph{New Journal of Physics}, 18\penalty0 (10):\penalty0 103028, October 2016.
\newblock ISSN 1367-2630.
\newblock \doi{10.1088/1367-2630/18/10/103028}.

\bibitem[Ollivier and Tillich(2003)]{ollivier_description_2003}
Harold Ollivier and Jean-Pierre Tillich.
\newblock Description of a {Quantum} {Convolutional} {Code}.
\newblock \emph{Physical Review Letters}, 91\penalty0 (17):\penalty0 177902, October 2003.
\newblock \doi{10.1103/PhysRevLett.91.177902}.

\bibitem[Oszmaniec and Brod(2018)]{oszmaniec_classical_2018}
Michał Oszmaniec and Daniel~J. Brod.
\newblock Classical simulation of photonic linear optics with lost particles.
\newblock \emph{New Journal of Physics}, 20\penalty0 (9):\penalty0 092002, September 2018.
\newblock ISSN 1367-2630.
\newblock \doi{10.1088/1367-2630/aadfa8}.

\bibitem[Ozdemir et~al.(2001)Ozdemir, Miranowicz, Koashi, and Imoto]{ozdemir_quantum-scissors_2001}
Sahin~Kaya Ozdemir, Adam Miranowicz, Masato Koashi, and Nobuyuki Imoto.
\newblock Quantum-scissors device for optical state truncation: {A} proposal for practical realization.
\newblock \emph{Physical Review A}, 64\penalty0 (6):\penalty0 063818, November 2001.
\newblock ISSN 1050-2947, 1094-1622.
\newblock \doi{10.1103/PhysRevA.64.063818}.

\bibitem[Paesani and Brown(2023)]{paesaniHighThresholdQuantumComputing2023}
Stefano Paesani and Benjamin~J. Brown.
\newblock High-{{Threshold Quantum Computing}} by {{Fusing One-Dimensional Cluster States}}.
\newblock \emph{Physical Review Letters}, 131\penalty0 (12):\penalty0 120603, September 2023.
\newblock \doi{10.1103/PhysRevLett.131.120603}.

\bibitem[Pan and Zeilinger(1998)]{panGreenbergerHorneZeilingerstateAnalyzer1998}
Jian-wei Pan and Anton Zeilinger.
\newblock Greenberger-{{Horne-Zeilinger-state}} analyzer.
\newblock \emph{Physical Review A}, 57\penalty0 (3):\penalty0 2208--2211, March 1998.
\newblock \doi{10.1103/PhysRevA.57.2208}.

\bibitem[Pankovich et~al.(2023)Pankovich, Neville, Kan, Omkar, Wan, and Br{\'a}dler]{pankovichFlexibleEntangledState2023}
Brendan Pankovich, Alex Neville, Angus Kan, Srikrishna Omkar, Kwok~Ho Wan, and Kamil Br{\'a}dler.
\newblock Flexible entangled state generation in linear optics, October 2023.
\newblock URL \url{https://doi.org/10.48550/arXiv.2310.06832}.

\bibitem[Pegg et~al.(1998)Pegg, Phillips, and Barnett]{pegg_optical_1998}
David~T. Pegg, Lee~S. Phillips, and Stephen~M. Barnett.
\newblock Optical {State} {Truncation} by {Projection} {Synthesis}.
\newblock \emph{Physical Review Letters}, 81\penalty0 (8):\penalty0 1604--1606, August 1998.
\newblock ISSN 0031-9007.
\newblock \doi{10.1103/PhysRevLett.81.1604}.

\bibitem[Power and Robinson(1997)]{power_premonoidal_1997}
John Power and Edmund Robinson.
\newblock Premonoidal categories and notions of computation.
\newblock \emph{Mathematical. Structures in Comp. Sci.}, 7\penalty0 (5):\penalty0 453--468, October 1997.
\newblock ISSN 0960-1295.
\newblock \doi{10.1017/S0960129597002375}.

\bibitem[Poór et~al.()Poór, Shaikh, and Wang]{poorZXcalculusCompleteFinitedimensional2024}
Boldizsár Poór, Razin~A. Shaikh, and Quanlong Wang.
\newblock {ZX}-calculus is complete for finite-dimensional hilbert spaces.
\newblock 426:\penalty0 127--158.
\newblock ISSN 2075-2180.
\newblock \doi{10.4204/EPTCS.426.5}.

\bibitem[Prevedel et~al.(2007)Prevedel, Walther, Tiefenbacher, Böhi, Kaltenbaek, Jennewein, and Zeilinger]{prevedel_high-speed_2007}
Robert Prevedel, Philip Walther, Felix Tiefenbacher, Pascal Böhi, Rainer Kaltenbaek, Thomas Jennewein, and Anton Zeilinger.
\newblock High-speed linear optics quantum computing using active feed-forward.
\newblock \emph{Nature}, 445\penalty0 (7123):\penalty0 65--69, January 2007.
\newblock ISSN 0028-0836, 1476-4687.
\newblock \doi{10.1038/nature05346}.

\bibitem[Reck et~al.(1994)Reck, Zeilinger, Bernstein, and Bertani]{reck_experimental_1994}
Michael Reck, Anton Zeilinger, Herbert~J. Bernstein, and Philip Bertani.
\newblock Experimental realization of any discrete unitary operator.
\newblock \emph{Physical Review Letters}, 73\penalty0 (1):\penalty0 58--61, July 1994.
\newblock \doi{10.1103/PhysRevLett.73.58}.

\bibitem[Rodatz et~al.()Rodatz, Poór, and Kissinger]{rodatzFloquetifyingStabiliser2024}
Benjamin Rodatz, Boldizsár Poór, and Aleks Kissinger.
\newblock Floquetifying stabiliser codes with distance-preserving rewrites.
\newblock URL \url{https://doi.org/10.48550/arXiv.2410.17240}.

\bibitem[Rodatz et~al.(2025)Rodatz, Po{\'o}r, and Kissinger]{rodatzFaultTolerance2025}
Benjamin Rodatz, Boldizs{\'a}r Po{\'o}r, and Aleks Kissinger.
\newblock Fault {{Tolerance}} by {{Construction}}, June 2025.
\newblock URL \url{https://doi.org/10.48550/arXiv.2506.17181}.

\bibitem[Rom{\'a}n(2020)]{romanOpenDiagramsCoend2020}
Mario Rom{\'a}n.
\newblock Open {{Diagrams}} via {{Coend Calculus}}.
\newblock In David~I. Spivak and Jamie Vicary, editors, \emph{Proceedings of the 3rd Annual International Applied Category Theory Conference 2020, {{ACT}} 2020, Cambridge, {{USA}}, 6-10th July 2020}, volume 333 of \emph{Electronic {{Proceedings}} in {{Theoretical Computer Science}}}, pages 65--78. Open Publishing Association, 2020.
\newblock \doi{10.4204/EPTCS.333.5}.

\bibitem[Román and Sobociński()]{roman_string_2025}
Mario Román and Paweł Sobociński.
\newblock String diagrams for premonoidal categories.
\newblock Volume 21, Issue 2.
\newblock ISSN 1860-5974.
\newblock \doi{10.46298/lmcs-21(2:9)2025}.

\bibitem[Rossetti et~al.(1983)Rossetti, Nakahara, and Brus]{rossettiQuantumSizeEffects1983}
R.~Rossetti, S.~Nakahara, and L.~E. Brus.
\newblock Quantum size effects in the redox potentials, resonance {{Raman}} spectra, and electronic spectra of {{CdS}} crystallites in aqueous solution.
\newblock \emph{The Journal of Chemical Physics}, 79\penalty0 (2):\penalty0 1086--1088, July 1983.
\newblock ISSN 0021-9606.
\newblock \doi{10.1063/1.445834}.

\bibitem[Schwartz et~al.(2016)Schwartz, Cogan, Schmidgall, Don, Gantz, Kenneth, Lindner, and Gershoni]{schwartzDeterministicGenerationCluster2016}
I.~Schwartz, D.~Cogan, E.~R. Schmidgall, Y.~Don, L.~Gantz, O.~Kenneth, N.~H. Lindner, and D.~Gershoni.
\newblock Deterministic generation of a cluster state of entangled photons.
\newblock \emph{Science}, 354\penalty0 (6311):\penalty0 434--437, October 2016.
\newblock \doi{10.1126/science.aah4758}.

\bibitem[Simmons(2021)]{simmonsRelatingMeasurementPatterns2021}
Will Simmons.
\newblock Relating measurement patterns to circuits via pauli flow.
\newblock In Chris Heunen and Miriam Backens, editors, \emph{Proceedings 18th International Conference on Quantum Physics and Logic}, volume 343 of \emph{Electronic Proceedings in Theoretical Computer Science}, pages 50--101, Gdansk, Poland, 2021. Open Publishing Association.
\newblock \doi{10.4204/EPTCS.343.4}.

\bibitem[Stanisic et~al.(2017)Stanisic, Linden, Montanaro, and Turner]{stanisicGeneratingEntanglementLinear2017}
Stasja Stanisic, Noah Linden, Ashley Montanaro, and Peter~S. Turner.
\newblock Generating entanglement with linear optics.
\newblock \emph{Physical Review A}, 96\penalty0 (4):\penalty0 043861, October 2017.
\newblock \doi{10.1103/PhysRevA.96.043861}.

\bibitem[Sunami et~al.(2025)Sunami, Tamiya, Inoue, Yamasaki, and Goban]{sunami_scalable_2025}
Shinichi Sunami, Shiro Tamiya, Ryotaro Inoue, Hayata Yamasaki, and Akihisa Goban.
\newblock Scalable {Networking} of {Neutral}-{Atom} {Qubits}: {Nanofiber}-{Based} {Approach} for {Multiprocessor} {Fault}-{Tolerant} {Quantum} {Computers}.
\newblock \emph{PRX Quantum}, 6\penalty0 (1):\penalty0 010101, February 2025.
\newblock \doi{10.1103/PRXQuantum.6.010101}.

\bibitem[Thomas et~al.(2022)Thomas, Ruscio, Morin, and Rempe]{thomasEfficientGenerationEntangled2022}
Philip Thomas, Leonardo Ruscio, Olivier Morin, and Gerhard Rempe.
\newblock Efficient generation of entangled multiphoton graph states from a single atom.
\newblock \emph{Nature}, 608\penalty0 (7924):\penalty0 677--681, August 2022.
\newblock ISSN 1476-4687.
\newblock \doi{10.1038/s41586-022-04987-5}.

\bibitem[Toumi et~al.()Toumi, de~Felice, and Yeung]{toumi_discopy_2022}
Alexis Toumi, Giovanni de~Felice, and Richie Yeung.
\newblock {DisCoPy} for the quantum computer scientist.
\newblock URL \url{https://doi.org/10.48550/arXiv.2205.05190}.

\bibitem[{Townsend-Teague} et~al.(2023){Townsend-Teague}, {Magdalena de la Fuente}, and Kesselring]{townsend-teagueFloquetifyingColourCode2023}
Alex {Townsend-Teague}, Julio {Magdalena de la Fuente}, and Markus Kesselring.
\newblock Floquetifying the colour code.
\newblock In Shane Mansfield, Benoit Val{\^i}ron, and Vladimir Zamdzhiev, editors, \emph{Proceedings of the Twentieth International Conference on Quantum Physics and Logic}, volume 384 of \emph{Electronic Proceedings in Theoretical Computer Science}, pages 265--303. Open Publishing Association, 2023.
\newblock \doi{10.4204/EPTCS.384.14}.

\bibitem[Vilmart(2019)]{vilmartNearMinimalAxiomatisationZXCalculus2019}
Renaud Vilmart.
\newblock A {{Near-Minimal Axiomatisation}} of {{ZX-Calculus}} for {{Pure Qubit Quantum Mechanics}}.
\newblock In \emph{2019 34th {{Annual ACM}}/{{IEEE Symposium}} on {{Logic}} in {{Computer Science}} ({{LICS}})}, pages 1--10, June 2019.
\newblock \doi{10.1109/LICS.2019.8785765}.

\bibitem[Wetering et~al.()Wetering, Yeung, Laakkonen, and Kissinger]{vandeweteringOptimalCompilationParametrised2024}
John van~de Wetering, Richie Yeung, Tuomas Laakkonen, and Aleks Kissinger.
\newblock Optimal compilation of parametrised quantum circuits.
\newblock 9:\penalty0 1828.
\newblock \doi{10.22331/q-2025-08-27-1828}.

\bibitem[Wilde(2009)]{wilde_quantum-shift-register_2009}
Mark~M. Wilde.
\newblock Quantum-shift-register circuits.
\newblock \emph{Physical Review A}, 79\penalty0 (6):\penalty0 062325, June 2009.
\newblock \doi{10.1103/PhysRevA.79.062325}.

\bibitem[Yard et~al.()Yard, Jones, Paesani, Maïnos, Bulmer, and Laing]{yardOnchipQuantumInformation2022}
Patrick Yard, Alex~E. Jones, Stefano Paesani, Alexandre Maïnos, Jacob~F. Bulmer, and Anthony Laing.
\newblock On-chip quantum information processing with distinguishable photons.
\newblock 132\penalty0 (15):\penalty0 150602.
\newblock \doi{10.1103/PhysRevLett.132.150602}.

\bibitem[Zanin et~al.(2021)Zanin, Jacquet, Spagnolo, Schiansky, Calafell, Rozema, and Walther]{zaninFibercompatiblePhotonicFeedforward2021}
Guilherme~Luiz Zanin, Maxime~J. Jacquet, Michele Spagnolo, Peter Schiansky, Irati~Alonso Calafell, Lee~A. Rozema, and Philip Walther.
\newblock Fiber-compatible photonic feed-forward with 99\% fidelity.
\newblock \emph{Optics Express}, 29\penalty0 (3):\penalty0 3425--3437, February 2021.
\newblock ISSN 1094-4087.
\newblock \doi{10.1364/OE.409867}.

\bibitem[Zhang et~al.(2023)Zhang, Wu, Wang, Li, Shapourian, Shabani, and Ding]{zhangOneQCompilationFramework2023}
Hezi Zhang, Anbang Wu, Yuke Wang, Gushu Li, Hassan Shapourian, Alireza Shabani, and Yufei Ding.
\newblock {{OneQ}}: {{A Compilation Framework}} for {{Photonic One-Way Quantum Computation}}.
\newblock In \emph{Proceedings of the 50th {{Annual International Symposium}} on {{Computer Architecture}}}, pages 1--14, Orlando FL USA, June 2023. ACM.
\newblock ISBN 9798400700958.
\newblock \doi{10.1145/3579371.3589047}.

\bibitem[Zilk et~al.(2022)Zilk, Staudacher, Guggemos, F{\"u}rlinger, Kranzlm{\"u}ller, and Walther]{zilkCompilerUniversalPhotonic2022}
Felix Zilk, Korbinian Staudacher, Tobias Guggemos, Karl F{\"u}rlinger, Dieter Kranzlm{\"u}ller, and Philip Walther.
\newblock A compiler for universal photonic quantum computers.
\newblock In \emph{2022 {{IEEE}}/{{ACM Third International Workshop}} on {{Quantum Computing Software}} ({{QCS}})}, pages 57--67, November 2022.
\newblock \doi{10.1109/QCS56647.2022.00012}.

\end{thebibliography}

\appendix

\section{Rewrite rules used in the paper}\label{sec:qpath}

This section contains a selection of rewrite rules used in the paper.

\paragraph{ZX calculus}
The axioms of the ZX calculus used in the paper are shown below.

\begin{figure}[H]
  \[
    \input{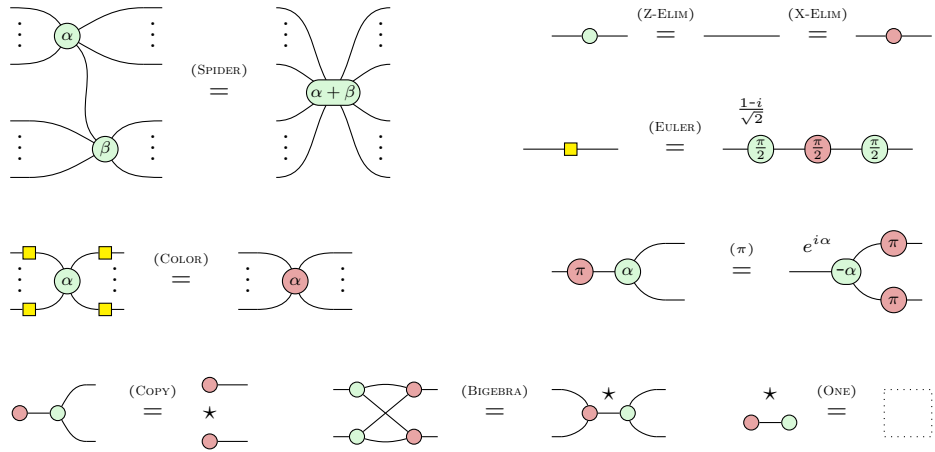}
  \]
  \caption{Axioms of the ZX calculus.}
  \label{fig:ZX-axioms}
\end{figure}

\paragraph{QPath}

We review the axioms of the QPath calculus~\cite{defeliceQuantumLinearOptics2023}, which are used in \cref{sec:mixedfusion}.
Diagrams in \textbf{QPath} are generated by:
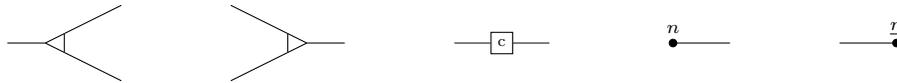
\[
  \begin{tikzpicture}
	\begin{pgfonlayer}{nodelayer}
		\node [style=none] (0) at (-0.5, 0) {};
		\node [style=none] (1) at (0, 0.25) {};
		\node [style=none] (2) at (0, -0.25) {};
		\node [style=none] (3) at (-1.5, 0) {};
		\node [style=none] (4) at (1.5, 1) {};
		\node [style=none] (5) at (1.5, -1) {};
	\end{pgfonlayer}
	\begin{pgfonlayer}{edgelayer}
		\draw (0.center) to (1.center);
		\draw (1.center) to (2.center);
		\draw (2.center) to (0.center);
		\draw (3.center) to (0.center);
		\draw (1.center) to (4.center);
		\draw (2.center) to (5.center);
	\end{pgfonlayer}
\end{tikzpicture}

  \qquad \qquad
  \begin{tikzpicture}
	\begin{pgfonlayer}{nodelayer}
		\node [style=none] (6) at (0.5, 0) {};
		\node [style=none] (7) at (0, 0.25) {};
		\node [style=none] (8) at (0, -0.25) {};
		\node [style=none] (9) at (-1.5, 1) {};
		\node [style=none] (10) at (-1.5, -1) {};
		\node [style=none] (11) at (1.5, 0) {};
	\end{pgfonlayer}
	\begin{pgfonlayer}{edgelayer}
		\draw (6.center) to (7.center);
		\draw (7.center) to (8.center);
		\draw (8.center) to (6.center);
		\draw (9.center) to (7.center);
		\draw (6.center) to (11.center);
		\draw (8.center) to (10.center);
	\end{pgfonlayer}
\end{tikzpicture}

  \qquad \qquad
  \begin{tikzpicture}
	\begin{pgfonlayer}{nodelayer}
		\node [style=none] (0) at (-1.25, 0) {};
		\node [style=none] (1) at (1.25, 0) {};
		\node [style=phase] (2) at (0, 0) {c};
	\end{pgfonlayer}
	\begin{pgfonlayer}{edgelayer}
		\draw [style=LOWire] (1.center) to (2);
		\draw [style=LOWire] (2) to (0.center);
	\end{pgfonlayer}
\end{tikzpicture}

  \qquad \qquad
  \begin{tikzpicture}
	\begin{pgfonlayer}{nodelayer}
		\node [style=black node] (0) at (-0.75, 0) {};
		\node [style=none] (1) at (0.75, 0) {};
		\node [style=sLabel] (2) at (-0.75, 0.375) {$n$};
	\end{pgfonlayer}
	\begin{pgfonlayer}{edgelayer}
		\draw (1.center) to (0);
	\end{pgfonlayer}
\end{tikzpicture}

  \qquad \qquad
  \begin{tikzpicture}
	\begin{pgfonlayer}{nodelayer}
		\node [style=black node] (0) at (0.75, 0) {};
		\node [style=none] (1) at (-0.75, 0) {};
		\node [style=sLabel] (2) at (0.75, 0.375) {$\cvar n$};
	\end{pgfonlayer}
	\begin{pgfonlayer}{edgelayer}
		\draw (1.center) to (0);
	\end{pgfonlayer}
\end{tikzpicture}

\]
for all $n \in \N$ and $c \in \C$.
The \textbf{QPath} calculus admits the following graphical rewrite rules.
Additionally, all rules hold under transposition of the linear maps, which is represented by horizontal reflection of the diagrams.
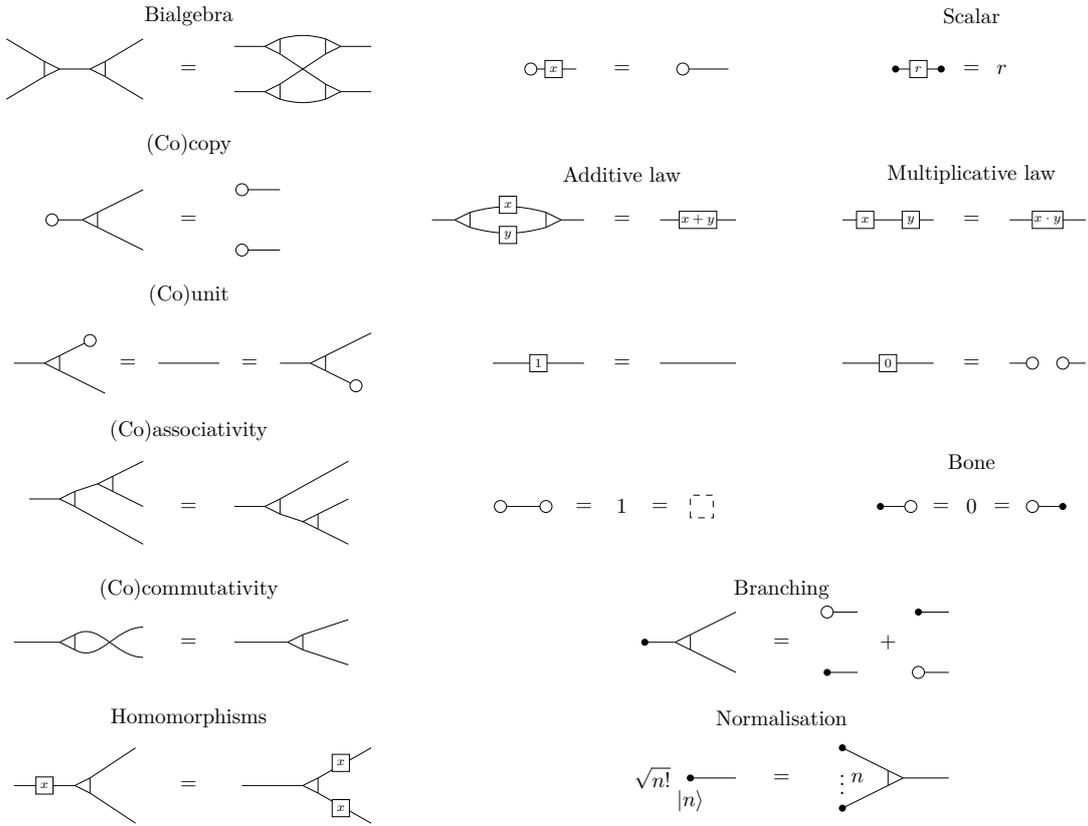
\begin{equation*}
    \scalebox{0.8}{\begin{tikzpicture}
	\begin{pgfonlayer}{nodelayer}
		\node [style=none] (215) at (-11.5, -6.5) {};
		\node [style=none] (12) at (-11.25, -1.5) {};
		\node [style=none] (13) at (-11.75, -1.25) {};
		\node [style=none] (14) at (-11.75, -1.75) {};
		\node [style=none] (15) at (-13, -0.5) {};
		\node [style=none] (16) at (-13, -2.5) {};
		\node [style=none] (17) at (-10.25, -1.5) {};
		\node [style=none] (18) at (-10.25, -1.5) {};
		\node [style=none] (19) at (-9.75, -1.25) {};
		\node [style=none] (20) at (-9.75, -1.75) {};
		\node [style=none] (21) at (-8.5, -0.5) {};
		\node [style=none] (22) at (-8.5, -2.5) {};
		\node [style=none] (23) at (-4.5, -0.75) {};
		\node [style=none] (24) at (-4, -0.5) {};
		\node [style=none] (25) at (-4, -1) {};
		\node [style=none] (26) at (-5.5, -0.75) {};
		\node [style=none] (32) at (-2, -0.75) {};
		\node [style=none] (33) at (-2.5, -0.5) {};
		\node [style=none] (34) at (-2.5, -1) {};
		\node [style=none] (35) at (-4.5, -2.25) {};
		\node [style=none] (36) at (-4, -2) {};
		\node [style=none] (37) at (-4, -2.5) {};
		\node [style=none] (38) at (-5.5, -2.25) {};
		\node [style=none] (39) at (-2, -2.25) {};
		\node [style=none] (40) at (-2.5, -2) {};
		\node [style=none] (41) at (-2.5, -2.5) {};
		\node [style=none] (42) at (-1, -0.75) {};
		\node [style=none] (43) at (-1, -2.25) {};
		\node [style=none] (44) at (-7, -1.5) {$=$};
		\node [style=none] (57) at (1.75, -6.5) {};
		\node [style=none] (58) at (2.25, -6.25) {};
		\node [style=none] (59) at (2.25, -6.75) {};
		\node [style=none] (60) at (1, -6.5) {};
		\node [style=none] (63) at (5.25, -6.5) {};
		\node [style=none] (64) at (4.75, -6.25) {};
		\node [style=none] (65) at (4.75, -6.75) {};
		\node [style=none] (68) at (6, -6.5) {};
		\node [style=none] (70) at (7.25, -6.5) {$=$};
		\node [style=none] (110) at (8.5, -6.5) {};
		\node [style=none] (111) at (11, -6.5) {};
		\node [style=none] (131) at (18.75, -6.5) {$=$};
		\node [style=none] (133) at (20, -6.5) {};
		\node [style=none] (134) at (22.5, -6.5) {};
		\node [style=none] (136) at (14.5, -6.5) {};
		\node [style=none] (137) at (17.5, -6.5) {};
		\node [style=none] (138) at (-12.75, -25.25) {};
		\node [style=none] (139) at (-10.75, -25.25) {};
		\node [style=none] (140) at (-10.75, -25.25) {};
		\node [style=none] (141) at (-10.25, -25) {};
		\node [style=none] (142) at (-10.25, -25.5) {};
		\node [style=none] (143) at (-8.75, -24) {};
		\node [style=none] (144) at (-8.75, -26.5) {};
		\node [style=none] (146) at (-7, -25.25) {$=$};
		\node [style=none] (147) at (-5.25, -25.25) {};
		\node [style=none] (148) at (-3.25, -25.25) {};
		\node [style=none] (149) at (-3.25, -25.25) {};
		\node [style=none] (150) at (-2.75, -25) {};
		\node [style=none] (151) at (-2.75, -25.5) {};
		\node [style=none] (152) at (-1, -24) {};
		\node [style=none] (153) at (-1, -26.5) {};
		\node [style=none] (175) at (-12.75, -11.25) {};
		\node [style=none] (176) at (-11.75, -11.25) {};
		\node [style=none] (177) at (-11.75, -11.25) {};
		\node [style=none] (178) at (-11.25, -11) {};
		\node [style=none] (179) at (-11.25, -11.5) {};
		\node [style=none] (180) at (-10.25, -10.5) {};
		\node [style=none] (181) at (-9.75, -12.25) {};
		\node [style=none] (182) at (-9, -11.25) {$=$};
		\node [style=none] (190) at (-5, -11.25) {$=$};
		\node [style=wn] (198) at (-10.25, -10.5) {};
		\node [style=wn] (199) at (-11.5, -6.5) {};
		\node [style=wn] (200) at (-5.25, -5.5) {};
		\node [style=none] (202) at (-8, -11.25) {};
		\node [style=none] (203) at (-6, -11.25) {};
		\node [style=none] (204) at (-4, -11.25) {};
		\node [style=none] (205) at (-3, -11.25) {};
		\node [style=none] (206) at (-3, -11.25) {};
		\node [style=none] (207) at (-2.5, -11) {};
		\node [style=none] (208) at (-2.5, -11.5) {};
		\node [style=none] (209) at (-1, -10.25) {};
		\node [style=none] (210) at (-1.5, -12) {};
		\node [style=wn] (211) at (-1.5, -12) {};
		\node [style=none] (214) at (-9, -7.25) {};
		\node [style=none] (216) at (-10.5, -6.5) {};
		\node [style=none] (217) at (-10.5, -6.5) {};
		\node [style=none] (218) at (-10, -6.25) {};
		\node [style=none] (219) at (-10, -6.75) {};
		\node [style=none] (220) at (-8.5, -5.5) {};
		\node [style=none] (221) at (-8.5, -7.5) {};
		\node [style=wn] (222) at (-5.25, -7.5) {};
		\node [style=none] (223) at (-7, -6.5) {$=$};
		\node [style=none] (224) at (-4, -5.5) {};
		\node [style=none] (225) at (-4, -7.5) {};
		\node [style=none] (239) at (4.25, -1.5) {};
		\node [style=none] (240) at (5.75, -1.5) {};
		\node [style=wn] (241) at (4.25, -1.5) {};
		\node [style=none] (242) at (7.25, -1.5) {$=$};
		\node [style=none] (243) at (9.25, -1.5) {};
		\node [style=none] (244) at (10.75, -1.5) {};
		\node [style=wn] (245) at (9.25, -1.5) {};
		\node [style=none] (265) at (-7, 0.25) {Bialgebra};
		\node [style=none] (266) at (-7, -9) {(Co)unit};
		\node [style=none] (267) at (7.25, -5) {Additive law};
		\node [style=none] (268) at (18.75, -5) {Multiplicative law};
		\node [style=none] (269) at (-7, -23) {Homomorphisms};
		\node [style=none] (270) at (-12.25, -15.75) {};
		\node [style=none] (271) at (-11.25, -15.75) {};
		\node [style=none] (272) at (-11.25, -15.75) {};
		\node [style=none] (273) at (-10.75, -15.5) {};
		\node [style=none] (274) at (-10.75, -16) {};
		\node [style=none] (275) at (-10, -15.25) {};
		\node [style=none] (276) at (-8.5, -17.25) {};
		\node [style=none] (278) at (-10, -15.25) {};
		\node [style=none] (279) at (-10, -15.25) {};
		\node [style=none] (280) at (-9.5, -15) {};
		\node [style=none] (281) at (-9.5, -15.5) {};
		\node [style=none] (282) at (-8.5, -14.5) {};
		\node [style=none] (283) at (-8.5, -16) {};
		\node [style=none] (284) at (-5.5, -16) {};
		\node [style=none] (285) at (-4.5, -16) {};
		\node [style=none] (286) at (-4.5, -16) {};
		\node [style=none] (287) at (-4, -15.75) {};
		\node [style=none] (288) at (-4, -16.25) {};
		\node [style=none] (289) at (-1.75, -14.5) {};
		\node [style=none] (290) at (-3.25, -16.5) {};
		\node [style=none] (292) at (-3.25, -16.5) {};
		\node [style=none] (293) at (-2.75, -16.25) {};
		\node [style=none] (294) at (-2.75, -16.75) {};
		\node [style=none] (295) at (-1.75, -15.75) {};
		\node [style=none] (296) at (-1.75, -17.25) {};
		\node [style=none] (297) at (-7, -16) {$=$};
		\node [style=none] (323) at (-12.75, -20.5) {};
		\node [style=none] (324) at (-11.25, -20.5) {};
		\node [style=none] (325) at (-11.25, -20.5) {};
		\node [style=none] (326) at (-10.75, -20.25) {};
		\node [style=none] (327) at (-10.75, -20.75) {};
		\node [style=none] (328) at (-8.5, -21) {};
		\node [style=none] (329) at (-8.5, -20) {};
		\node [style=none] (330) at (-5.5, -20.5) {};
		\node [style=none] (331) at (-3.75, -20.5) {};
		\node [style=none] (332) at (-3.75, -20.5) {};
		\node [style=none] (333) at (-3.25, -20.25) {};
		\node [style=none] (334) at (-3.25, -20.75) {};
		\node [style=none] (335) at (-1.75, -19.75) {};
		\node [style=none] (336) at (-1.75, -21.25) {};
		\node [style=none] (337) at (-7, -20.5) {$=$};
		\node [style=none] (377) at (-7, -4) {(Co)copy};
		\node [style=none] (378) at (-7, -13.5) {(Co)associativity};
		\node [style=none] (379) at (-7, -18.75) {(Co)commutativity};
		\node [style=none] (380) at (3.25, -16) {};
		\node [style=none] (381) at (4.75, -16) {};
		\node [style=wn] (382) at (3.25, -16) {};
		\node [style=wn] (383) at (4.75, -16) {};
		\node [style=none] (384) at (6, -16) {$=$};
		\node [style=none] (414) at (14.5, -11.25) {};
		\node [style=none] (415) at (17.5, -11.25) {};
		\node [style=none] (417) at (18.75, -11.25) {=};
		\node [style=none] (418) at (20, -11.25) {};
		\node [style=none] (419) at (20.75, -11.25) {};
		\node [style=wn] (420) at (20.75, -11.25) {};
		\node [style=none] (421) at (21.75, -11.25) {};
		\node [style=none] (422) at (22.5, -11.25) {};
		\node [style=wn] (423) at (21.75, -11.25) {};
		\node [style=none] (424) at (3, -11.25) {};
		\node [style=none] (425) at (6, -11.25) {};
		\node [style=none] (427) at (8.5, -11.25) {};
		\node [style=none] (428) at (11, -11.25) {};
		\node [style=none] (429) at (7.25, -11.25) {$=$};
		\node [style=none] (430) at (9, -20.5) {};
		\node [style=none] (431) at (9.5, -20.25) {};
		\node [style=none] (432) at (9.5, -20.75) {};
		\node [style=none] (433) at (11, -19.5) {};
		\node [style=none] (434) at (11, -21.5) {};
		\node [style=none] (435) at (8, -20.5) {};
		\node [style=black node] (436) at (8, -20.5) {};
		\node [style=none] (437) at (15, -19.5) {};
		\node [style=none] (438) at (14, -19.5) {};
		\node [style=none] (439) at (15, -21.5) {};
		\node [style=none] (440) at (14, -21.5) {};
		\node [style=black node] (441) at (14, -21.5) {};
		\node [style=none] (442) at (12.5, -20.5) {$=$};
		\node [style=wn] (443) at (14, -19.5) {};
		\node [style=none] (444) at (18, -21.5) {};
		\node [style=none] (445) at (17, -21.5) {};
		\node [style=none] (446) at (18, -19.5) {};
		\node [style=none] (447) at (17, -19.5) {};
		\node [style=black node] (448) at (17, -19.5) {};
		\node [style=wn] (449) at (17, -21.5) {};
		\node [style=none] (450) at (16, -20.5) {$+$};
		\node [style=none] (473) at (12.5, -18.75) {Branching};
		\node [style=none] (474) at (16.75, -16) {};
		\node [style=none] (475) at (15.75, -16) {};
		\node [style=black node] (476) at (15.75, -16) {};
		\node [style=wn] (477) at (16.75, -16) {};
		\node [style=none] (478) at (17.75, -16) {$=$};
		\node [style=none] (479) at (18.75, -16) {$0$};
		\node [style=none] (480) at (17.75, -1.5) {};
		\node [style=none] (481) at (16.25, -1.5) {};
		\node [style=black node] (482) at (16.25, -1.5) {};
		\node [style=black node] (483) at (17.75, -1.5) {};
		\node [style=none] (484) at (18.75, -1.5) {$=$};
		\node [style=none] (485) at (19.75, -1.5) {$r$};
		\node [style=none] (488) at (18.75, 0.25) {Scalar};
		\node [style=none] (489) at (21.75, -16) {};
		\node [style=none] (490) at (20.75, -16) {};
		\node [style=black node] (491) at (21.75, -16) {};
		\node [style=wn] (492) at (20.75, -16) {};
		\node [style=none] (493) at (19.75, -16) {$=$};
		\node [style=none] (494) at (12.5, -23) {Normalisation};
		\node [style=none] (495) at (18.75, -14.5) {Bone};
		\node [style=none] (505) at (14.5, -25) {$\vdots$};
		\node [style=none] (506) at (15, -25) {$n$};
		\node [style=none] (510) at (16.5, -25) {};
		\node [style=none] (511) at (16, -24.75) {};
		\node [style=none] (512) at (16, -25.25) {};
		\node [style=none] (513) at (18, -25) {};
		\node [style=black node] (514) at (14.5, -24) {};
		\node [style=none] (515) at (14.5, -24) {};
		\node [style=none] (516) at (14.5, -26) {};
		\node [style=black node] (518) at (14.5, -26) {};
		\node [style=none] (519) at (12.5, -25) {=};
		\node [style=black node] (522) at (9.5, -25) {};
		\node [style=none] (523) at (11, -25) {};
		\node [style=none] (524) at (9.5, -25.75) {$\ket{n}$};
		\node [style=none] (526) at (8.25, -25) {$\sqrt{n!}$};
		\node [style=phase] (530) at (-11.75, -25.25) {$x$};
		\node [style=phase] (531) at (-2, -24.5) {$x$};
		\node [style=phase] (532) at (-2, -26) {$x$};
		\node [style=phase] (537) at (5, -1.5) {$x$};
		\node [style=phase] (538) at (3.5, -6) {$x$};
		\node [style=phase] (540) at (3.5, -7) {$y$};
		\node [style=phase] (541) at (9.75, -6.5) {$x + y$};
		\node [style=phase] (542) at (15.25, -6.5) {$x$};
		\node [style=phase] (543) at (16.75, -6.5) {$y$};
		\node [style=phase] (544) at (21.25, -6.5) {$x \cdot y$};
		\node [style=phase] (545) at (17, -1.5) {$r$};
		\node [style=phase] (546) at (16, -11.25) {$0$};
		\node [style=phase] (547) at (4.5, -11.25) {$1$};
		\node [style=none] (548) at (7.25, -16) {$1$};
		\node [style=none] (549) at (8.5, -16) {$=$};
		\node [style=none] (550) at (9.5, -15.625) {};
		\node [style=none] (551) at (9.5, -16.375) {};
		\node [style=none] (552) at (10.25, -16.375) {};
		\node [style=none] (553) at (10.25, -15.625) {};
	\end{pgfonlayer}
	\begin{pgfonlayer}{edgelayer}
		\draw (12.center) to (13.center);
		\draw (13.center) to (14.center);
		\draw (14.center) to (12.center);
		\draw (15.center) to (13.center);
		\draw (12.center) to (17.center);
		\draw (14.center) to (16.center);
		\draw (18.center) to (19.center);
		\draw (19.center) to (20.center);
		\draw (20.center) to (18.center);
		\draw (19.center) to (21.center);
		\draw (20.center) to (22.center);
		\draw (23.center) to (24.center);
		\draw (24.center) to (25.center);
		\draw (25.center) to (23.center);
		\draw (26.center) to (23.center);
		\draw (32.center) to (33.center);
		\draw (33.center) to (34.center);
		\draw (34.center) to (32.center);
		\draw (35.center) to (36.center);
		\draw (36.center) to (37.center);
		\draw (37.center) to (35.center);
		\draw (38.center) to (35.center);
		\draw (39.center) to (40.center);
		\draw (40.center) to (41.center);
		\draw (41.center) to (39.center);
		\draw [bend left=15] (24.center) to (33.center);
		\draw (25.center) to (40.center);
		\draw (34.center) to (36.center);
		\draw [bend right=15] (37.center) to (41.center);
		\draw (39.center) to (43.center);
		\draw (32.center) to (42.center);
		\draw (57.center) to (58.center);
		\draw (58.center) to (59.center);
		\draw (59.center) to (57.center);
		\draw (60.center) to (57.center);
		\draw (63.center) to (64.center);
		\draw (64.center) to (65.center);
		\draw (65.center) to (63.center);
		\draw (63.center) to (68.center);
		\draw [bend left=15] (58.center) to (64.center);
		\draw [bend right=15] (59.center) to (65.center);
		\draw (110.center) to (111.center);
		\draw (133.center) to (134.center);
		\draw (136.center) to (137.center);
		\draw (138.center) to (139.center);
		\draw (140.center) to (141.center);
		\draw (141.center) to (142.center);
		\draw (142.center) to (140.center);
		\draw (141.center) to (143.center);
		\draw (142.center) to (144.center);
		\draw (147.center) to (148.center);
		\draw (149.center) to (150.center);
		\draw (150.center) to (151.center);
		\draw (151.center) to (149.center);
		\draw (150.center) to (152.center);
		\draw (151.center) to (153.center);
		\draw (175.center) to (176.center);
		\draw (177.center) to (178.center);
		\draw (178.center) to (179.center);
		\draw (179.center) to (177.center);
		\draw (178.center) to (180.center);
		\draw (179.center) to (181.center);
		\draw (202.center) to (203.center);
		\draw (204.center) to (205.center);
		\draw (206.center) to (207.center);
		\draw (207.center) to (208.center);
		\draw (208.center) to (206.center);
		\draw (207.center) to (209.center);
		\draw (208.center) to (210.center);
		\draw (215.center) to (216.center);
		\draw (217.center) to (218.center);
		\draw (218.center) to (219.center);
		\draw (219.center) to (217.center);
		\draw (218.center) to (220.center);
		\draw (219.center) to (221.center);
		\draw (200) to (224.center);
		\draw (222) to (225.center);
		\draw (239.center) to (240.center);
		\draw (243.center) to (244.center);
		\draw (270.center) to (271.center);
		\draw (272.center) to (273.center);
		\draw (273.center) to (274.center);
		\draw (274.center) to (272.center);
		\draw (273.center) to (275.center);
		\draw (274.center) to (276.center);
		\draw (279.center) to (280.center);
		\draw (280.center) to (281.center);
		\draw (281.center) to (279.center);
		\draw (280.center) to (282.center);
		\draw (281.center) to (283.center);
		\draw (284.center) to (285.center);
		\draw (286.center) to (287.center);
		\draw (287.center) to (288.center);
		\draw (288.center) to (286.center);
		\draw (287.center) to (289.center);
		\draw (288.center) to (290.center);
		\draw (292.center) to (293.center);
		\draw (293.center) to (294.center);
		\draw (294.center) to (292.center);
		\draw (293.center) to (295.center);
		\draw (294.center) to (296.center);
		\draw (323.center) to (324.center);
		\draw (325.center) to (326.center);
		\draw (326.center) to (327.center);
		\draw (327.center) to (325.center);
		\draw [in=180, out=30] (326.center) to (328.center);
		\draw [in=-180, out=-30] (327.center) to (329.center);
		\draw (330.center) to (331.center);
		\draw (332.center) to (333.center);
		\draw (333.center) to (334.center);
		\draw (334.center) to (332.center);
		\draw (333.center) to (335.center);
		\draw (334.center) to (336.center);
		\draw (380.center) to (381.center);
		\draw (414.center) to (415.center);
		\draw (418.center) to (419.center);
		\draw (421.center) to (422.center);
		\draw (424.center) to (425.center);
		\draw (427.center) to (428.center);
		\draw (430.center) to (431.center);
		\draw (431.center) to (432.center);
		\draw (432.center) to (430.center);
		\draw (431.center) to (433.center);
		\draw (432.center) to (434.center);
		\draw (435.center) to (430.center);
		\draw (438.center) to (437.center);
		\draw (440.center) to (439.center);
		\draw (445.center) to (444.center);
		\draw (447.center) to (446.center);
		\draw (475.center) to (474.center);
		\draw (481.center) to (480.center);
		\draw (490.center) to (489.center);
		\draw (510.center) to (511.center);
		\draw (511.center) to (512.center);
		\draw (512.center) to (510.center);
		\draw (513.center) to (510.center);
		\draw (516.center) to (512.center);
		\draw (511.center) to (515.center);
		\draw (523.center) to (522);
		\draw [style=dashedE] (551.center)
			 to (550.center)
			 to (553.center)
			 to (552.center)
			 to cycle;
	\end{pgfonlayer}
\end{tikzpicture}
}
\end{equation*}

\paragraph{Dual-rail}
The following rewrite rules relate to dual-rail encoding.

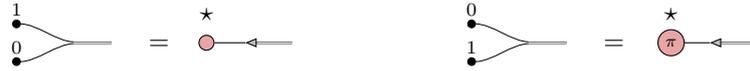
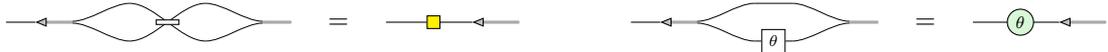
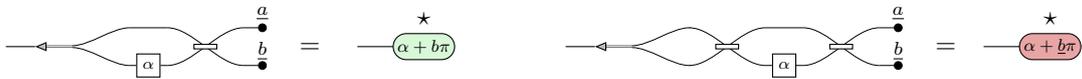
\encodeKet*

\encodeHad*

\encodeMeas*

\paragraph{Channel}
The following rewrite rules relate to the $\mathtt{Channel}$ construction.

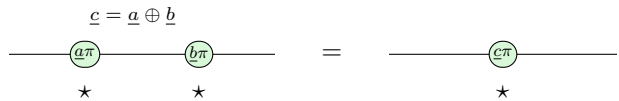
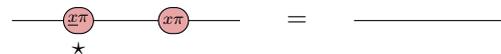
\corseGraining*

\feedForward*

\paragraph{Stream}
The following rewrite rules relate to the $\mathtt{Stream}$ construction.

\unrollingStream*

\unrollingFuture*

\followedBy*

\slidingStream*

\steamAddition*

\section{Kraus decomposition of Type-I fusion}\label{sec:mixedfusion}

To prove \cref{prop:type-I}, we derive diagrams in \textbf{QPath} for each measurement outcome of the linear optical circuit implementing Type I fusion:
\begin{equation*}
    D^{\cvar{a},\cvar{b}}
    \quad\coloneqq\quad
    \input{RUS/rus0b.tikz}
\end{equation*}
Because the input to $D^{\cvar{a},\cvar{b}}$ is two dual rail qubits and hence at most two photons, and no photons are created in this process, we can restrict our attention to only measurement outcomes observing at most two photons.

From the definitions of the dual-rail encoding \cref{eq:dual-rail-encoding}, we can define the following decomposition:
\begin{equation}
  \label{eq:dual-rail-decomposition}
  \input{RUS/dualrail-decomp.tikz}
\end{equation}
Using this, we can now prove a set of lemmas for the different cases of $D^{\cvar{a},\cvar{b}}$.
First, $D^{\cvar{a}=0,\cvar{b}=0}$ evaluates to:
\begin{equation}
    \input{RUS/d00.tikz}
\end{equation}

\noindent Continuing on, we compute $D^{\cvar{a}=1,\cvar{b}=0}$:
\begin{equation}
    \input{RUS/d10.tikz}
\end{equation}
Similarly, $D^{\cvar{a}=0,\cvar{b}=1}$ computes to:
\begin{align}
    \input{RUS/d01-0.tikz}\quad
    &\input{RUS/d01-1.tikz}\\
    &\input{RUS/d01-2.tikz}\\
    &\input{RUS/d01-3.tikz}
\end{align}
We can show that $D^{\cvar{a}=1,\cvar{b}=1}$ evaluates to:
\begin{equation}
    \input{RUS/d11.tikz}
\end{equation}
This means that the probability of observing $D^{\cvar{a}=1,\cvar{b}=1}$ is zero which is due to the Hong-Ou-Mandel effect.

For the last two cases of $D^{\cvar{a},\cvar{b}}$, we use the following equality:
\begin{equation}
  \input{RUS/lem2.tikz}
\end{equation}
Proceeding with the calculations, we compute $D^{\cvar{a}=0,\cvar{b}=2}$ as follows:
\begin{equation}
    \input{RUS/d02.tikz}
\end{equation}
and finally, $D^{\cvar{a}=2,\cvar{b}=0}$ evaluates to:
\begin{equation}
    \input{RUS/d20.tikz}
\end{equation}

To realize $D^{\cvar{s},\cvar{k}}$, we determine the action of a classical function from measurement outcomes to the Boolean variable indicators of success ($\cvar{s}$) and correction ($\cvar{k}$):
\begin{center}
    \begin{tabular}{ll|ll} 
        \hline
        \cvar{a} & \cvar{b} & \cvar{s} & \cvar{k}\\
        \hline
        0 & 0 & 0 & 1\\
        1 & 0 & 1 & 0\\
        0 & 1 & 1 & 1\\
        2 & 0 & 0 & 0\\
        0 & 2 & 0 & 0\\
        \hline
    \end{tabular} 
\end{center}
For three of the five possible values of $(\cvar{s},\cvar{k})$, the original measurement outcomes $(\cvar{a},\cvar{b})$ can be identified.
For the case $(\cvar{s}=0,\cvar{k}=1)$, it happens that $D^{\cvar{a}=2,\cvar{b}=0} = D^{\cvar{a}=0,\cvar{b}=2}$ up to a global phase which is thereafter eliminated upon invoking the CPM construction.
Because of this, in mixed-state quantum mechanics we have
\begin{equation}
    D^{\cvar{s}=0,\cvar{k}=1} = D^{\cvar{a}=2,\cvar{b}=0} + D^{\cvar{a}=0,\cvar{b}=2} = 2 \,\, D^{\cvar{a}=2,\cvar{b}=0}
\end{equation}
Therefore, $D^{\cvar{s},\cvar{k}}$ is a non-destructive measurement that is a sum of four terms, one for each possible value of $(\cvar{s},\cvar{k})$.
By performing a mixed sum over all possible measurement outcomes of $D^{\cvar{a},\cvar{b}}$, and quotienting by $\cvar{s}$ and $\cvar{k}$, we obtain the proposition.

\typeOneFusionProp*
\begin{proof}
  \[
    \input{RUS/mixed-fusion-thm.tikz}
  \]
\end{proof}

\section{Proof of characterization}\label{app:fusion}

In section, we prove \cref{thm-characterization}, characterizing all fusion measurements with green failure and Pauli error.

\characterization*

We are interested in fusions with green failure satisfying the following equation:
\begin{equation}
    \input{ZX/PauliFusionCommutation.tikz}
\end{equation}
for some $w,x,y,z \in \{0,1\}$.

First note that in order for the $\cvar{j} \pi$ error to commute to the inputs as a Pauli error, 
$\phi$ must satisfy the commutation relation
\begin{equation}
    \input{ZX/PauliFusionCase34.tikz}
\end{equation}
for some $\cvar{p}, \cvar{q} \in \{0,1\}$. That is, $\phi$ must be a Clifford phase. 
Therefore, necessarily $\phi = v \frac{\pi}{2}$ for some $v \in \{0,1,2,3\}$.
We thus consider two cases: (i) $v$ is odd and (ii) $v$ is even.

If $v$ is even, we can take $\phi = 0$; the $\phi = \pi$ case is redundant because of the random error $\cvar{k} \pi$.
In this case, the $\cvar{j} \pi$ error induces $Z$ errors on both input qubits. 
In order to correct $\cvar{k} \pi$ error, one of $\alpha_1$ or $\alpha_2$ must be Clifford. Without loss of generality, 
we assume $\alpha_1$ is a Clifford phase, i.e. $\alpha_1 = d \frac{\pi}{2}$, and the result follows with $\lambda = \XZm$.
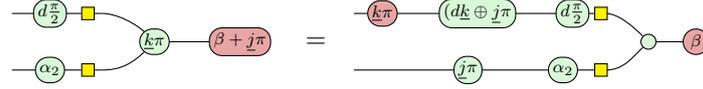
\[ 
  \input{ZX/PauliFusionCase1.tikz}
\]

If $v$ is odd, we can take $\phi = \frac{\pi}{2}$; the $\phi = -\frac{\pi}{2}$ case is redundant because of the random error $\cvar{k}$.
Then, the $\cvar{j} \pi$ error induces a $Y$ error when commuting with $\phi$. Since $Y = XZ$, this induces an $Z$ error 
on both input qubits and an $X$ error which merges with the $\cvar{k} \pi$ error. In order to correct the resulting 
$(\cvar{k} \oplus \cvar{j}) \pi$ error, either $\alpha_1$ or $\alpha_2$ must be a an integer multiple of $\frac{\pi}{2}$. 
Without loss of generality we can set $\alpha_1 = d \frac{\pi}{2}$, and the result follows with $\lambda = \YZm$.
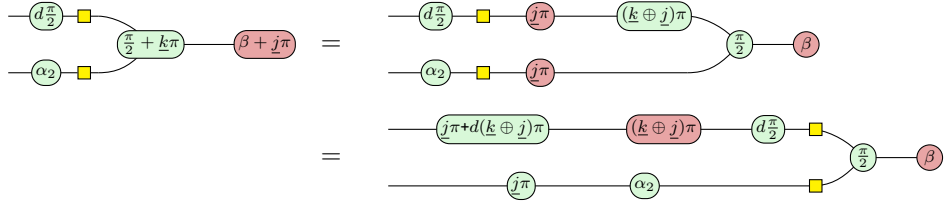
\[ 
  \input{ZX/PauliFusionCase2.tikz}
\]

The third case, where $\lambda = \XYm$ is obtained when $\beta = \frac{\pi}{2} + \alpha$ for some angle $\alpha$ and 
$\phi = \frac{\pi}{2}$. This is technically an instance of the case where $v$ is odd, but we distinguish 
it from the other cases because the errors are different, as they are induced by a single-qubit measurement in the $\XYm$ plane. 
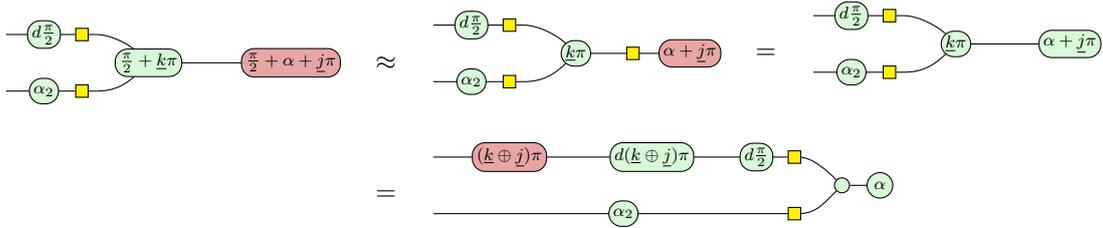
\[ 
  \input{ZX/PauliFusionCase3.tikz}
\]

\section{Probability of success of fusions with green failure}\label{subsec:fusion-prob}

Recall the circuit for fusion measurements with green failure, parametrised by three phases:
\begin{equation}
  \input{fusion-green.tikz}
\end{equation}
Recall from \cref{prop:type-I}, that the behaviour of Type I fusion measurements on dual-rail encoded qubits can be described by a sum of ZX diagrams.
A similar equation holds for the circuit defined above.
\begin{equation}\label{fusion-green-decomp}
  \input{fusion-prob-1.tikz}
\end{equation}
where $\cvar s = \cvar a \oplus \cvar b$ and $\cvar k = \cvar s \cvar b + \neg \cvar s (1 - \frac{\cvar a + \cvar b}{2})$.
The two diagrams above represent the action of fusion in case of success and failure, respectively.
The \emph{probability of success} for an input state $\Psi$ is obtained by taking the \emph{trace} of the success diagram, 
and discarding the classical output $\cvar{k}$, which corresponds to summing over its possible values.
\begin{equation*}
  \Pr(\cvar{s}=1\,\vert\Psi) = \; \input{fusion-prob-2.tikz} \; = \; \input{fusion-prob-2-1.tikz} \; = \; \input{fusion-prob-2-2.tikz}
\end{equation*}
Note that the probability will usually depend on the input state $\Psi$.
For example, we may engineer input states for which the fusion \enquote{always succeeds}.
\[
  \input{fusion-prob-fail.tikz}
\]
The calculation above uses dotted wires to represent ZX diagrams in the standard (rather than the CP) interpretation.
Replacing the first input with a green $\theta_1 + \pi$ spider we obtain an example where the fusion \enquote{always fails}.
If the input state is the completely mixed state on two qubits, the probability is found to be $\frac{1}{2}$ by the following derivation.
\[
  \input{fusion-prob-fully-mixed.tikz}
\]

By the purification theorem, a general mixed state $\Psi$ can be expressed in terms of a pure state $\psi$ on a larger space:
\[
  \input{purification.tikz}
\]
The dependence of the probability of success on the input can be avoided if we assume that the inputs are unmeasured qubits in a graph state.
Then the state $\psi$ has the following form:
\[
  \input{fusion-prob-3.tikz}
\]
where $\ket{G}$ is a graph state.
\begin{proposition}\label{prop:green2qhalf}
  When the inputs are unmeasured qubits in a graph state $\ket{G}$, 
  the success probability of any fusion with green failure is $\frac{1}{2}$.
\end{proposition}
\begin{proof}
  \begin{align*}
    \Pr(\cvar{s}=1\,|\Psi)\quad
    &\input{fusion-prob-graph-state.tikz}\\
    &\input{RUS/fusion-prob-4.tikz}
  \end{align*}
\end{proof}

The action of fusion on unmeasured qubits in a graph state may be seen as a non-demolition measurement, 
obtained by composing the fusion operation with Z spiders as above.
It is useful to write this operation as a classical probability distribution over causal maps, as follows.

\section{Proof of flow preserving rewrites}\label{app:flow}

In this appendix, we prove the results about preservation of Pauli flow of \cref{sec:flow}.
We use an alternative notation to simplify the diagrams, and replace a Hadamard
between two spiders by a blue dashed edge, as illustrated below.
\[
    \begin{tikzpicture}
	\begin{pgfonlayer}{nodelayer}
		\node [style=gn] (1) at (-5, 0) {};
		\node [style=hbox] (2) at (-3.75, 0) {};
		\node [style=none] (5) at (-6, 0.75) {};
		\node [style=none] (6) at (-6, -0.75) {};
		\node [style=none] (7) at (-6.25, 0.75) {};
		\node [style=none] (8) at (-6.25, -0.75) {};
		\node [style=none] (9) at (-6, 0.25) {$\vdots$};
		\node [style=gn] (10) at (-2.5, 0) {};
		\node [style=none] (11) at (-1.5, 0.75) {};
		\node [style=none] (12) at (-1.5, -0.75) {};
		\node [style=none] (13) at (-1.25, 0.75) {};
		\node [style=none] (14) at (-1.25, -0.75) {};
		\node [style=none] (15) at (-1.5, 0.25) {$\vdots$};
		\node [style=none] (16) at (0, 0) {$\coloneqq$};
		\node [style=gn] (17) at (2.5, 0) {};
		\node [style=none] (19) at (1.5, 0.75) {};
		\node [style=none] (20) at (1.5, -0.75) {};
		\node [style=none] (21) at (1.25, 0.75) {};
		\node [style=none] (22) at (1.25, -0.75) {};
		\node [style=none] (23) at (1.5, 0.25) {$\vdots$};
		\node [style=gn] (24) at (5, 0) {};
		\node [style=none] (25) at (6, 0.75) {};
		\node [style=none] (26) at (6, -0.75) {};
		\node [style=none] (27) at (6.25, 0.75) {};
		\node [style=none] (28) at (6.25, -0.75) {};
		\node [style=none] (29) at (6, 0.25) {$\vdots$};
	\end{pgfonlayer}
	\begin{pgfonlayer}{edgelayer}
		\draw (2) to (1);
		\draw [in=-120, out=0, looseness=0.75] (6.center) to (1);
		\draw [in=0, out=120, looseness=0.75] (1) to (5.center);
		\draw (7.center) to (5.center);
		\draw (8.center) to (6.center);
		\draw [in=-60, out=180, looseness=0.75] (12.center) to (10);
		\draw [in=180, out=60, looseness=0.75] (10) to (11.center);
		\draw (13.center) to (11.center);
		\draw (14.center) to (12.center);
		\draw (10) to (2);
		\draw [in=-120, out=0, looseness=0.75] (20.center) to (17);
		\draw [in=0, out=120, looseness=0.75] (17) to (19.center);
		\draw (21.center) to (19.center);
		\draw (22.center) to (20.center);
		\draw [in=-60, out=180, looseness=0.75] (26.center) to (24);
		\draw [in=180, out=60, looseness=0.75] (24) to (25.center);
		\draw (27.center) to (25.center);
		\draw (28.center) to (26.center);
		\draw [style=hadamard edge] (24) to (17);
	\end{pgfonlayer}
\end{tikzpicture}

\]
Both the blue edge notation and the Hadamard box can always be translated back
into spiders when necessary.
We refer to the blue edge as a Hadamard edge.

\begin{lemma}[Copy]
  \label{lem:copy}
  Copying preserves the existence of Pauli flow~\cite[Lemma D.6]{simmonsRelatingMeasurementPatterns2021}.
  Graphically, this corresponds to the copy rule of the ZX calculus~\cite[Lemmas 2.7 and 2.8 ]{mcelvanneyFlowpreservingZXcalculusRewrite2023}:
  \[
    \input{fusion/CopyFlowPreserving.tikz}
  \]
\end{lemma}

\begin{lemma}[Local Complementation]
  \label{lem:lc}
  Local complementation about a vertex $u$ preserves the existence of Pauli flow~\cite[Lemma D.12]{simmonsRelatingMeasurementPatterns2021}.
  In the ZX calculus, this rule is also called local complementation, and is given as follows~\cite[Lemma 2.10]{mcelvanneyFlowpreservingZXcalculusRewrite2023}:
  \[
    \input{fusion/LCFlowPreserving.tikz}
  \]
\end{lemma}

\begin{lemma}[Pivot]
  \label{lem:pivot}
  Pivoting about an edge $(u, v)$ preserves the existence of Pauli flow~\cite[Lemma D.21]{simmonsRelatingMeasurementPatterns2021}.
  In the ZX calculus, this rule is also called pivoting, and is given as the following rewrite rule~\cite[Lemma 2.11]{mcelvanneyFlowpreservingZXcalculusRewrite2023}:
  \[
    \input{fusion/PivotFlowPreserving.tikz}
  \]
\end{lemma}

\begin{lemma}[State Change]
  \label{lem:state-change}
  \[
    \normalfont
    \input{ZX/state-change.tikz}
  \]
\end{lemma}

\xFusionFlowPreserving*
\begin{proof}
  \[
    \input{fusion/XFusionFlowPreservingProof.tikz}
  \]
  Since each of the rewrites preserves the existence of Pauli flow, the additional errors appearing above can always be corrected, but we have given them here for completeness.
\end{proof}

\yFusionFlowPreserving*
\begin{proof}
  \[
    \input{fusion/YFusionFlowPreservingProof.tikz}
  \]
\end{proof}

\section{Proofs of repeat-until-success}\label{app:RUS}

We prove the results of \cref{subsec:RUS}. To this end, we will need the following lemmas.

\begin{lemma}\label{prop:fusion-stochastic}
  Let $F_\theta$ be the optical circuit defined above and $F^\alpha_\theta$ the same optical circuit followed by a 
  measurement in the $\XZm$ plane of angle $\alpha$. Then we have:
  \[
    \input{fusion-stochastic.tikz}
  \]
  \[
    \input{fusion-stochastic2.tikz}
  \]
\end{lemma}
\begin{proof}
  This follows from \cref{prop:green-failure} and $\interp{\star}_{CP} = \frac{1}{2}$.
\end{proof}

\begin{lemma}
  \[
    \input{lemma-measure-emitter.tikz}
  \]
\end{lemma}
\begin{proof}
  We prove this by induction. The statement for $n=0$ is easy to show. Then using \cref{def:unrolling}, we have
  \[
    \scalebox{0.8}{\input{lemma-measure-proof-2.tikz}}
  \]
  Finally, the last equality holds after coarse-graining the measurement outcomes $a_t$. Indeed there are $2^n$ possible measurement outcomes for the list $a_t$, 
  but they together induce one uniformly random bitflip $c_n \pi$ after $n$ time-steps. The statement then follows from \cref{eq:coarse-graining-rewrite}.
\end{proof}
\RUS*
\begin{proof}
  We prove this inductively.
  The $n = 1$ case follows from \cref{prop:fusion-stochastic}.
  Then, given the statement for $n$, we show it for $n + 1$:
  \allowdisplaybreaks
  \begin{align*}
    &\scalebox{.8}{\input{repeat-prooff-1.tikz}}\\
    &\scalebox{.8}{\input{repeat-prooff-2.tikz}}\\
    &\scalebox{.8}{\input{repeat-prooff-3.tikz}}\\
    &\scalebox{.8}{\input{repeat-prooff-4.tikz}}\\
    &\scalebox{.8}{\input{repeat-prooff-5.tikz}}\\
    &\scalebox{.8}{\input{repeat-prooff-5+.tikz}}\\
    &\scalebox{.8}{\input{repeat-prooff-6+.tikz}}
  \end{align*}
\end{proof}

\YRUS*
\begin{proof}
  Follows from \cref{thm:RUS} and:
  \[
    \scalebox{.8}{\input{repeat-Y-proof.tikz}}
  \]
  where $x_T = k_T \oplus j_T$.
\end{proof}

\XBoosted*
\begin{proof}
  \begin{align*}
    &\scalebox{.8}{\input{boosted-X-proof-0.tikz}}\\
    &\scalebox{.8}{\input{boosted-X-proof.tikz}}
  \end{align*}
\end{proof}

\end{document}